\documentclass[a4paper, 10pt]{article}

\usepackage{graphicx}
\usepackage{color}  
\usepackage{algpseudocode, algorithm}
\usepackage{amssymb}
\usepackage{amsmath}
\usepackage{pgfplots}
\usepackage{times}
\usepackage{fullpage}
\usepackage{url}
\usepackage{authblk}
\usepackage{appendix}
\usepackage{xspace}
\usepackage{hyperref}
\usepackage{cleveref}
\usepackage{subfiles}
\usepackage{subfig}

\usepackage{layout}
\usepackage[english]{babel}
\usepackage[utf8]{inputenc}
\usepackage{fancyhdr}
\usepackage{tabularx}
\usepackage{ragged2e}
\usepackage{hyperref}
\usepackage{cleveref}

\usepackage{fancyvrb}

\usepackage{epstopdf}
\usepackage{epsfig}
\usepackage{etoolbox}
\usepackage{breqn}
\usepackage{relsize}

\usepackage{pgfplots}
\usepackage{subfiles}

\usepackage{authblk}
\usepackage{pdfpages}
\pgfplotsset{compat=1.7}
\usepackage{calligra}
\usepackage{calrsfs}

\usepackage{subfig}
\usepackage{todonotes}

\usepackage{xcolor}
\usepackage[justification=centering]{caption}

\usepackage{import}

\newcommand{\punt}[1]{}
\newcommand{\cmnt}[1]{}
\newcommand{\ignore}[1]{}

\newtheorem{theorem}{Theorem}
\newtheorem{lemma}[theorem]{Lemma}
\newtheorem{corollary}[theorem]{Corollary}

\newtheorem{property}[theorem]{Property}
\newtheorem{observation}[theorem]{Observation}
\newtheorem{result}[theorem]{Result}

\newtheorem{definition}{Definition}

\newcounter{history}

\newtheorem{assumption}{Assumption}
 
\newenvironment{proof}[1][Proof]{\noindent\textbf{#1.} }{} 

\newcommand{\secref}[1]{Section~\ref{sec:#1}}
\newcommand{\figref}[1]{Fig~\ref{fig:#1}}

\newcommand{\stref}[1]{Step~\ref{step:#1}}
\newcommand{\csref}[1]{Case~\ref{case:#1}}
\newcommand{\thmref}[1]{Theorem~\ref{thm:#1}}
\newcommand{\lemref}[1]{Lemma~\ref{lem:#1}}
\newcommand{\corref}[1]{Corollary~\ref{cor:#1}}

\newcommand{\eqnref}[1]{Eq.(\ref{eq:#1})}

\newcommand{\propref}[1]{Property~\ref{prop:#1}}
\newcommand{\obsref}[1]{Observation~\ref{obs:#1}}
\newcommand{\asmref}[1]{Assumption~\ref{asm:#1}}
\newcommand{\resref}[1]{Result~\ref{res:#1}}

\newcommand{\Lineref}[1]{Line~\ref{lin:#1}}
\newcommand{\algoref}[1]{Algorithm~\ref{algo:#1}}
\newcommand{\subsecref}[1]{SubSection~\ref{subsec:#1}}

\newcommand{\apnref}[1]{Appendix~\ref{apn:#1}}

\newcommand{\linereff}[1]{Line~\ref{lin:#1}}

\newcommand*{\affaddr}[1]{#1} 
\newcommand*{\affmark}[1][*]{\textsuperscript{#1}}

\newcommand{\Wset}{\textit{Wset}}

\newcommand{\tobj} {t-object\xspace}

\newcommand{\txns}[1] {#1.txns}
\newcommand {\comm}[1] {#1.committed}
\newcommand {\aborted}[1] {#1.aborted}
\newcommand {\incomp}[1] {#1.incomp}
\newcommand {\live}[1] {#1.live}

\newcommand{\commit}{\mathcal{C}}
\newcommand{\abort}{\mathcal{A}}

\newcommand{\shist}[2] {#2.subhist(#1)\xspace}
\newcommand{\subhist} {subhist\xspace}

\newcommand{\shset}[1] {#1.subhistSet\xspace}

\newcommand{\op} {operation\xspace}
\newcommand{\mth} {method\xspace}

\newcommand{\termop} {terminal operation\xspace}
\newcommand{\term} {term\text{-}op\xspace}
\newcommand{\termed}[1] {#1.terminated\xspace}

\newcommand{\cc} {correctness-criterion\xspace}
\newcommand{\ccs} {correctness-criteria\xspace}
\newcommand{\gen}[1] {gen(#1)}

\newcommand{\evts}[1] {#1.evts}
\newcommand{\ssch} {sub-history\xspace}

\newcommand{\tseq} {t-sequential\xspace}

\newcommand{\rt} {real-time\xspace}

\newcommand{\stfdm} {starvation\text{-}freedom\xspace}
\newcommand{\stf} {starvation\text{-}free\xspace}

\newcommand{\cts} {CTS\xspace}
\newcommand{\its} {ITS\xspace}
\newcommand{\wts} {WTS\xspace}
\newcommand{\ltl} {tltl\xspace}
\newcommand{\utl} {tutl\xspace}
\newcommand{\lock} {lock\xspace}
\newcommand{\val} {valid\xspace}
\newcommand{\stat} {state\xspace}

\newcommand{\gtcnt} {G\_tCntr\xspace}

\newcommand{\tcntr} {tCntr\xspace}

\newcommand{\incv} {incrVal\xspace}

\newcommand{\glock} {G\_lock\xspace}
\newcommand{\gval} {G\_valid\xspace}
\newcommand{\gstat} {G\_state\xspace}
\newcommand{\gits} {G\_its\xspace}
\newcommand{\gcts} {G\_cts\xspace}
\newcommand{\gwts} {G\_wts\xspace}
\newcommand{\tltl} {G\_tltl\xspace}
\newcommand{\tutl} {G\_tutl\xspace}

\newcommand{\htlock}[2] {#2.lock_#1\xspace}
\newcommand{\htval}[2] {#2.valid_#1\xspace}
\newcommand{\htstat}[2] {#2.state_#1\xspace}
\newcommand{\htits}[2] {#2.its_#1\xspace}
\newcommand{\htcts}[2] {#2.cts_#1\xspace}
\newcommand{\htwts}[2] {#2.wts_#1\xspace}
\newcommand{\htltl}[2] {#2.tltl_#1\xspace}
\newcommand{\htutl}[2] {#2.tutl_#1\xspace}

\newcommand{\getl} {getLar}
\newcommand{\getsm} {getSm}

\newcommand{\tcts}[1] {cts_#1\xspace}
\newcommand{\tits}[1] {its_#1\xspace}
\newcommand{\twts}[1] {wts_#1\xspace}
\newcommand{\ttltl}[1] {tltl_#1\xspace}
\newcommand{\ttutl}[1] {tutl_#1\xspace}
\newcommand{\tlock}[1] {lock_#1\xspace}
\newcommand{\tval}[1] {valid_#1\xspace}
\newcommand{\tstat}[1] {state_#1\xspace}

\newcommand{\syst} {sys\text{-}time\xspace}
\newcommand{\hsyst}[1] {#1.sys\text{-}time\xspace}

\newcommand {\rab} {relAbort\xspace}

\newcommand{\lastw} {lastWrite}
\newcommand{\lwrite}[2] {#2.lastWrite(#1)}

\newcommand{\valid} {valid\xspace}

\newcommand{\legal} {legal\xspace}

\newcommand{\gc} {garbage-collection\xspace}
\newcommand{\mvstm} {MVSTM\xspace}

\newcommand{\rwph} {read/local-write phase\xspace}
\newcommand{\tryph} {try-Commit phase\xspace}

\newcommand{\focc} {FOCC\xspace}
\newcommand{\mvto} {MVTO\xspace}
\newcommand{\sfmv} {SFMVTO\xspace}
\newcommand{\sfkv} {SFKTO\xspace}

\newcommand{\sftm} {\textit{SV\text{-}SFTM}\xspace}
\newcommand{\svsftm} {\textit{SV-SFTM}\xspace}
\newcommand{\mvsftm} {\textit{UVSFTM}\xspace}
\newcommand{\mvsftmgc} {\textit{UVSFTM-GC}\xspace}
\newcommand{\ksftm} {\textit{KSFTM}\xspace}
\newcommand{\pmvto} {\textit{PMVTO}\xspace}
\newcommand{\pmvtogc} {\textit{PMVTO-GC}\xspace}
\newcommand{\pkto} {\textit{PKTO}\xspace}

\newcommand{\vlist} {vlist\xspace}
\newcommand{\rlist} {read\text{-}list\xspace}

\newcommand{\init} {init\xspace}
\newcommand{\begt} {stm\text{-}begin\xspace}
\newcommand{\tread} {stm\text{-}read\xspace}
\newcommand{\twrite} {stm\text{-}write\xspace}
\newcommand{\tryc} {stm\text{-}tryC\xspace}
\newcommand{\trya} {stm\text{-}tryA\xspace}

\newcommand{\findls} {findLTS\xspace}
\newcommand{\findsl} {findSTL\xspace}

\newcommand{\isab} {isAborted\xspace}
\newcommand{\lowp} {lowPriAbt\xspace}

\newcommand{\relll} {relLL\xspace}

\newcommand{\allrl} {allRL\xspace}
\newcommand{\srl} {smallRL\xspace}
\newcommand{\lrl} {largeRL\xspace}
\newcommand{\pvl} {prevVL\xspace}
\newcommand{\nvl} {nextVL\xspace}

\newcommand{\prevv} {prevVer\xspace}
\newcommand{\nextv} {nextVer\xspace}
\newcommand{\abl} {abortRL\xspace}

\newcommand{\ct} {comTime\xspace}

\newcommand{\livel} {\emph{live-list}\xspace}

\newcommand{\rs}{rset\xspace}
\newcommand{\ws}{wset\xspace}
\newcommand{\rset}[1] {rset_{#1}}
\newcommand{\wset}[1] {wset_{#1}}

\newcommand{\vl} {\texttt{vl}\xspace}
\newcommand{\ts} {\texttt{ts}\xspace}
\newcommand{\rl} {\texttt{rl}\xspace}

\newcommand{\vltl} {\texttt{vrt}\xspace}
\newcommand{\vutl} {\texttt{vutl}\xspace}
\newcommand{\vt} {\texttt{vrt}\xspace}

\newcommand{\vtup} {vTuple\xspace}
\newcommand{\dtup} {dTuple\xspace}

\newcommand{\opq} {opaque\xspace}
\newcommand{\opty} {opacity\xspace}

\newcommand{\slbty} {serializability\xspace}
\newcommand{\stsbty} {strict\text{-}serializability\xspace}

\newcommand{\stsble} {strict\text{-}serializable\xspace}
\newcommand{\lble} {linearizable\xspace}

\newcommand{\lo} {LO\xspace}
\newcommand{\lopty} {local opacity\xspace}
\newcommand{\lopq} {locally\text{-}opaque\xspace}

\newcommand {\wct} {max\text{-}time\xspace}
\newcommand {\inc} {incarnation\xspace}
\newcommand {\incn} {incNum\xspace}
\newcommand {\inum}[1] {T_#1.incNum\xspace}
\newcommand {\ninc} {nextInc\xspace}
\newcommand {\nexti}[1] {T_#1.nextInc\xspace}
\newcommand {\incs}[2] {#2.incarSet(T_#1)\xspace}
\newcommand {\incset} {incarSet\xspace}
\newcommand {\inct}[2] {#2.incarCt(T_#1)\xspace}
\newcommand {\incct} {incarCt\xspace}
\newcommand {\aptr} {application-transaction\xspace}

\newcommand {\cdset} {cds\xspace}

\newcommand {\hcds}[2] {#2.cds(T_#1)\xspace}

\newcommand {\haffset}[2] {#2.affectSet(T_#1)\xspace}
\newcommand {\affset} {affectSet\xspace}
\newcommand {\hmaxwts}[2] {#2.maxWTS(T_#1)\xspace}
\newcommand {\maxwts} {maxWTS\xspace}
\newcommand {\haffwts}[2] {#2.affWTS(T_#1)\xspace}
\newcommand {\affwts} {affWTS\xspace}

\newcommand {\cis} {cis\xspace}
\newcommand {\hcis}[2] {#2.cis(T_#1)\xspace}
\newcommand {\depits} {depIts\xspace}
\newcommand {\hdep}[2] {#2.depIts(T_#1)\xspace}
\newcommand {\pawts} {partAffWTS\xspace}
\newcommand {\hpawts}[2] {#2.partAffWTS(T_#1)\xspace}

\newcommand {\itsen} {itsEnabled\xspace}
\newcommand {\itsenb}[2] {#2.itsEnabled(T_#1)\xspace}
\newcommand {\cdsen} {cdsEnabled\xspace}
\newcommand {\cdsenb}[2] {#2.cdsEnabled(T_#1)\xspace}
\newcommand {\finen} {finEnabled\xspace}
\newcommand {\finenb}[2] {#2.finEnabled(T_#1)\xspace}
\newcommand {\enbd} {finEnabled\xspace}

\newcommand{\opg}[2] {OPG(#1, #2)}

\newcommand{\mv} {mv}
\newcommand{\rf} {rf}
\newcommand{\rtx} {rt}

\algrenewcommand{\algorithmiccomment}[1]{/* #1 */}

\newcommand{\schdr}{scheduler\xspace}
\newcommand{\bdtm}{bounded\text{-}termination\xspace}

\cmnt {
\author{
Ved Prakash Chaudary \\
cs14mtech11019@iith.ac.in \\
CSE Department \\
IIT Hyderabad \\
India
\and 
Raj Kripal Danday \\
es12b1007@iith.ac.in \\
CSE Department \\
IIT Hyderabad \\
India
\and 
Hima Varsha Dureddy \\
cs12b1042@iith.ac.in \\
CSE Department \\
IIT Hyderabad \\
India
\and 
Sweta Kumari \\
cs15resch01004@iith.ac.in \\
CSE Department \\
IIT Hyderabad \\
India
\and 
Sathya Peri \\
sathya\_p@iith.ac.in \\
CSE Department \\
IIT Hyderabad \\
India  
}
}

\title{Achieving Starvation-Freedom in Multi-Version Transactional Memory Systems \thanks{A preliminary version of this work was accepted in AADDA 2017 as \textbf{work in progress}. }} 

\author[1]{Ved Prakash Chaudhary \footnote{A part of this work was submitted towards the fulfillment of M.Tech thesis requirement by the author.}}
\author[1]{Chirag Juyal} 
\author[2]{Sandeep Kulkarni}
\author[1]{Sweta Kumari}
\author[1]{Sathya Peri\footnote{Author sequence follows a lexical order of last names.}}
\affil[1]{Department of Computer Science \& Engineering, IIT Hyderabad, India \\
	\texttt{\{cs14mtech11019, cs17mtech11014, cs15resch01004, sathya\_p\}@iith.ac.in}}
\affil[2]{Department of Computer Science, Michigan State University, USA\\
	\texttt{sandeep@cse.msu.edu}}

\cmnt {
\author{
	Ved Prakash Chaudhary \affmark[1], Raj Kripal Danday\affmark[1], Hima Varsha Dureddy\affmark[1], Sandeep Kulkarni\affmark[2], Sweta Kumari\affmark[1], Sathya Peri \affmark[1]\\
	\email{\affmark[1]\{cs14mtech11019,es12b1007,cs12b1042\}@iith.ac.in}, \email{\affmark[2]sandeep@cse.msu.edu}, \email{\affmark[1]\{cs15resch01004,sathya\_p\}@iith.ac.in}\\
	\affaddr{\affmark[1]Department of Computer Science \& Engineering, IIT Hyderabad}\\
	\affaddr{\affmark[2]Department of Computer Science, Michigan State University}\\
}

\author[1]{Ved Prakash Chaudhary \thanks{cs14mtech11019@iith.ac.in}}
\author[1]{Raj Kripal Danday \thanks{es12b1007@iith.ac.in}}
\author[1]{Hima Varsha Dureddy \thanks{cs12b1042@iith.ac.in}}
\author[2]{Sandeep Kulkarni \thanks{sandeep@cse.msu.edu}}
\author[1]{Sweta Kumari \thanks{cs15resch01004@iith.ac.in}}
\author[1]{Sathya Peri \thanks{sathya\_p@iith.ac.in}}
\affil[1]{Department of Computer Science \& Engineering, IIT Hyderabad}
\affil[2]{Department of Computer Science, Michigan State University}
}

\date{}

\begin{document}

\maketitle              
\thispagestyle{empty}

\cmnt {
\begingroup
\centering
{\vspace{-1.6cm} \large Ved Prakash Chaudhary, Raj Kripal Danday, Hima Varsha Dureddy, Sweta Kumari, Sathya Peri} \\[0.3em]
\url{(cs14mtech11019, es12b1007, cs12b1042, cs15resch01004, sathya_p)@iith.ac.in}
\\[0.2em]
\hspace{4.4cm} {\large CSE Department, IIT Hyderabad, India}
\endgroup
}

\begin{abstract}
Software Transactional Memory systems (STMs) have garnered significant interest as an elegant alternative for addressing synchronization and concurrency issues with multi-threaded programming in multi-core systems. For STMs to be efficient, they must guarantee some progress properties. This work explores the notion of one of the progress property, i.e., \emph{\stfdm}, in STMs. An STM system is said to be \stf if every thread invoking a transaction gets the opportunity to take a step (due to the presence of a fair scheduler) such that the transaction eventually commits.

A few \emph{\stf} algorithms have been proposed in the literature in context of single-version STMs. These algorithms are priority based i.e. if two transactions are in conflict, then the transaction with lower priority will abort. A transaction running for a long time will eventually have the highest priority and hence commit. But the drawback with this approach is that if a set of high-priority transactions become slow, then they can cause several other transactions to abort.\cmnt{ In that case, this approach becomes similar to pessimistic lock-based approach.} So, we propose multi-version \emph{\stf} STM system which addresses this issue.

Multi-version STMs maintain multiple-versions for each transactional object. By storing multiple versions, these systems can achieve greater concurrency. In this paper, we propose multi-version \emph{\stf} STM, \ksftm, which as the name suggests achieves \stfdm while storing $K$-$versions$ of each \tobj. Here $K$ is an input parameter fixed by the application programmer depending on the requirement. Our algorithm is dynamic which can support different values of $K$ ranging from one to infinity. If $K$ is infinite, then there is no limit on the number of versions. But a separate garbage-collection mechanism is required to collect unwanted versions. On the other hand, when $K$ is one, it becomes the same as a single-version \emph{\stf} STM system. We prove the correctness and \emph{\stfdm} property of the \ksftm algorithm. 

To the best of our knowledge, this is the first multi-version STM system that satisfies \emph{\stfdm}. We implement \ksftm and compare its performance with single-version \emph{\stf} STM system (\sftm) which works on the priority principle. Our experiments show that \ksftm gives an average speedup on the worst-case time to commit of a transaction by a factor of 1.22, 1.89, 23.26 and 13.12 times over \pkto, \svsftm, NOrec STM and ESTM respectively for counter application. \ksftm performs 1.5 and 1.44 times better than \pkto and \svsftm but 1.09 times worse than NOrec for low contention KMEANS application of STAMP benchmark whereas \ksftm performs 1.14, 1.4 and 2.63 times better than \pkto, \svsftm and NOrec for LABYRINTH application of STAMP benchmark which has high contention with long-running transactions.
\end{abstract}

\cmnt{
}

\section{Introduction}
\label{sec:intro}

 STMs \cite{HerlMoss:1993:SigArch,ShavTou:1995:PODC} are a convenient programming interface for a programmer to access shared memory without worrying about consistency issues. STMs often use an optimistic approach for concurrent execution of \emph{transactions} (a piece of code invoked by a thread). In optimistic execution, each transaction reads from the shared memory, but all write updates are performed on local memory. On completion, the STM system \textit{validates} the reads and writes of the transaction. If any inconsistency is found, the transaction is \emph{aborted}, and its local writes are discarded. Otherwise, the transaction is committed, and its local writes are transferred to the shared memory. A transaction that has begun but has not yet committed/aborted is referred to as \emph{live}.

A typical STM is a library which exports the following methods: \emph{\begt} which begins a transaction, \textit{\tread} which reads a \emph{transactional object} or \emph{\tobj}, \textit{\twrite} which writes to a \emph{\tobj}, \textit{\tryc} which tries to commit the transaction. Typical code for using STMs is as shown in \algoref{sfdemo} which shows how an insert of a concurrent linked-list library is implemented using STMs. 

\ignore{
\color{blue}
\color{red}
A typical code using STMs is as shown in \algoref{sfdemo}. It shows the overview of a concurrent \emph{insert} \mth which inserts an element $e$ into a linked-list $LL$. It consists of a loop where the thread creates a transaction. This transaction executes the code to insert an element $e$ in a linked-list $LL$ using $\tread$ and $\twrite$ operations. (The result of $\twrite$ operation are stored locally.) At the end of the transaction, the thread calls \textit{\tryc}. At this point, the STM checks if the given transaction can be committed while satisfying the required safety properties (e.g., \slbty \cite{Papad:1979:JACM}, \opty \cite{GuerKap:2008:PPoPP}). If yes, then the transaction is committed. At this time, any updates done by the transaction are reflected in the shared memory. Otherwise, it is aborted. In this case, all the updates made by the transaction are discarded. If the given transaction is aborted, then the invoking thread may retry that transaction again like \linereff{retry} in \algoref{sfdemo}. 
\color{black}
}

\noindent
\textbf{Correctness:} 
\ignore{
\color{red}
By committing/aborting the transactions, the STM system ensures atomicity and consistency of transactions. Thus, an important requirement of STMs is to precisely identify the criterion as to when a transaction should be \emph{aborted/committed}, referred to as \emph{\cc}. \color{black}
}
Several \emph{\ccs} have been proposed for STMs such as opacity \cite{GuerKap:2008:PPoPP}, local opacity \cite{KuzSat:NI:ICDCN:2014,KuzSat:NI:TCS:2016}. All these \emph{\ccs} require that all the transactions including aborted ones appear to execute sequentially in an order that agrees with the order of non-overlapping transactions. Unlike the \ccs for traditional databases, such as serializability,  strict-serializability \cite{Papad:1979:JACM}, the \ccs for STMs ensure that even aborted transactions read correct values. This ensures that programmers do not see any undesirable side-effects due to the reads by transaction that get aborted later such as divide-by-zero, infinite-loops, crashes etc. in the application due to concurrent executions. This additional requirement on aborted transactions is a fundamental requirement of STMs which differentiates STMs from databases as observed by Guerraoui \& Kapalka \cite{GuerKap:2008:PPoPP}. Thus in this paper, we focus on optimistic executions with the \emph{\cc} being \emph{\lopty} \cite{KuzSat:NI:TCS:2016}. 

\ignore{
\todo{This para could be dropped. We have already said optimistic execution before}
To ensure correctness such as \opty, most STMs execute optimistically. With this approach, a transaction only reads values written by other committed transactions. \textbf{To achieve this, all writes are written to local memory first. They are added to the shared memory only when the transaction commits.} This (combined with required validation) can, in turn, ensure that the reads of the transaction are consistent as required by the \emph{\cc}. Thus in this paper, we focus only on optimistic executions with the \emph{\cc} being \emph{\lopty} \cite{KuzSat:NI:TCS:2016} (explained in \secref{model}).

\color{black}
}

\vspace{1mm}
\noindent
\textbf{Starvation Freedom:} In the execution shown in \algoref{sfdemo}, there is a possibility that the transaction which a thread tries to execute gets aborted again and again. Every time, it executes the transaction, say $T_i$, $T_i$ conflicts with some other transaction and hence gets aborted. In other words, the thread is effectively starving because it is not able to commit $T_i$ successfully. 

A well known blocking progress condition associated with concurrent programming is \stfdm \cite[chap 2]{HerlihyShavit:AMP:Book:2012}, \cite{HerlihyShavit:Progress:Opodis:2011}. In the context of STMs, \stfdm ensures that every aborted transaction that is retried infinitely often eventually commits. It can be defined as: an STM system is said to be \emph{\stf} if a thread invoking a transaction $T_i$ gets the opportunity to retry $T_i$ on every abort (due to the presence of a fair underlying scheduler with bounded termination) and $T_i$ is not \emph{parasitic}, i.e., $T_i$ will try to commit given a chance then $T_i$ will eventually commit. Parasitic transactions \cite{Bushkov+:Live-TM:PODC:2012} will not commit even when given a chance to commit possibly because they are caught in an infinite loop or some other error. 

\setlength{\intextsep}{0pt}
\vspace{1mm}
\setlength{\textfloatsep}{0pt}
\begin{algorithm}
	\caption{Insert($LL, e$): Invoked by a thread to insert an element $e$ into a linked-list $LL$. This method is implemented using transactions.} \label{algo:sfdemo} 
	\begin{algorithmic}[1]
		\State $retry$ = 0; 
		\While {$(true)$} \label{lin:wstart1}
		\State $id$ = \begt($retry$); \label{lin:beg-illus}
		\State ...
		\State ...
		\State $v$ = $\tread(id, x)$; \Comment{reads the value of $x$ as $v$}
		\State ...
		\State ...
		\State $\twrite(id, x, v')$; \Comment{writes a value $v'$ to $x$}
		\State ...
		\State ...
		\State $ret$ = $\tryc(id)$; \Comment{$\tryc$ can return $commit$ or $abort$}
		\If {($ret == commit$)}
		\State break;
		\Else 
		\State $retry$++; \label{lin:retry}
		\EndIf
		\EndWhile \label{lin:wend1}
	\end{algorithmic}

\end{algorithm}

\emph{Wait-freedom} is another interesting progress condition for STMs in which every transaction commits regardless of the nature of concurrent transactions and the underlying scheduler \cite{HerlihyShavit:Progress:Opodis:2011}. But it was shown by Guerraoui and  Kapalka \cite{Bushkov+:Live-TM:PODC:2012} that it is not possible to achieve \emph{wait-freedom} in dynamic STMs in which data sets of transactions are not known in advance. So in this paper, we explore the weaker progress condition of \emph{\stfdm} for transactional memories while assuming that the data sets of the transactions are \textit{not} known in advance. 


\ignore {
\color{red}
With a \stf STM, the thread invoking insert in \algoref{sfdemo} will eventually be able to complete. Otherwise, every transaction invoked by the thread could potentially abort and the \mth will never complete. 
\color{black}

}

\vspace{1mm}
\noindent
\textbf{Related work on the \stf STMs:} Starvation-freedom in STMs has been explored by a few researchers in literature such as Gramoli et al. \cite{Gramoli+:TM2C:Eurosys:2012}, Waliullah and Stenstrom \cite{WaliSten:Starve:CC:2009}, Spear et al. \cite{Spear+:CCM:PPoPP:2009}. Most of these systems work by assigning priorities to transactions. In case of a conflict between two transactions, the transaction with lower priority is aborted. They ensure that every aborted transaction, on being retried a sufficient number of times, will eventually have the highest priority and hence will commit. We denote such an algorithm as \emph{single-version \stf STM} or \emph{\svsftm}. 

Although \svsftm guarantees \stfdm, it can still abort many transactions spuriously. Consider the case where a transaction $T_i$ has the highest priority. Hence, as per \svsftm, $T_i$ cannot be aborted. But if it is slow (for some reason), then it can cause several other conflicting transactions to abort and hence, bring down the efficiency and progress of the entire system.

\vspace{1mm}

\figref{svsf-draw} illustrates this problem. Consider the execution: $r_1(x,0) r_1(y,0) w_2(x,10) w_2(z,10) w_3(y,15) w_1(z,7)$. It has three transactions $T_1$, $T_2$ and $T_3$. Let $T_1$ has the highest priority. After reading $y$, suppose $T_1$ becomes slow. Next $T_2$ and $T_3$ want to write to $x, z$ and $y$ respectively and \emph{commit}. But $T_2$ and $T_3$'s write \op{s} are in conflict with $T_1$'s read operations. Since $T_1$ has higher priority and has not committed yet, $T_2$ and $T_3$ have to \emph{abort}. If these transactions are retried and again conflict with $T_1$ (while it is still live), they will have to \emph{abort} again. Thus, any transaction with the priority lower than $T_1$ and conflicts with it has to abort. It is as if $T_1$ has locked the \tobj{s} $x, y$ and does not allow any other transaction, write to these \tobj{s} and to \emph{commit}.

\ignore{
	\color{red} This text should be in another section: 
	
	To illustrate the need for starvation freedom, consider a transaction, say $T1$, that reads all t-objects to obtain a consistent state of the system. This can be achieved by having the transaction read all t-objects with the highest timestamp that is less than $\cts$ value. However, such a transaction may abort because by the time the transaction reads the a $\tobj$, it may be deleted due to the fact that $K$ additional versions are created. Furthermore, transactions that deleted those versions were not aware of the interest of $T1$ in reading those \tobj{s}. However, in our protocol, if this transaction is aborted several times eventually $\wts_{T1}$ would be high thereby preventing it from being aborted due to unavailability of the version. 
	
	\color{blue}
}

\vspace{1mm}
\begin{figure}
	\center
	\scalebox{0.45}{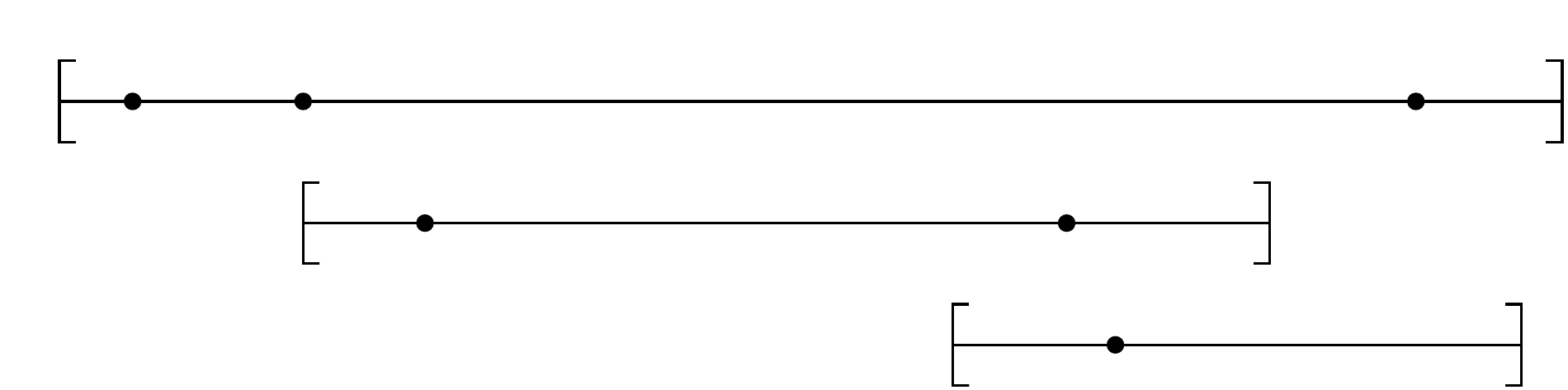}
	\captionsetup{justification=centering}
	\caption{Limitation of Single-version Starvation Free Algorithm} 
	\label{fig:svsf-draw}
\end{figure}
\vspace{1mm}
\noindent
\textbf{Multi-version \stf STM:} A key limitation of single-version STMs is limited concurrency. As shown above, it is possible that one long transaction conflicts with several transactions causing them to abort. This limitation can be overcome by using multi-version STMs where we store multiple versions of the data item (either unbounded versions with garbage collection, or bounded versions where the oldest version is replaced when the number of versions exceeds the bound).


Several multi-version STMs have been proposed in the literature \cite{Kumar+:MVTO:ICDCN:2014,LuScott:GMV:DISC:2013,Fern+:LSMVSTM:PPoPP:2011,Perel+:2011:SMV:DISC} that provide increased concurrency. But none of them provide \stfdm. Furthermore, achieving \stfdm while using only bounded versions is especially challenging given that a transaction may rely on the oldest version that is removed. In that case, it would be necessary to abort that transaction, making it harder to achieve \stfdm. 

A typical code using STMs is as shown in \algoref{sfdemo}. It shows the overview of a concurrent \emph{insert} \mth which inserts an element $e$ into a linked-list $LL$. It consists of a loop where the thread creates a transaction. This transaction executes the code to insert an element $e$ in a linked-list $LL$ using $\tread$ and $\twrite$ operations. (The result of $\twrite$ operation are stored locally.) At the end of the transaction, the thread calls \textit{\tryc}. At this point, the STM checks if the given transaction can be \emph{committed} while satisfying the required safety properties (e.g., \slbty \cite{Papad:1979:JACM}, \opty \cite{GuerKap:2008:PPoPP}). If yes, then the transaction is \emph{committed}. At this time, any updates done by the transaction are reflected in the shared memory. Otherwise, it is aborted. In this case, all the updates made by the transaction are discarded. If the given transaction is aborted, then the invoking thread may retry that transaction again like \linereff{retry} in \algoref{sfdemo}. 

The advantage of multi-version STMs, is that they allow greater concurrency by allowing more transactions to commit. Consider the execution shown in \figref{svsf-draw}. Suppose this execution used multiple versions for each \tobj. Then it is possible for all the three transactions to commit. Transactions $T_2$ and $T_3$ create a new version corresponding to each \tobj $x$, $z$ and $y$ and return commit. Since multiple versions are being used, $T_1$ need not abort as well. $T_1$ reads the initial value of $z$, and returns commit. So, by maintaining multiple versions all the transactions $T_1$, $T_2$, and $T_3$ can commit with equivalent serial history as $T_1 T_2 T_3$ or $T_1 T_3 T_2$. Thus multiple versions can help with \stfdm without sacrificing on concurrency. This motivated us to develop a multi-version \stf STM system.

Although multi-version STMs provide greater concurrency, they suffer from the cost of garbage collection. One way to avoid this is to use bounded-multi-version STMs, where the number of versions is bounded to be at most $K$. Thus, when $(K+1)^{th}$ version is created, the oldest version is removed. Bounding the number of versions can hinder with starvation freedom: a transaction needing to read a version that is currently removed must be aborted.

This paper addresses this gap by developing a starvation-free algorithm for bounded \mvstm{s}. Our approach is different from  the approach used in \svsftm to provide starvation-freedom in single version STMs (the policy of aborting lower priority transactions in case of conflict) as it does not work for \mvstm{s}. As part of the derivation of our final \stf algorithm, we consider an algorithm (\pkto) that considers this approach and show that it is insufficient to provide starvation freedom.

\color{black}

\color{black}
\vspace{1mm}
\noindent
\textbf{Contributions of the paper:}

\begin{itemize}
\item We propose a multi-version \stf STM system as \emph{K-version \stf STM} or \emph{\ksftm} for a given parameter $K$. Here $K$ is the number of versions of each \tobj and can range from 1 to $\infty$. To the best of our knowledge, this is the first \stf \mvstm. We develop \ksftm algorithm in a step-wise manner starting from \mvto \cite{Kumar+:MVTO:ICDCN:2014} as follows:  
\begin{itemize}
\item First, in \subsecref{mvto}, we use the standard idea to provide higher priority to older transactions. Specifically, we propose priority-based $K$-version STM algorithm \emph{Priority-based $K$-version MVTO} or \emph{\pkto}. This algorithm guarantees the safety properties of \stsbty and \lopty. However, it is not \stf. 
\item We analyze \pkto to identify the characteristics that will help us to achieve preventing a transaction from getting aborted forever. This analysis leads us to the development of \emph{\stf K-version TO} or \emph{\sfkv}  (\subsecref{sfmvto}), a multi-version starvation-free STM obtained by revising \pkto. But \sfkv does not satisfy correctness, i.e., \stsbty, and \lopty.
\item Finally, we extend \sfkv to develop \ksftm (\subsecref{ksftm}) that preserves the \stfdm, strict-serializability, and \lopty. Our algorithm works on the assumption that any transaction that is not deadlocked, terminates (commits or aborts) in a bounded time. 
\end{itemize}
\item Our experiments (\secref{exp}) show that \ksftm gives an average speedup on the worst-case time to commit of a transaction by a factor of 1.22, 1.89, 23.26 and 13.12 times over \pkto, \svsftm, NOrec STM \cite{Dalessandro:2010:NSS:1693453.1693464} and ESTM \cite{Felber:2017:jpdc} respectively for counter application. \ksftm performs 1.5 and 1.44 times better than \pkto and \svsftm but 1.09 times worse than NOrec for low contention KMEANS application of STAMP \cite{stamp123} benchmark whereas \ksftm performs 1.14, 1.4 and 2.63 times better than \pkto, \svsftm and NOrec for LABYRINTH application of STAMP benchmark which has high contention with long-running transactions.

\end{itemize}

\section{System Model and Preliminaries}
\label{sec:model}

Following~\cite{tm-book,KuzSat:NI:TCS:2016}, we assume a system of $n$ processes/threads, $p_1,\ldots,p_n$ that access a collection of \emph{transactional objects} (or \emph{{\tobj}s}) via atomic \emph{transactions}. Each transaction has a unique identifier. Within a transaction, processes can perform \emph{transactional operations or \mth{s}}: $\begt{}$ that begins a transaction, \textit{\twrite}$(x,v)$ operation that updates a \tobj $x$ with value $v$ in its local memory, the \textit{\tread}$(x)$ operation tries to read  $x$, \textit{\tryc}$()$ that tries to commit the transaction and returns $commit$ if it succeeds, and \textit{\trya}$()$ that aborts the transaction and returns $\abort$. For the sake of presentation simplicity, we assume that the values taken as arguments by \textit{\twrite} operations are unique. 

Operations \textit{\tread} and \textit{\tryc}$()$ may return $\abort$, in which case we say that the operations \emph{forcefully abort}. Otherwise, we say that the operations have \emph{successfully} executed.  Each operation is equipped with a unique transaction identifier. A transaction $T_i$ starts with the first operation and completes when any of its operations return $\abort$ or $\commit$. We denote any \op{} that returns $\abort$ or $\commit$ as \emph{\termop{s}}. Hence, \op{s} $\tryc$$()$ and $\trya$$()$ are \termop{s}. A transaction does not invoke any further \op{s} after \termop{s}. 

For a transaction $T_k$, we denote all the \tobj{s} accessed by its read \op{s} as $\rs_k$ and \tobj{s} accessed by its write operations as $\ws_k$.  We denote all the \op{s} of a transaction $T_k$ as $\evts{T_k}$ or $evts_k$. 

\noindent
\textbf{History:}
A \emph{history} is a sequence of \emph{events}, i.e., a sequence of 
invocations and responses of transactional operations. The collection of events 
is denoted as $\evts{H}$. For simplicity, we only consider \emph{sequential} 
histories here: the invocation of each transactional operation is immediately 
followed by a matching response. Therefore, we treat each transactional 
operation as one atomic event, and let $<_H$ denote the total order on the 
transactional operations incurred by $H$. With this assumption, the only 
relevant events of a transaction $T_k$ is of the types: $r_k(x,v)$, 
$r_k(x,\abort)$, $w_k(x, v)$, $\tryc_k(\commit)$ (or $c_k$ for short), 
$\tryc_k(\abort)$, $\trya_k(\abort)$ (or $a_k$ for short).  We identify a 
history $H$ as tuple $\langle \evts{H},<_H \rangle$. 


Let $H|T$ denote the history consisting of events of $T$ in $H$, and $H|p_i$ denote the history consisting of events of $p_i$ in $H$. We only consider \emph{well-formed} histories here, i.e., no transaction of a process begins before the previous transaction invocation has completed (either $commits$ or $aborts$). We also assume that every history has an initial \emph{committed} transaction $T_0$ that initializes all the t-objects with value $0$. 

The set of transactions that appear in $H$ is denoted by $\txns{H}$. The set of \emph{committed} (resp., \emph{aborted}) transactions in $H$ is denoted by $\comm{H}$ (resp., $\aborted{H}$). The set of \emph{incomplete} or \emph{live} transactions in $H$ is denoted by $\incomp{H} = \live{H} = (\txns{H}-\comm{H}-\aborted{H})$. 

For a history $H$, we construct the \emph{completion} of $H$, denoted as 
$\overline{H}$, by inserting $\trya_k(\abort)$ immediately after the last event 
of every transaction $T_k\in \live{H}$. But for $\tryc_i$ of transaction $T_i$, if it released the lock 
on first \tobj successfully that means updates made by $T_i$ is consistent so, $T_i$ will immediately return commit.

\noindent
\textbf{Transaction orders:} For two transactions $T_k,T_m \in \txns{H}$, we say that  $T_k$ \emph{precedes} $T_m$ in the \emph{real-time order} of $H$, denote $T_k\prec_H^{RT} T_m$, if $T_k$
is complete in $H$ and the last event of $T_k$ precedes the first event of $T_m$ in $H$. If neither $T_k \prec_H^{RT} T_m$ nor $T_m \prec_H^{RT} T_k$, then $T_k$ and $T_m$ \emph{overlap} in $H$. We say that a history is \emph{\tseq} if all the transactions are ordered by this real-time order. Note that from our earlier assumption all the transactions of a single process are ordered by real-time.

\ignore{
We say that $T_k, T_m$ are in conflict, if (1) (tryc-tryc) $\tryc_k(C)<_H \tryc_m(C)$ and $Wset(T_k) \cap Wset(T_m) \neq\emptyset$; (2) (tryc-r) $\tryc_k(C)<_H r_m(x,v)$, $x \in Wset(T_k)$ and $v \neq A$; (3) (r-tryc) $r_k(x,v)<_H \tryc_m(C)$, $x\in Wset(T_m)$ and $v \neq A$. 

Thus, it can be seen that the conflict order is defined only on \op{s} that have successfully executed. We denote the corresponding \op{s} as conflicting. 
}


\noindent
\textbf{Sub-history:} A \textit{sub-history} ($SH$) of a history ($H$) 
denoted as the tuple $\langle \evts{SH},$ $<_{SH}\rangle$ and is defined as: 
(1) $<_{SH} \subseteq <_{H}$; (2) $\evts{SH} \subseteq \evts{H}$; (3) If an 
event of a transaction $T_k\in\txns{H}$ is in $SH$ then all the events of $T_k$ 
in $H$ should also be in $SH$. 

For a history $H$, let $R$ be a subset of $\txns{H}$. Then $\shist{R}{H}$ denotes  the \ssch{} of $H$ that is formed  from the \op{s} in $R$. 

\noindent
\textbf{Valid and legal history:}
%
A successful read $r_k(x, v)$ (i.e., $v \neq \abort$)  in a history $H$ is said to be \emph{\valid} if there exist a transaction $T_j$ that wrote $v$ to $x$ and \emph{committed} before $r_k(x,v)$. Formally, $\langle r_k(x, v)$  is \valid{} $\Leftrightarrow \exists T_j: (c_j <_{H} r_k(x, v)) \land (w_j(x, v) \in \evts{T_j}) \land (v \neq \abort) \rangle$.  The history $H$ is \valid{} 
if all its successful read \op{s} are \valid. 

We define $r_k(x, v)$'s \textit{\lastw{}} as the latest commit event $c_i$ preceding $r_k(x, v)$ in $H$ such that $x\in wset_i$ ($T_i$ can also be $T_0$). A successful read \op{} $r_k(x, v)$, is
said to be \emph{\legal{}}  if the transaction containing $r_k$'s \lastw{} also writes $v$ onto $x$:  $\langle r_k(x, v)$ \text{is \legal{}} $\Leftrightarrow (v \neq \abort) \land (\lwrite{r_k(x, v)}{H} = c_i) \land (w_i(x,v) \in \evts{T_i})\rangle$.  The history $H$ is \legal{} if all its successful read \op{s} are \legal.  From the definitions we get that if $H$ is \legal{} then it is also \valid.


\noindent
\textbf{Opacity and Strict Serializability:} We say that two histories 
$H$ and $H'$ are \emph{equivalent} if they have the same set of events. Now a 
history $H$ is said to be \textit{opaque} \cite{GuerKap:2008:PPoPP,tm-book} if 
it is \valid{} and there exists a t-sequential legal history $S$ such that (1) 
$S$ is equivalent to $\overline{H}$ and (2) $S$ respects $\prec_{H}^{RT}$, i.e., 
$\prec_{H}^{RT} \subset \prec_{S}^{RT}$.  By requiring $S$ being equivalent to 
$\overline{H}$, opacity treats all the incomplete transactions as aborted. We 
call $S$ an (opaque) \emph{serialization} of $H$. 

Along same lines, a \valid{} history $H$ is said to be \textit{strictly serializable} if $\shist{\comm{H}}{H}$ is opaque. Unlike opacity, strict serializability does not include aborted or incomplete transactions in the global serialization order. An opaque history $H$ is also strictly serializable: a serialization of $\shist{\comm{H}}{H}$ is simply the subsequence of a serialization of $H$ that only contains transactions in $\comm{H}$. 

Serializability is commonly used criterion in databases. But it is not suitable for STMs as it does not consider the correctness of \emph{aborted} transactions as shown by Guerraoui \& Kapalka \cite{GuerKap:2008:PPoPP}. Opacity, on the other hand, considers the correctness of \emph{aborted} transactions as well. Similarly, \lopty (described below) is another \cc for STMs but is not as restrictive as \opty. 

\noindent
\textbf{Local opacity:} For a history H, we define a set of sub-histories, denoted as $\shset{H}$ as follows: (1) For each aborted transaction $T_i$, we consider a $\subhist$ consisting of \op{s} from all previously \emph{committed} transactions and including all successful \op{s} of $T_i$ (i.e., \op{s} which did not return $\mathcal{A}$) while immediately putting commit after last successful operation of $T_i$; (2) for last \emph{committed} transaction $T_l$ considers all the previously \emph{committed} transactions including  $T_l$.

A history H is said to be \emph{\lopq} \cite{KuzSat:NI:ICDCN:2014,KuzSat:NI:TCS:2016} if all the sub-histories in \shset{H} are opaque. It must be seen that in the construction of sub-history of an aborted transaction $T_i$, the $\subhist$ will contain \op{s} from only one aborted transaction which is $T_i$ itself and no other live/aborted transactions. Similarly, the sub-history of \emph{committed} transaction $T_l$ has no \op{s} of aborted and live transactions. Thus in \lopty, no aborted or live transaction can cause another transaction to abort. It was shown that \lopty 
\cite{KuzSat:NI:ICDCN:2014,KuzSat:NI:TCS:2016} allows greater concurrency than \opty. Any history that is \opq is also \lopq but not necessarily the vice-versa. On the other hand, a history that is \lopq is also \stsble, but the vice-versa need not be true.\\
\noindent
\textbf{Graph Characterization of Local Opacity:}
To prove correctness of STM systems, it is useful to consider graph characterization of histories. In this section, we describe the graph characterization developed by Kumar et al \cite{Kumar+:MVTO:ICDCN:2014} for proving \opty which is based on characterization by Bernstein and Goodman \cite{BernGood:1983:MCC:TDS}.  We extend this characterization for \lo. 

Consider a history $H$ which consists of multiple versions for each \tobj. The graph characterization uses the notion of \textit{version order}. Given $H$ and a \tobj{} $x$, we define a version order for $x$ as any (non-reflexive) total order on all the versions of $x$ ever created by committed transactions in $H$. It must be noted that the version order may or may not be the same as the actual order in which the version of $x$ are generated in $H$. A version order of $H$, denoted as $\ll_H$ is the union of the version orders of all the \tobj{s} in $H$. 

Consider the history $H2: r_1(x, 0) r_2(x, 0) r_1(y, 0) r_3(z, 0) w_1(x, 5) w_3(y, 15)  w_2(y, 10) w_1(z, 10)
c_1 c_2 r_4(x, 5) \\r_4(y, 10) w_3(z, 15) c_3 r_4(z, 10)$. Using the notation that a committed transaction $T_i$ writing to $x$ creates a version $x_i$, a possible version order for $H2$ $\ll_{H2}$ is: $\langle x_0 \ll x_1 \rangle, \langle y_0 \ll y_2 \ll y_3 \rangle, \langle z_0 \ll z_1 \ll z_3 \rangle $. 

We define the graph characterization based on a given version order. Consider a history $H$ and a version order $\ll$. We then define a graph (called opacity graph) on $H$ using $\ll$, denoted as $\opg{H}{\ll} = (V, E)$. The vertex set $V$ consists of a vertex for each transaction $T_i$ in $\overline{H}$. The edges of the graph are of three kinds and are defined as follows:

\begin{enumerate}	
	\item \textit{\rt}(real-time) edges: If $T_i$ commits before $T_j$ starts in $H$, then there is an edge from $v_i$ to $v_j$. This set of edges are referred to as $\rtx(H)$.
	
	\item \textit{\rf}(reads-from) edges: If $T_j$ reads $x$ from $T_i$ in $H$, then there is an edge from $v_i$ to $v_j$. Note that in order for this to happen, $T_i$ must have committed before $T_j$ and $c_i <_H r_j(x)$. This set of edges are referred to as $\rf(H)$.
	
	\item \textit{\mv}(multiversion) edges: The \mv{} edges capture the multiversion relations and is based on the version order. Consider a successful read \op{} $r_k(x,v)$ and the write \op{} $w_j(x,v)$ belonging to transaction $T_j$ such that $r_k(x,v)$ reads $x$ from $w_j(x,v)$ (it must be noted $T_j$ is a committed transaction and $c_j <_H r_k$). Consider a committed transaction $T_i$ which writes to $x$, $w_i(x, u)$ where $u \neq v$. Thus the versions created $x_i, x_j$ are related by $\ll$. Then, if $x_i \ll x_j$ we add an edge from $v_i$ to $v_j$. Otherwise ($x_j \ll x_i$), we add an edge from $v_k$ to $v_i$. This set of edges are referred to as $\mv(H, \ll)$.
\end{enumerate}

Using the construction, the $\opg{H2}{\ll_{H2}}$ for history $H2$ and $\ll_{H2}$ is shown in \figref{opg}. The edges are annotated. The only \mv{} edge from $T4$ to $T3$ is because of \tobj{s} $y, z$. $T4$ reads value 5 for $z$ from $T1$ whereas $T3$ also writes 15 to $z$ and commits before $r_4(z)$. 

\begin{figure}[tbph]
	\centerline{\scalebox{0.7}{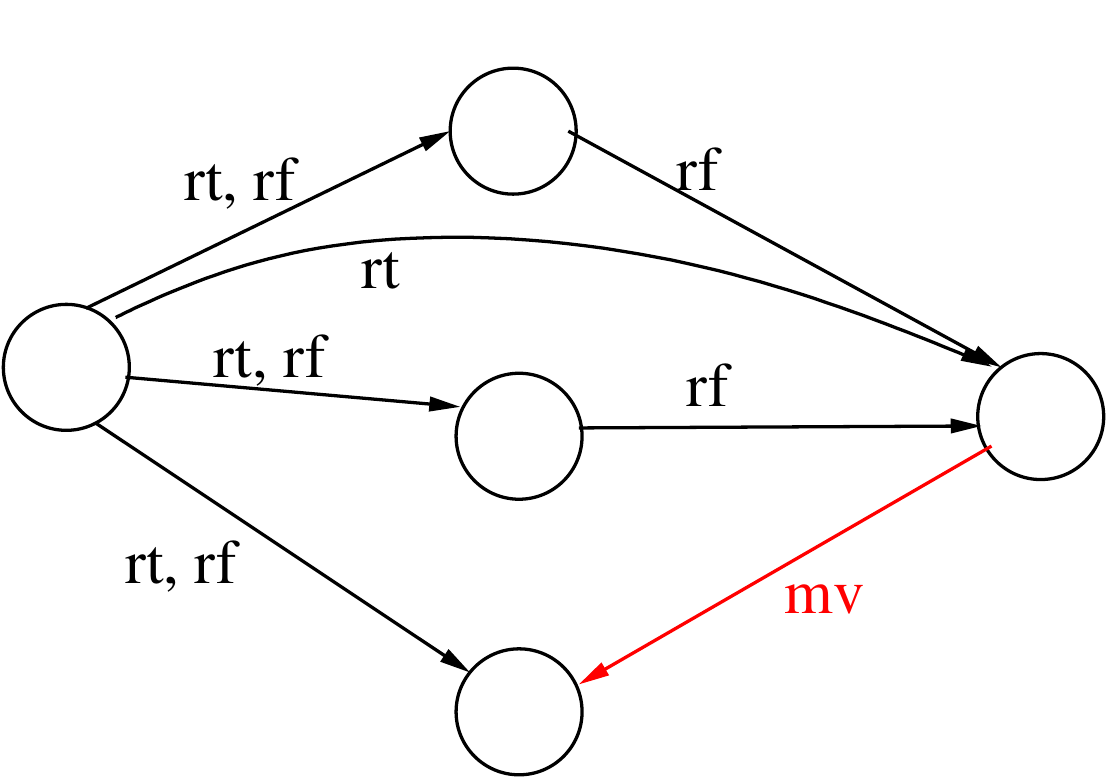}}
	\captionsetup{justification=centering}
	\caption{$\opg{H2}{\ll_{H2}}$}
	\label{fig:opg}
\end{figure}

Kumar et al \cite{Kumar+:MVTO:ICDCN:2014} showed that if a version order $\ll$ exists for a history $H$ such that $\opg{H}{\ll_H}$ is acyclic, then $H$ is \opq. This is captured in the following result.

\ignore{
	\begin{result}
		\label{res:main-opg}
		A \valid{} history $H$ is opaque iff there exists a version order $\ll_H$ such that $\opg{H}{\ll_H}$ is acyclic.
	\end{result}
	
	\noindent This result can be extended to characterize \lo using graphs with the following theorem. The proof is in Appendix \thmref{log}. 
	
	\begin{theorem}
		\label{thm:main-log}
		A \valid{} history $H$ is \lopq iff for each sub-history $sh$ in $\shset{H}$ there exists a version order $\ll_{sh}$ such that $\opg{sh}{\ll_{sh}}$ is acyclic. Formally, $\langle (H \text{ is \lopq}) \Leftrightarrow (\forall sh \in \shset{H}, \exists \ll_{sh}: \opg{sh}{\ll_{sh}} \text{ is acyclic}) \rangle$. 
	\end{theorem}
}

\begin{result}
	\label{res:opg}
	A \valid{} history $H$ is opaque iff there exists a version order $\ll_H$ such that $\opg{H}{\ll_H}$ is acyclic.
\end{result}

\noindent This result can be easily extended to prove \lo as follows

\begin{theorem}
	\label{thm:log}
	A \valid{} history $H$ is \lopq iff for each sub-history $sh$ in $\shset{H}$ there exists a version order $\ll_{sh}$ such that $\opg{sh}{\ll_{sh}}$ is acyclic. Formally, $\langle (H \text{ is \lopq}) \Leftrightarrow (\forall sh \in \shset{H}, \exists \ll_{sh}: \opg{sh}{\ll_{sh}} \text{ is acyclic}) \rangle$. 
\end{theorem}

\begin{proof}
	To prove this theorem, we have to show that each sub-history $sh$ in $\shset{H}$ is \valid. Then the rest follows from \resref{opg}. Now consider a sub-history $sh$. Consider any read \op $r_i(x, v)$ of a transaction $T_i$. It is clear that $T_i$ must have read a version of $x$ created by a previously committed transaction. From the construction of $sh$, we get that all the transaction that committed before $r_i$ are also in $sh$. Hence $sh$ is also \valid. 
	
	Now, proving $sh$ to be \opq iff there exists a version order $\ll_{sh}$ such that $\opg{sh}{\ll_{sh}}$ is acyclic follows from \resref{opg}. 
\end{proof} 


\section{The Working of \ksftm Algorithm}
\label{sec:idea}

In this section, we propose \emph{K-version \stf STM} or \emph{\ksftm} for a given parameter $K$. Here $K$ is the number of versions of each \tobj and can range from 1 to $\infty$. When $K$ is 1, it boils down to single-version \emph{\stf} STM. If $K$ is $\infty$, then \ksftm uses unbounded versions and needs a separate garbage collection mechanism to delete old versions like other \mvstm{s} proposed in the literature \cite{Kumar+:MVTO:ICDCN:2014,LuScott:GMV:DISC:2013}. We denote \ksftm using unbounded versions as \emph{\mvsftm} and \mvsftm with garbage collection as \emph{\mvsftmgc}. 

Next, we describe some \emph{\stfdm} preliminaries in \subsecref{prelim} to explain the working of \ksftm algorithm. To explain the intuition behind the \ksftm algorithm, we start with the modification of \mvto \cite{BernGood:1983:MCC:TDS,Kumar+:MVTO:ICDCN:2014} algorithm in \subsecref{mvto}. We then make a sequence of modifications to it to arrive at \ksftm algorithm. 

\subsection{\emph{Starvation-Freedom} Preliminaries}
\label{subsec:prelim}

In this section, we start with the definition of \emph{\stfdm}. Then we describe the invocation of transactions by the application. Next, we describe the data structures used by the algorithms. 


\begin{definition}
	\label{defn:stf}
	\textbf{Starvation-Freedom:} A STM system is said to be \stf if a thread invoking a non-parasitic transaction $T_i$ gets the opportunity to retry $T_i$ on every abort, due to the presence of a fair scheduler, then $T_i$ will eventually commit.
\end{definition}

As explained by Herlihy \& Shavit \cite{HerlihyShavit:Progress:Opodis:2011}, a fair scheduler implies that no thread is forever delayed or crashed. Hence with a fair scheduler, we get that if a thread acquires locks then it will eventually release the locks. Thus a thread cannot block out other threads from progressing. 

\ignore{
	\noindent
	\textbf{\stfdm:} A STM system is said to be \stf if a thread invoking a transaction $T_i$ gets the opportunity to retry $T_i$ on every abort due to the presence of a fair scheduler then $T_i$ will eventually commit (as described in \secref{intro}). As explained by Herlihy \& Shavit \cite{HerlihyShavit:AMP:Book:2012}, a fair scheduler implies that no thread is forever delayed or crashed. Hence with this assumption of a fair scheduler, we get that if in the course of execution a thread acquires locks then it will eventually release the locks. Since it is not forever delayed, it cannot block out other threads from progressing. 
}

\vspace{1mm}
\noindent
\textbf{Assumption about Scheduler: }
In order for \stf algorithm \ksftm (described in \subsecref{ksftm}) to work correctly, we make the following assumption about the fair scheduler: 
\begin{assumption}
	\label{asm:bdtm}
	\textbf{Bounded-Termination}: For any transaction $T_i$, invoked by a thread $Th_x$, the fair system \schdr ensures, in the absence of deadlocks, $Th_x$ is given sufficient time on a CPU (and memory etc.) such that $T_i$ terminates (either commits or aborts) in bounded time.
\end{assumption}
While the bound for each transaction may be different, we use $L$ to denote the maximum bound. In other words, in time $L$, every transaction will either abort or commit due to the absence of deadlocks. 

There are different ways to satisfy the scheduler requirement. For example, a round-robin scheduler which provides each thread equal amount of time in any window satisfies this requirement as long as the number of threads is bounded. In a system with two threads, even if a scheduler provides one thread 1\% of CPU and another thread 99\% of the CPU, it satisfies the above requirement. On the other hand, a scheduler that schedules the threads as `$T_1, T_2, T_1, T_2, T_2, T_1, T_2$, $T_2, T_2, T_2$, $T_1, T_2, T_2$, $T_2, T_2, T_2, T_2, T_2, T_2, T_1, T_2 (16 times) $' does not satisfy the above requirement. This is due to the fact that over time, thread 1 gets infinitesimally smaller portion of the CPU and, hence, the time required for it to complete (commit or abort) will continue to increase over time. 

In our algorithm, we will ensure that it is deadlock free using standard techniques from the literature. In other words, each thread is in a position to make progress. We assume that the scheduler provides sufficient CPU time to complete (either commit or abort) within a bounded time.

As explained by Herlihy \& Shavit \cite{HerlihyShavit:Progress:Opodis:2011}, a fair scheduler implies that no thread is forever delayed or crashed. Hence with a fair scheduler, we get that if a thread acquires locks then it will eventually release the locks. Thus a thread cannot block out other threads from progressing. 

\ignore{
\noindent
\textbf{\stfdm:} A STM system is said to be \emph{\stf} if a thread invoking a transaction $T_i$ gets the opportunity to retry $T_i$ on every abort due to the presence of a fair scheduler then $T_i$ will eventually commit (as described in \secref{intro}). As explained by Herlihy \& Shavit \cite{HerlihyShavit:AMP:Book:2012}, a fair scheduler implies that no thread is forever delayed or crashed. Hence with this assumption of a fair scheduler, we get that if in the course of execution a thread acquires locks then it will eventually release the locks. Since it is not forever delayed, it cannot block out other threads from progressing. 
}

\noindent
\textbf{Transaction Invocation:} Transactions are invoked by threads. Suppose a thread $Th_x$ invokes a transaction $T_i$. If this transaction $T_i$ gets \emph{aborted}, $Th_x$ will reissue it, as a new incarnation of $T_i$, say $T_j$. 
The thread $Th_x$ will continue to invoke new \inc{s} of $T_i$ until an \inc commits. 


When the thread $Th_x$ invokes a transaction, say $T_i$, for the first time then the STM system assigns $T_i$ a unique timestamp called \emph{current timestamp or \cts}. If it aborts and retries again as $T_j$, then its \cts will change. However, in this case, the thread $Th_x$ will also pass the \cts value of the first incarnation ($T_i$) to $T_j$. By this, $Th_x$ informs the STM system that, $T_j$ is not a new invocation but is an \inc of $T_i$.



We denote the \cts of $T_i$ (first \inc) as \emph{Initial Timestamp or \its} for all the \inc{s} of $T_i$. Thus, the invoking thread $Th_x$ passes $\tcts{i}$ to all the \inc{s} of $T_i$ (including $T_j$). Thus for $T_j$, $\tits{j} = \tcts{i}$. The transaction $T_j$ is associated with the timestamps: $\langle \tits{j}, \tcts{j} \rangle$. For $T_i$, which is the initial \inc, its \its and \cts are the same, i.e., $\tits{i} = \tcts{i}$. For simplicity, we use the notation that for transaction $T_j$, $j$ is its \cts, i.e., $\tcts{j} = j$.
\begin{figure}[H]
	\centerline{\scalebox{0.50}{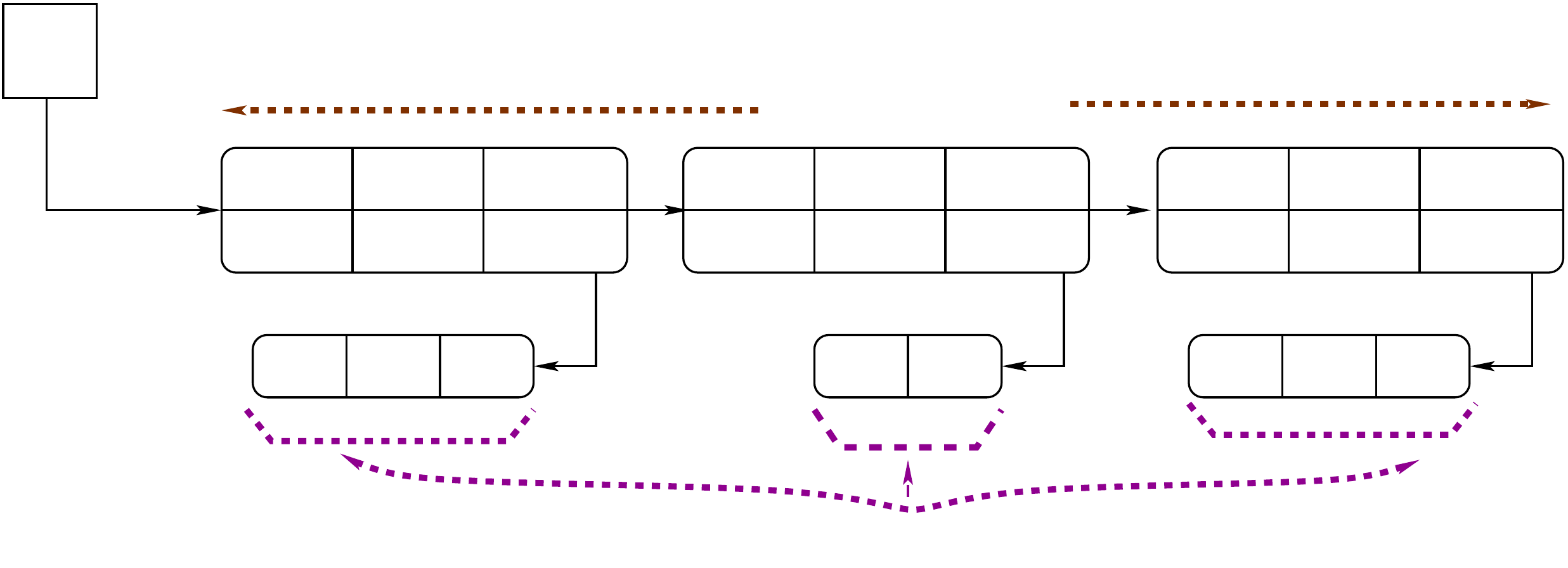}}
	\captionsetup{justification=centering}	
	\caption{Data Structures for Maintaining Versions}
	\label{fig:tobj}
\end{figure}

We also assume that in the absence of other concurrent conflicting transactions, every transaction will commit. In other words, if a transaction is executing in a system where other concurrent conflicting transactions are not present then it will not self-abort. If transactions can self-abort then providing \emph{\stfdm} is impossible.

\ignore{
\color{red}
We assume that a transaction $T_i$ on getting \emph{aborted} will be retried by the invoking thread $Th_x$, as a new transaction, say $T_j$. $T_i$ and $T_j$ are said to \emph{incarnations} of each other. The thread $Th_x$ will continue to invoke a new \inc{s} of $T_i$ until an \inc commits. Thus, we can order all the \inc{s} of $T_i$ based on the order in which they have been invoked. 

Whenever a transaction $T_i$ begins, it is assigned a unique timestamp by the STM system called as current timestamp or \emph{\cts}. $T_i$ on getting \emph{aborted} will be retried by the invoking thread $Th_x$, as a new transaction, say $T_j$. While invoking $T_j$, we assume that $Th_x$ will pass the \cts of the first \inc of $T_j$, which in this case is suppose $T_i$, to the STM system. By this, $Th_x$ informs the STM system this $T_j$ is not the initial invocation and is an \inc of $T_i$. In this case, we denote the \cts of $T_i$ as \emph{Initial Timestamp} or \emph{\its} for all the \inc{s} of $T_i$. Any \inc of $T_i$, $T_j$ is characterized by : $\langle \its_j, \cts_j \rangle$. For $T_i$ both its \its and \cts are the same.

\color{black}
}

\noindent
\textbf{Common Data Structures and STM Methods:} Here we describe the common data structures used by all the algorithms proposed in this section. For each \tobj, the algorithms maintain multiple versions in $version$-$list$ (or \emph{\vlist}) using list. Similar to versions in \mvto \cite{Kumar+:MVTO:ICDCN:2014}, each version of a \tobj{} is a tuple denoted as \emph{\vtup} and consists of three fields: (1) timestamp, (or $ts$) of the transaction that created this version which normally is the \cts; (2) the value (or $val$) of the version; (3) a list, called \rlist{} (or $rl$), consisting of transactions ids (can be \cts as well) that read from this version. The \rlist of a version is initially empty. \figref{tobj} illustrates this structure. For a \tobj $x$, we use the notation $x[t]$ to access the version with timestamp $t$. Depending on the algorithm considered, the fields change of this structure. 

The algorithms have access to a global atomic counter, $\gtcnt$ used for generating timestamps in the various transactional \mth{s}. We assume that the STM system exports the following \mth{s} for a transaction $T_i$: (1) $\begt(t)$ where $t$ is provided by the invoking thread, $Th_x$. From our earlier assumption, it is the \cts of the first \inc. In case $Th_x$ is invoking this transaction for the first time, then $t$ is $null$. This \mth returns a unique timestamp to $Th_x$ which is the \cts/id of the transaction. (2) $\tread_i(x)$ tries to read  \tobj $x$. It returns either value $v$ or $\abort$. (3) $\twrite_i(x,v)$ operation that updates a \tobj $x$ with value $v$ locally. It returns $ok$. (4) $\tryc_i()$ tries to commit the transaction and returns $\commit$ if it succeeds. Otherwise, it returns $\abort$. 

\ignore{
In addition to these global structures, for each transaction $T_i$, \pkto maintains structure that are local to $T_i$:
\begin{itemize}
	\item $\rs_i$(read-set): It is a list of data tuples or \emph{\dtup} of the form $\langle x, val \rangle$, where $x$ is the t-object and $v$ is the value read by the transaction $T_i$. We refer to a tuple in $T_i$'s read-set by $\rs_i[x]$.
	\item $\ws_i$(write-set): It is a list of \dtup{s} where $x$ is the \tobj{} to which transaction $T_i$ writes the value $val$. Similarly, we refer to a tuple in $T_i$'s write-set by $\ws_i[x]$.
\end{itemize}
\noindent Next, we describe the \mth{s} exported by \pkto. Here in this discussion, we use the notation that a transaction $T_i$ has its \cts as $i$. 
}

\noindent
\textbf{Correctness Criteria:} 
For ease of exposition, we initially consider \stsbty as \emph{\cc} to illustrate the correctness of the algorithms. But \stsbty does not consider the correctness of \emph{aborted} transactions and as a result not a suitable \emph{\cc} for STMs. Finally, we show that the proposed STM algorithm \ksftm satisfies \lopty, a \emph{\cc} for STMs (described in \secref{model}).  We denote the set of histories generated by an STM algorithm, say $A$, as $\gen{A}$.

\subsection{Motivation for Starvation Freedom in Multi-Version Systems
}
\label{sec:sftm}

In this section, first we describe the starvation freedom solution used for single version i.e. \sftm algorithm and then the drawback of it. 
\subsubsection{Illustration of \sftm}
\label{subsec:sftm-main}

Forward-oriented optimistic concurrency control protocol (\focc), is a commonly used optimistic algorithm in databases \cite[Chap 4]{WeiVoss:2002:Morg}. In fact, several STM Systems are also based on this idea. In a typical STM system (also in database optimistic concurrency control algorithms), a transaction execution is divided can be two phases - a \emph{\rwph} and \emph{\tryph} (also referred to as validation phase in databases). The various algorithms differ in how the \tryph{} executes. Let the write-set or \ws{} and read-set or \rs{} of a $t_i$ denotes the set of \tobj{s} written \& read by $t_i$. In \focc{} a transaction $t_i$ in its \tryph{} is validated against all live transactions that are in their \rwph{} as follows: $\langle \ws(t_i) \cap (\forall t_j: \rs^{n}(t_j)) = \Phi \rangle$. This implies that the \ws{} of $t_i$ can not have any conflict with the current \rs{} of any transaction $t_j$ in its \rwph. Here $\rs^{n}(t_j)$ implies the \rs{} of $t_j$ till the point of validation of $t_i$. If there is a conflict, then either $t_i$ or $t_j$ (all transactions conflicting with $t_i$) is aborted. A commonly used approach in databases is to abort $t_i$, the validating transaction. 

In \sftm{} we use \emph{\ts{s}} which are monotonically in increasing order. We implement the \ts{s} using atomic counters. Each transaction $t_i$ has two time-stamps: (i) \emph{current time-stamp or \cts}: this is a unique \ts{} alloted to $t_i$ when it begins; (ii) \emph{initial time-stamp or \its}: this is same as \cts{} when a transaction $t_i$ starts for the first time. When $t_i$ aborts and re-starts later, it gets a new \cts. But it retains its original \cts{} as \its. The value of \its{} is retained across aborts. For achieving starvation freedom, \sftm{} uses \its{} with a modification to \focc{} as follows: a transaction $t_i$ in \tryph{} is validated against all other conflicting transactions, say $t_j$ which are in  their \rwph. The \its{} of $t_i$ is compared with the \its{} of any such transaction $t_j$. If \its{} of $t_i$ is smaller than \its{} of all such $t_j$, then all such $t_j$ are aborted while $t_i$ is committed. Otherwise, $t_i$ is aborted. We show that \sftm{} satisfies \opty{} and \stf. 

\begin{theorem}
	\label{thm:sftm-correct}
	Any history generated by \sftm{} is \opq.
\end{theorem}

\begin{theorem}
	\label{thm:sftm-stf}
	\sftm{} ensure \stfdm.
\end{theorem}

We prove the correctness by showing that the conflict graph \cite[Chap 3]{WeiVoss:2002:Morg}, \cite{KuzSat:NI:ICDCN:2014} of any history generated by \sftm{} is acyclic. We show \stfdm{} by showing that for each transaction $t_i$ there eventually exists a global state in which it has the smallest \its. 

\cmnt{
\subsection{Illustration of \sftm}
\label{sec:sftm-illus}
}
\figref{sftmex} shows the a sample execution of \sftm. It compares the execution of \focc{} with \sftm. The execution on the left corresponds to \focc, while the execution one the right is of \sftm{} for the same input. It can be seen that each transaction has two \ts{s} in \sftm. They correspond to \cts, \its{} respectively. Thus, transaction $T_{1,1}$ implies that \cts{} and \its{} are $1$. In this execution, transaction $T_3$ executes the read \op{} $r_3(z)$ and is aborted due to conflict with $T_2$. The same happens with $T_{3,3}$. Transaction $T_5$ is re-execution of $T_3$. With \focc{} $T_5$ again aborts due to conflict with $T_4$. In case of \sftm{}, $T_{5,3}$ which is re-execution of $T_{3,3}$ has the same \its{} $3$. Hence, when $T_{4,4}$ validates in \sftm, it aborts as $T_{5,3}$ has lower \its. Later $T_{5,3}$ commits. 

It can be seen that \its{s} prioritizes the transactions under conflict and the transaction with lower \its{} is given higher priority. 

\begin{figure}[tbph]
\centerline{\scalebox{0.5}{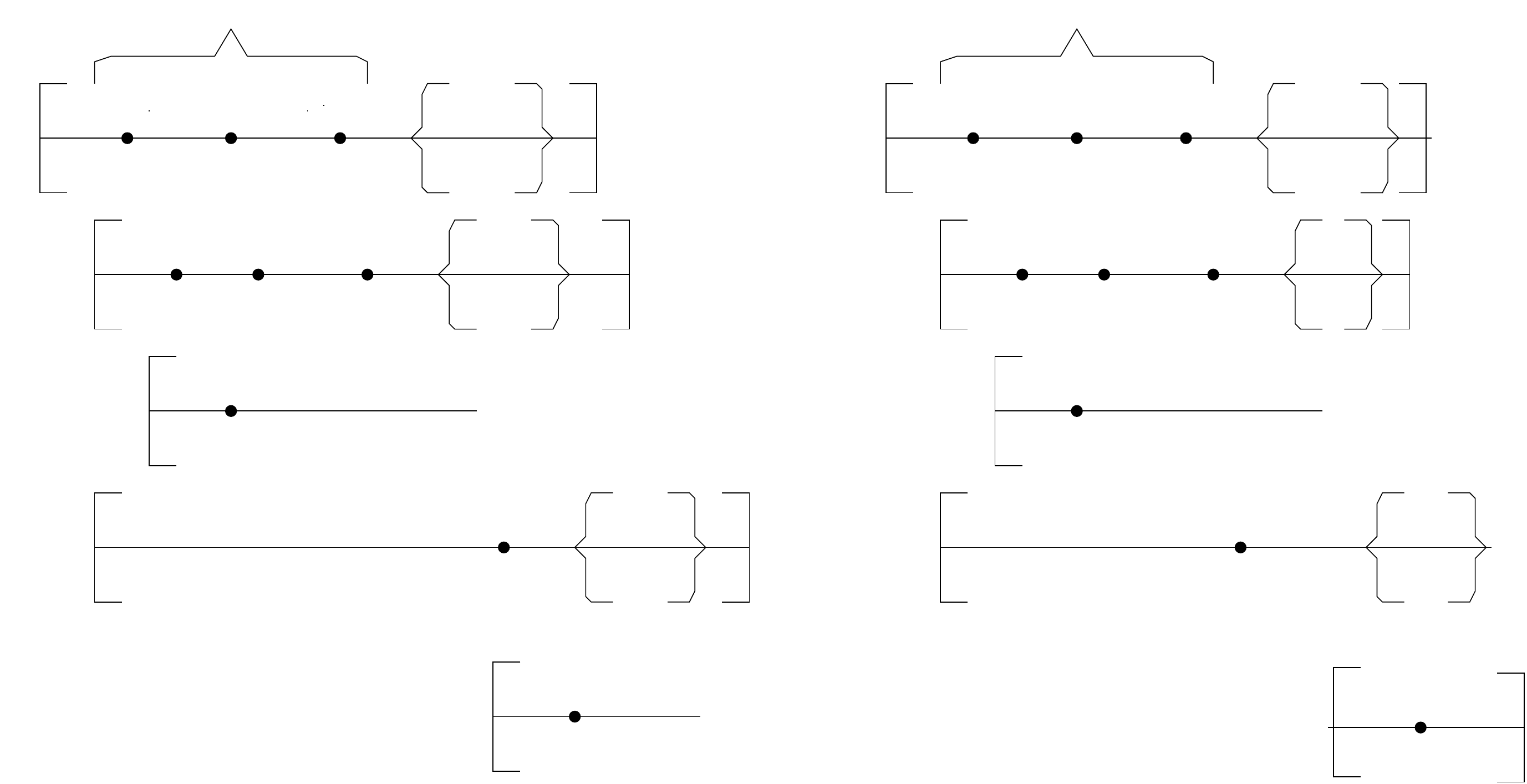}}
\caption{Sample execution of \sftm}
\label{fig:sftmex}
\end{figure}


\subsubsection{Drawback of \sftm}
\label{subsec:sftm-drawback}

Figure \ref{fig:stmvl1} is representing history H: $r_1(x,0)r_1(y,0) 
w_2(x,10)w_3(y,15)a_2 a_3 c_1$ It has three transactions $T_1$, $T_2$ and 
$T_3$. $T_1$ is having lowest time stamp and after reading it became slow. 
$T_2$ and $T_3$ wants to write to $x$ and $y$ respectively but when it came 
into validation phase, due to $r_1(x)$, $r_1(y)$ and not committed yet, $T_2$ 
and $T_3$ gets aborted. However, when we are using multiple version $T_2$ and 
$T_3$ both can commit and $T_1$ can also read from $T_0$. The equivalent serial 
history is $T_1 T_2 T_3$.

\begin{figure}[H]
	\center
    \scalebox{0.55}{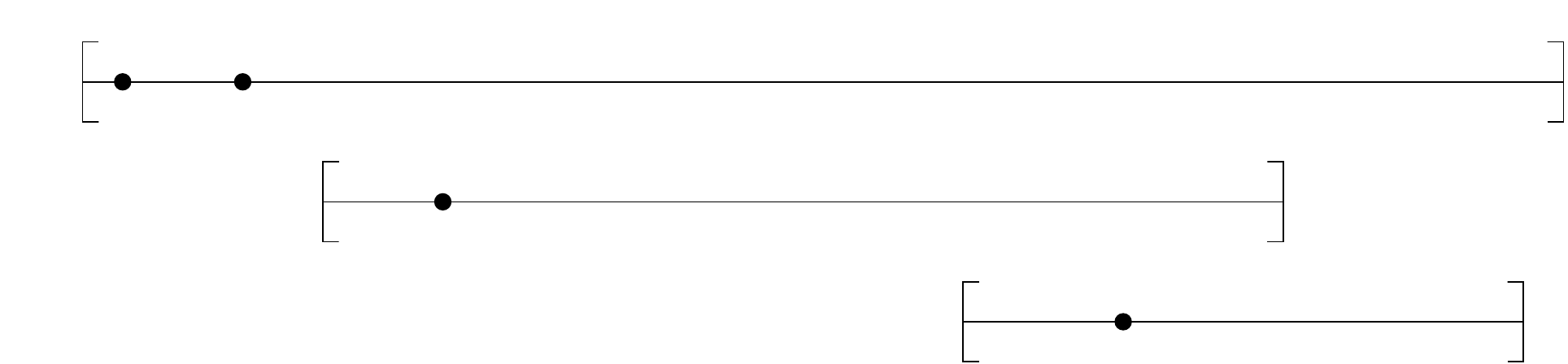}
	\caption{Pictorial representation of execution under SFTM}
	\label{fig:stmvl1}
\end{figure}
\subsubsection{Data Structures and Pseudocode of \sftm}
\label{apn:SFTM}

We start with data-structures that are local to each transaction. For each transaction $T_i$:

\begin{itemize}
	\item $\rs_i$(read-set): It is a list of data tuples ($d\_tuples$) of the form $\langle x, val \rangle$, where $x$ is the t-object and $v$ is the value read by the transaction $T_i$. We refer to a tuple in $T_i$'s read-set by $\rs_i[x]$.
	
	\item $\ws_i$(write-set): It is a list of ($d\_tuples$) of the form $\langle x, val \rangle$, where $x$ is the \tobj{} to which transaction $T_i$ writes the value $val$. Similarly, we refer to a tuple in $T_i$'s write-set by $\ws_i[x]$.
\end{itemize}

In addition to these local structures, the following shared global structures are maintained that are shared across transactions (and hence, threads). We name all the shared variable starting with `G'. 

\begin{itemize}
	\item $\gtcnt$ (counter): This a numerical valued counter that is incremented when a transaction begins.

\end{itemize}

\noindent For each transaction $T_i$ we maintain the following shared time-stamps:
\begin{itemize}
	\item $\glock_i$: A lock for accessing all the shared variables of $T_i$.
	
	\item $\gits_i$ (initial timestamp): It is a time-stamp assigned to $T_i$ when it was invoked for the first time.
	
	\item $\gcts_i$ (current timestamp): It is a time-stamp when $T_i$ is invoked again at a later time. When $T_i$ is created for the first time, then its $\gcts$ is same as its $its$.

	\item $\gval_i$: This is a boolean variable which is initially true ($T$). If it becomes false ($F$) then $T_i$ has to be aborted.
	
	\item $\gstat_i$: This is a variable which states the current value of $T_i$. It has three states: \texttt{live}, \texttt{commit} or \texttt{abort}. 
	
\end{itemize}

\noindent For each data item $x$ in history $H$, we maintain:
\begin{itemize}
	\item $x.val$ (value): It is the successful previous closest value written by any transaction.
	
	\item $x.rl$ (readList): It is the read list consists of all the transactions that have read $x$. 
	
\end{itemize}

\begin{algorithm} [H]
	\caption{STM $\init()$: Invoked at the start of the STM system. Initializes all the data items used by the STM System} \label{alg:init} 
	\begin{algorithmic}[1]
		\State $\gtcnt$ = 1;
		\ForAll {data item $x$ used by the STM System}
		
		\State add $\langle 0, nil \rangle$ to $x.val$;\Comment{ $T_0$ is initializing $x$} \label{lin:t0-init} 
		\EndFor;
	\end{algorithmic}
\end{algorithm}

\begin{algorithm} 
	\label{alg:begin} 
	\caption{STM $\begt(its)$: Invoked by a thread to start a new transaction  $T_i$. Thread can pass a parameter $its$ which is the initial timestamp when  this transaction was invoked for the first time. If this is the first invocation then $its$ is $nil$. It returns the tuple $\langle id, \gcts 
		\rangle$}
	\begin{algorithmic}[1]
		\State $i$ = unique-id;
        \Comment{An unique id to identify this transaction. It could be same as $\gcts$.} 
		\If {($its == nil$)} 
		\State $\gits_i = \gcts_i = \gtcnt.get\&Inc()$; 
		\State \Comment{$\gtcnt.get\&Inc()$ returns the current value of $\gtcnt$ and atomically increments it by 1.}
		
		\Else 
		\State $\gits_i = its$;
		\State $\gcts_i = \gtcnt.get\&Inc()$; 
		\EndIf 
		\State $\rs_i = \ws_i = null$;
		\State $\gstat_i$ = \texttt{live};
		\State $\gval_i = T$;
		\State return $\langle i, \gcts_i\rangle$
	\end{algorithmic}
\end{algorithm}

\begin{algorithm}
	\label{alg:read} 
	\caption{STM $read(i, x)$: Invoked by a transaction $T_i$ to read $x$. It returns either the value of $x$ or $\mathcal{A}$}
	\begin{algorithmic}[1]  
		\If {($x \in \ws_i$)} \Comment{Check if $x$ is in $\ws_i$}
				\State return $\ws_i[x].val$;
		\ElsIf {($x \in \rs_i$)} \Comment{Check if $x$ is in $\rs_i$}
				\State return $\rs_i[x].val$;
		\Else \Comment{$x$ is not in $\rs_i$ and $\ws_i$}         
		\State lock $x$; 
		
		\State lock $\glock_i$;
		\If {$(\gval_i == F)$} 
		\State return $abort(i)$; \label{lin:rabort}
		\EndIf    
		
		
		\cmnt
		{
			\State /* \findsl: From $x.\vl$, returns the smallest \ts value greater 
			than $\gwts_i$. If no such version exists, it returns $nil$ */
			\State $nextVer  = \findsl(\gwts_i,x)$;    
			\If {$(nextVer \neq nil)$}
			\State \Comment{Ensure that $\tutl_i$ remains smaller than $nextVer$'s \vltl}
			\State $\tutl_i = min(\tutl_i, x[nextVer].vltl-1)$; 
			\EndIf
			
			\State \Comment{$\tltl_i$ should be greater than $x[curVer].\vltl$}
			\State $\tltl_i = max(\tltl_i, x[curVer].\vltl + 1)$;
			\If {($\tltl_i > \tutl_i$)} \Comment{If the limits have crossed each other, then $T_i$ is aborted}
			\State return abort(i);
			\EndIf    
		}
		
		\State $val = x.val$; 
		\State add $T_i$ to $x.rl$; 
		\State unlock $\glock_i$; 
		\State unlock $x$;
		\State return $val$;
		\EndIf
	\end{algorithmic}
\end{algorithm}

\begin{algorithm}[H]
	\label{alg:write} 
	\caption{STM $write_i(x,val)$: A Transaction $T_i$ writes into local memory}
	\begin{algorithmic}[1]
		\State Append the $d\_tuple \langle x,val \rangle$ to 
		$\ws_i$.\Comment{If same dataitem then overwrite the tuple }
		\State return $ok$;
	\end{algorithmic}
\end{algorithm}

\begin{algorithm}[H]
	\label{alg:llts} 
	\caption{STM $findLLTS(TSet)$:  Find the lowest $its$ value among all the live trasactions in $TSet$.}
	\begin{algorithmic}[1]
		\State $min\_its$  = $\infty$
		\ForAll {( $T_j \in TSet$)}
		
		\If {(($\gits_j$ $< min\_its)$ \&\& $(\gstat_j == \texttt{live})$)} 
 				\State $min\_its$ = $\gits_j$;		
		\EndIf
		\EndFor
		\State return $min\_its$;
		
	\end{algorithmic}
\end{algorithm}
\begin{algorithm}[H]  
	\label{alg:tryc} 
	\caption{STM $\tryc()$: Returns $\mathcal{C}$ on commit else return Abort $\mathcal{A}$}
	\begin{algorithmic}[1]
		\State lock $\glock_i$
		\If {$(\gval_i == F)$} return $abort(i)$;
		\EndIf
		\State $TSet = null$ {} \Comment{$TSet$ storing 
		transaction Ids}
		\ForAll {($x \in \ws_i$)}
		\State lock $x$ in pre-defined order;
		\ForAll {($T_j \in x.rl$)}
		\State $TSet$ = $TSet$ $\cup$ \{$T_j$\}
		\EndFor
			 
		\EndFor \Comment{$x \in \ws_i$}  
		\State $TSet$ = $TSet$ $\cup$ \{$T_i$\} \Comment{Add current transaction $T_i$ into $TSet$}
		\ForAll {( $T_k \in TSet$)}
		\State lock $\glock_k$ in pre-defined order; 
												\Comment{Note: Since $T_i$ is also in 
												$TSet$, $\glock_i$ is also locked}
		\EndFor
		\If {$(\gval_i == F)$} return $abort(i)$;
		\Else
		\If {($\gits_i == findLLTS(TSet$))} \Comment{Check if $T_i$ has lowest $its$ among all \texttt{live} transactions in $TSet$}
		\ForAll {($T_j \in TSet$)}  \Comment{ ($T_i \neq T_j$)} 
		\State $G\_valid_j = F$
		\State unlock $\glock_j$;					
		\EndFor
		\Else
		\State return $abort(i)$;
		\EndIf
		\EndIf
		
		\cmnt {
			\algstore{tryc-break}
		\end{algorithmic}
	\end{algorithm}
	
	\begin{algorithm}  
		\label{alg:tryc-cont} 
		\caption{STM $\tryc()$: Continued}
		
		\begin{algorithmic}[1]
			\algrestore{tryc-break}

			\State \Comment{Ensure that $\vutl_i$ is less than \vltl of versions in $\nvl$} 
			\ForAll {$(ver \in \nvl)$}
			\State $x$ = \tobj of $ver$;
			\State $\tutl_i = min(\tutl_i, x[ver].\vltl - 1)$;
			\EndFor 
		}

\ForAll {($x \in \ws_i$)}
		\State replace the old value in $x.val$ with $newValue$;
		\State $x.rl$ = null;
		\EndFor
		\State $\gstat_i$ = \texttt{commit};
		\State unlock all variables locked by $T_i$;
		\State return $\mathcal{C}$;
		
	\end{algorithmic}
\end{algorithm}

\cmnt {
	\begin{algorithm}  
		\label{alg:lowPri} 
		\caption{$\lowp(T_k, T_i)$: Verifies if $T_k$ has lower priority than $T_i$ and is not already committed}
		\begin{algorithmic}[1]
			
			\If {$(\gits_i < \gits_k) \land (\gstat_k == \texttt{live})$}
			\State return $T$;
			\Else
			\State return $F$;
			\EndIf
		\end{algorithmic}
	\end{algorithm}
	
	\begin{algorithm}  
		\label{alg:\lowp} 
		\caption{$\lowp(T_k, T_i)$: Aborts lower priority transaction among $T_k$ and $T_i$}
		\begin{algorithmic}[1]
			
			\If {$(\gits_i < \gits_k)$}
			\State $\abl = \abl \cup T_k$; \Comment{Store lower priority $T_k$ in \abl}
			\Else
			\State return abort(i); \Comment{Abort $T_i$}
			\EndIf
		\end{algorithmic}
	\end{algorithm}
}

\begin{algorithm}[H] 
	\label{alg:abort} 
	\caption{$abort(i)$: Invoked by various STM methods to abort transaction $T_i$. It returns $\mathcal{A}$}
	\begin{algorithmic}[1]
		\State $\gval_i = F$;
		\State $\gstat_i$ = \texttt{abort};
		\State unlock all variables locked by $T_i$;      
		\State return $\mathcal{A}$;
	\end{algorithmic}
\end{algorithm}
\cmnt{
\vspace{1mm}
\noindent
\textbf{Common Data-Structures \& STM Methods:} Here we describe the common data-structures used by all the algorithms proposed in this section. For each \tobj, the algorithm(s) maintains multiple version as a linked-list, \emph{\vlist}. Similar to versions in \mvto \cite{Kumar+:MVTO:ICDCN:2014}, each version of a \tobj{} is a tuple denoted as \emph{\vtup} and consists of three fields: (1) timestamp, $ts$ of the transaction that created this version which normally is the \cts; (2) the value of the version; (3) a list, called \rlist{}, consisting of transactions ids (could be \cts as well) that read from this version. The \rlist of a version is initially empty. \figref{tobj} illustrates this structure. For a \tobj $x$, we use the notation $x[t]$ to access the version with timestamp $t$. Depending on the algorithm considered, the fields change of this structure. 
\begin{figure} [h]
	\centerline{\scalebox{0.50}{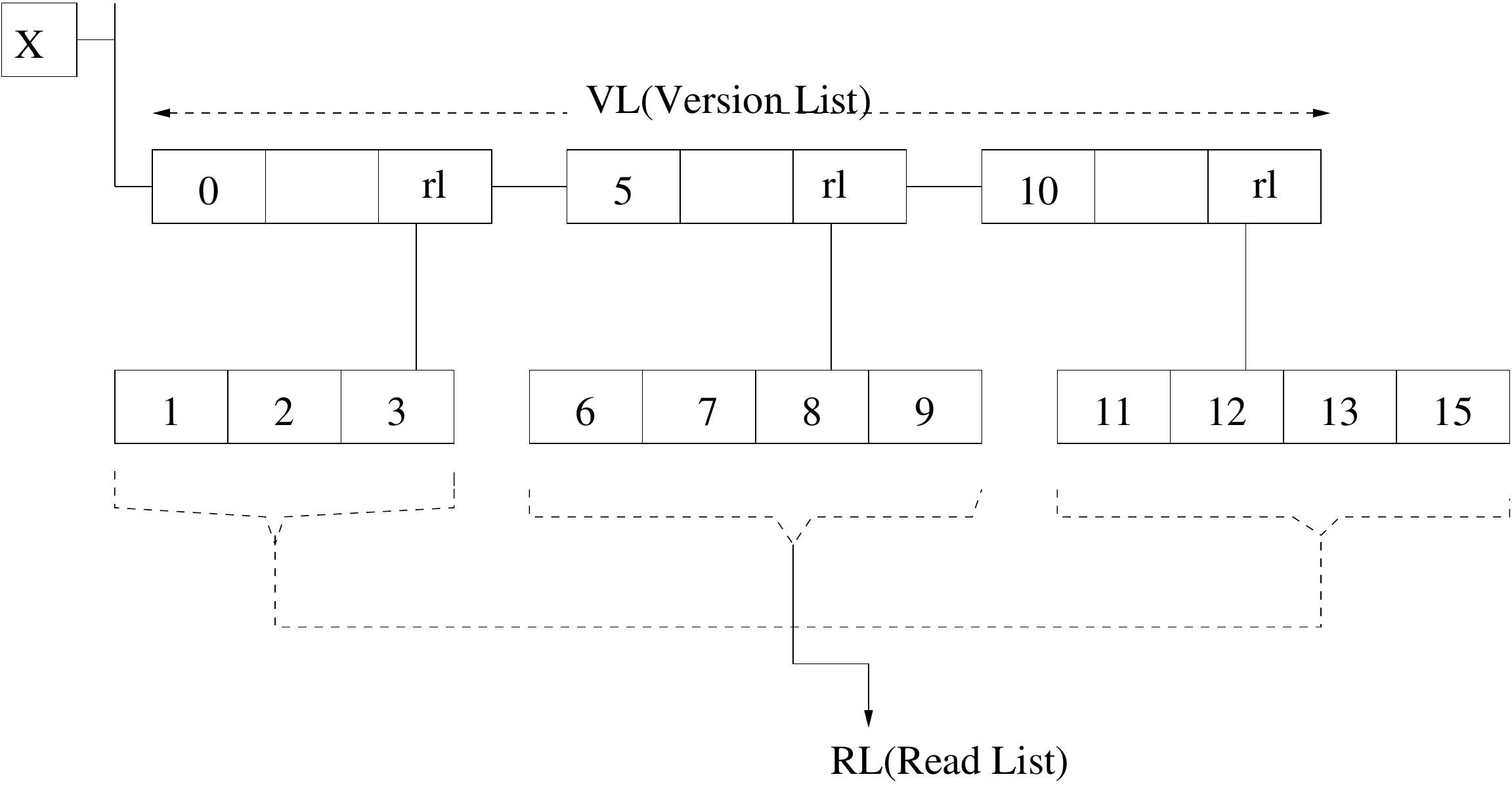}}
	\captionsetup{justification=centering}	
	\caption{Data Structures for Maintaining Versions}
	\label{fig:tobj}
\end{figure}


The algorithms have access to a global atomic counter, $\gtcnt$ used for generating timestamps in the various transactional \mth{s}. We assume that the STM system exports the following \mth{s} for a transaction $T_i$: (1) $\begt(t)$ where $t$ is provided by the invoking thread. From our earlier assumption, it is the \cts of the first \inc. This \mth returns a unique timestamp to the invoking thread which is the \cts/id of the transaction. (2) $\tread_i(x)$ tries to read  \tobj $x$. It returns either value $v$ or $\abort$. (3) $\twrite_i(x,v)$ operation that tries to update a \tobj $x$ with value $v$. It returns $ok$. (4) $\tryc_i()$ that tries to commit the transaction and returns $ok$ if it succeeds. Otherwise, it returns $\abort$. (5) $\trya()$ that aborts the transaction and returns $\abort$.}

\ignore{
In addition to these global structures, for each transaction $T_i$, \pmvto maintains structure that are local to $T_i$:
\begin{itemize}
	\item $\rs_i$(read-set): It is a list of data tuples or \emph{\dtup} of the form $\langle x, val \rangle$, where $x$ is the t-object and $v$ is the value read by the transaction $T_i$. We refer to a tuple in $T_i$'s read-set by $\rs_i[x]$.
	\item $\ws_i$(write-set): It is a list of \dtup{s} where $x$ is the \tobj{} to which transaction $T_i$ writes the value $val$. Similarly, we refer to a tuple in $T_i$'s write-set by $\ws_i[x]$.
\end{itemize}
\noindent Next, we describe the \mth{s} exported by \pmvto. Here in this discussion, we use the notation that a transaction $T_i$ has its \cts as $i$. 
}

\vspace{1mm}
\noindent
\textbf{Simplifying Assumptions:} We next describe the main idea behind the starvation-free STM algorithm \ksftm through a sequence of algorithms. For ease of exposition, we make two simplifying assumptions (1) We assume that in the absence of other concurrent conflicting transactions, every transaction will commit. In other words, if a transaction is executed in a system by itself, it will not self-abort. (2) We initially consider \stsbty as \cc to illustrate the correctness of the algorithms. But \stsbty does not consider the correctness of aborted transactions and as a result not a suitable \cc for STMs. Finally, we show that the proposed STM algorithm \ksftm satisfies \lopty, a \cc for STMs. 

\noindent We denote the set of histories generated by an STM algorithm, say $A$, as $\gen{A}$.

\cmnt {
	A STM system exports the functions:\begtrans{}, $\read_i, \writei_i$ and $\tryci_i$. A thread invokes a transaction with a \begtrans{} function which returns a unique transaction id which is the timestamp of the transactions. This timestamp is numerically greater than the timestamps of all the transactions invoked so far. This thread invokes future functions of the transaction using this timestamp. We use the notation $T_i$ to denote a transaction where $i$ is the timestamp of $T_i$. 
}

\subsection{Priority-based MVTO Algorithm}
\label{subsec:mvto}

In this subsection, we describe a modification to the multi-version timestamp ordering (\mvto) algorithm \cite{BernGood:1983:MCC:TDS,Kumar+:MVTO:ICDCN:2014} to ensure that it provides preference to transactions that have low \its, i.e., transactions that have been in the system for a longer time. We denote the basic algorithm which maintains unbounded versions as \emph{Priority-based MVTO} or \emph{\pmvto} (akin to the original \mvto). We denote the variant of \pmvto that maintains $K$ versions as \emph{\pkto} and the unbounded versions variant with garbage collection as \emph{\pmvtogc}. In this sub-section, we specifically describe \pkto. But most of these properties apply to \pmvto and \pmvtogc as well. 

\ignore{
\color{blue}
\color{red} REMOVE 
We have not described the details of locking of data-items by transactions.
\color{black}
}

\vspace{1mm}
\noindent
\textbf{$\begt(t)$:} A unique timestamp $ts$ is allocated to $T_i$ which is its \cts ($i$ from our assumption). The timestamp $ts$ is generated by atomically incrementing the global counter $\gtcnt$. If the input $t$ is null, then $\tcts{i} = \tits{i} = ts$ as this is the first \inc of this transaction. Otherwise, the non-null value of $t$ is assigned as $\tits{i}$. 

\noindent
\textbf{$\tread(x)$:} Transaction $T_i$ reads from a version of $x$ in the shared memory (if $x$ does not exist in $T_i$'s local buffer) with timestamp $j$ such that $j$ is the largest timestamp less than $i$ (among the versions $x$), i.e., there exists no version of $x$ with timestamp $k$ such that $j<k<i$. After reading this version of $x$, $T_i$ is stored in $x[j]$'s \rlist. If no such version exists then $T_i$ is \emph{aborted}. 

\noindent
\textbf{$\twrite(x,v)$:} $T_i$ stores this write to value $x$ locally in its $\ws_i$. If $T_i$ ever reads $x$ again, this value will be returned.

\noindent
\textbf{$\tryc:$} This \op{} consists of three steps. In \stref{val}, it checks whether $T_i$ can be \emph{committed}. In \stref{updt}, it performs the necessary tasks to mark $T_i$ as a \emph{committed} transaction and in \stref{commit}, $T_i$ return commits.  

\begin{enumerate}
	\item Before $T_i$ can commit, it needs to verify that any version it creates does not violate consistency. Suppose $T_i$ creates a new version of $x$ with timestamp $i$. Let $j$ be the largest timestamp smaller than $i$ for which version of $x$ exists. Let this version be $x[j]$. Now, $T_i$ needs to make sure that any transaction that has read $x[j]$ is not affected by the new version created by $T_i$. There are two possibilities of concern: \label{step:val}

  
	\begin{enumerate}		        
		\item Let $T_k$ be some transaction that has read $x[j]$ and $k > i$ ($k$ = \cts of $T_k$). In this scenario, the value read by $T_k$ would be incorrect (w.r.t \stsbty) if $T_i$ is allowed to create a new version. In this case, we say that the transactions $T_i$ and $T_k$ are in \emph{conflict}. So, we do the following: \\(i) if $T_k$ has already \emph{committed} then $T_i$ is \emph{aborted}; \\(ii) if $T_k$ is live and $\tits{k}$ is less than $\tits{i}$. Then again $T_i$ is \emph{aborted}; \\(iii) If $T_k$ is still live with $its_i$ less than $its_k$ then $T_k$ is \emph{aborted}. \label{step:verify}		
		
		\item The previous version $x[j]$ does not exist. This happens when the previous version $x[j]$ has been overwritten. In this case, $T_i$ is \emph{aborted} since \pkto does not know if $T_i$ conflicts with any other transaction $T_k$ that has read the previous version. \label{step:notfound}        	        		        
	\end{enumerate}   
	

	\item After \stref{val}, we have verified that it is ok for $T_i$ to commit. Now, we have to create a version of each \tobj $x$ in the $\ws$ of $T_i$. This is achieved as follows: \label{step:updt}
	\begin{enumerate}
		\item $T_i$ creates a $\vtup$ $\langle i, \wset{i}.x.v, null \rangle$. In this tuple, $i$ (\cts of $T_i$) is the timestamp of the new version; $\wset{i}.x.v$ is the value of $x$ is in $T_i$'s $\ws$, and the \rlist of the $\vtup$ is $null$.
		\item Suppose the total number of versions of $x$ is $K$. Then among all the versions of $x$, $T_i$ replaces the version with the smallest timestamp with $\vtup$ $\langle i, \wset{i}.x.v, null \rangle$. Otherwise, the $\vtup$ is added to $x$'s $\vlist$. 
	\end{enumerate}
	\item Transaction $T_i$ is then \emph{committed}. \label{step:commit}
\end{enumerate}

The algorithm described here is only the main idea. The actual implementation will use locks to ensure that each of these \mth{s} are \lble \cite{HerlWing:1990:TPLS}. It can be seen that \pkto gives preference to the transaction having lower \its in \stref{verify}. Transactions having lower \its have been in the system for a longer time. Hence, \pkto gives preference to them. 

\subsection{Pseudocode of \pkto}
\label{apn:pcode}
\begin{algorithm}[H] 
\caption{STM $\init()$: Invoked at the start of the STM system. Initializes all the \tobj{s} used by the STM System} \label{alg:init} 
\begin{algorithmic}[1]
\State $\gtcnt$ = 1;
\ForAll {$x$ in $\mathcal{T}$} \Comment{All the \tobj{s} used by the STM System}
 
  \State add $\langle 0, 0, nil \rangle$ to $x.\vl$;  \Comment { $T_0$ is initializing $x$} \label{lin:t0-init} 
\EndFor;
\end{algorithmic}
\end{algorithm}

\begin{algorithm}[H] 
\label{alg:ap-begin} 
\caption{STM $\begt(its)$: Invoked by a thread to start a new transaction 
$T_i$. Thread can pass a parameter $its$ which is the initial timestamp when 
this transaction was invoked for the first time. If this is the first 
invocation then $its$ is $nil$. It returns the tuple $\langle id, \gcts 
\rangle$}
\begin{algorithmic}[1]
  \State $i$ = unique-id; \Comment{An unique id to identify this transaction. It could be same as \gcts}
  \State \Comment{Initialize transaction specific local and global variables}  
  \If {($its == nil$)}
    \State \Comment{$\gtcnt.get\&Inc()$ returns the current value of \gtcnt and atomically increments it}   
    \State $\gits_i = \gcts_i = \gtcnt.get\&Inc()$; 
  \Else 
    \State $\gits_i = its$;
    \State $\gcts_i = \gtcnt.get\&Inc()$; 
  \EndIf 
  \State $\rs_i = \ws_i = null$;
  \State $\gstat_i$ = \texttt{live}; 
  \State $\gval_i = T$;
    \State return $\langle i, \gcts_i\rangle$
\end{algorithmic}
\end{algorithm}

\begin{algorithm}[H] 
\label{alg:ap-read} 
\caption{STM $read(i, x)$: Invoked by a transaction $T_i$ to read \tobj{} $x$. It returns either the value of $x$ or $\mathcal{A}$}
\begin{algorithmic}[1]  
  \If {($x \in \rs_i$)} \Comment{Check if the \tobj{} $x$ is in $\rs_i$}
    \State return $\rs_i[x].val$;
  \ElsIf {($x \in \ws_i$)} \Comment{Check if the \tobj{} $x$ is in $\ws_i$}
    \State return $\ws_i[x].val$;
  \Else \Comment{\tobj{} $x$ is not in $\rs_i$ and $\ws_i$}         
    \State lock $x$; lock $\glock_i$;
    \If {$(\gval_i == F)$} return abort(i); \label{lin:rabort}
    \EndIf    
  
    \State \Comment{ \findls: From $x.\vl$, returns the largest \ts value less than 
    $\gcts_i$. If no such version exists, it returns $nil$ }
    \State $curVer  = \findls(\gcts_i,x)$;

    \If {$(curVer == nil)$} return abort(i); \Comment{Proceed only if $curVer$ is not nil}
    \EndIf
    \cmnt
    {
    \State /* \findsl: From $x.\vl$, returns the smallest \ts value greater 
    than $\gwts_i$. If no such version exists, it returns $nil$ */
    \State $nextVer  = \findsl(\gwts_i,x)$;    
    \If {$(nextVer \neq nil)$}
      \State \Comment{Ensure that $\tutl_i$ remains smaller than $nextVer$'s \vltl}
      \State $\tutl_i = min(\tutl_i, x[nextVer].vltl-1)$; 
    \EndIf
    
    \State \Comment{$\tltl_i$ should be greater than $x[curVer].\vltl$}
    \State $\tltl_i = max(\tltl_i, x[curVer].\vltl + 1)$;
    \If {($\tltl_i > \tutl_i$)} \Comment{If the limits have crossed each other, then $T_i$ is aborted}
      \State return abort(i);
    \EndIf    
    }
             
    \State $val = x[curVer].v$; add $\langle x, val \rangle$ to $\rs_i$;
    \State add $T_i$ to $x[curVer].rl$; 
    \State unlock $\glock_i$; unlock $x$;
    \State return $val$;
  \EndIf
\end{algorithmic}
\end{algorithm}

\begin{algorithm}[H]  
\label{alg:ap-write} 
\caption{STM $write_i(x,val)$: A Transaction $T_i$ writes into local memory}
\begin{algorithmic}[1]
\State Append the $d\_tuple \langle x,val \rangle$ to $\ws_i$.
\State return $ok$;
\end{algorithmic}
\end{algorithm}

\begin{algorithm}[H]  
\label{alg:ap-tryc} 
\caption{STM $\tryc()$: Returns $ok$ on commit else return Abort}
\begin{algorithmic}[1]
\State \Comment{The following check is an optimization which needs to be performed again later}
\State lock $\glock_i$;
\If {$(\gval_i == F)$} 
\State return abort(i);
\EndIf
\State unlock $\glock_i$;

\State $\lrl = \allrl = nil$; \Comment{Initialize larger read list (\lrl), all read list (\allrl) to nil} 
\ForAll {$x \in \ws_i$}
  \State lock $x$ in pre-defined order;
  \State \Comment{ \findls: returns the version with the largest \ts value less than 
  $\gcts_i$. If no such version exists, it returns $nil$. }
  \State $\prevv = \findls(\gcts_i, x)$; \Comment{\prevv: largest version 
  smaller than $\gcts_i$}
  \If {$(\prevv == nil)$} \Comment{There exists no version with \ts value less 
  than $\gcts_i$}
   \State lock $\glock_i$; return abort(i); 
  \EndIf
  
  \State \Comment{\textbf{\getl}: obtain the list of reading transactions of 
  $x[\prevv].rl$ whose $\gcts$ is greater than $\gcts_i$}
  \State $\lrl = \lrl \cup \getl(\gcts_i, x[\prevv].rl)$;

\EndFor \Comment{$x \in \ws_i$}  

\State $\relll = \lrl \cup T_i$; \Comment{Initialize relevant Lock List 
(\relll)}

\ForAll {($T_k \in \relll$)}
  \State lock $\glock_k$ in pre-defined order; \Comment{Note: Since $T_i$ is also in $\relll$, $\glock_i$ is also locked}
\EndFor

\State \Comment{Verify if $\gval_i$ is false}

\If {$(\gval_i == F)$} 
 \State return abort(i);
\EndIf

\State $\abl = nil$ \Comment{Initialize abort read list (\abl)}

\State \Comment{Among the transactions in $T_k$ in $\lrl$, either $T_k$ or $T_i$ has to be aborted}

\ForAll {$(T_k \in \lrl)$}  
  \If {$(\isab(T_k))$} 
     \Comment{Transaction $T_k$ can be ignored since it is already aborted or about to be aborted}
    \State continue;
  \EndIf

  \If {$(\gits_i < \gits_k) \land (\gstat_k == \texttt{live})$}   
    \State \Comment{Transaction $T_k$ has lower priority and is not yet committed. So it needs to be aborted}
    \State $\abl = \abl \cup T_k$; \Comment{Store $T_k$ in \abl}
\Else \Comment{Transaction $T_i$ has to be aborted}
    \State return abort(i);
  \EndIf      
\EndFor

 \algstore{myalgtc}
\end{algorithmic}
\end{algorithm}

\begin{algorithm}
\begin{algorithmic}[1]
\algrestore{myalgtc}
\cmnt {
\algstore{tryc-break}
\end{algorithmic}
\end{algorithm}

\begin{algorithm}  
\label{alg:ap-tryc-cont} 
\caption{STM $\tryc()$: Continued}

\begin{algorithmic}[1]
\algrestore{tryc-break}

\State \Comment{Ensure that $\vutl_i$ is less than \vltl of versions in $\nvl$} 
\ForAll {$(ver \in \nvl)$}
  \State $x$ = \tobj of $ver$;
  \State $\tutl_i = min(\tutl_i, x[ver].\vltl - 1)$;
\EndFor 
}

\State \Comment{Store the current value of the global counter as commit time and increment it}
\State $\ct = \gtcnt.get\&Inc()$; 

\cmnt{
\ForAll {$(T_k \in \srl)$}  \Comment{Iterate through $\srl$ to see if $T_k$ or $T_i$ has to aborted}
  \If {$(\isab(T_k))$} 
    \State \Comment{Transaction $T_k$ can be ignored since it is already aborted or about to be aborted}
    \State continue;
  \EndIf

  \If {$(\tltl_k \geq \tutl_i)$} \label{lin:tk-check} \Comment{Ensure that the limits do not cross for both $T_i$ and $T_k$}
    \If {$(\gstat_k == live)$}  \Comment{Check if $T_k$ is live}        
      \If {$(\gits_i < \gits_k)$}
        \State \Comment{Transaction $T_k$ has lower priority and is not yet committed. So it needs to be aborted}
        \State $\abl = \abl \cup T_k$; \Comment{Store $T_k$ in \abl}      
      \Else \Comment{Transaction $T_i$ has to be aborted}
        \State return abort(i);
      \EndIf \Comment{$(\gits_i < \gits_k)$}    
    \Else \Comment{($T_k$ is committed. Hence, $T_i$ has to be aborted)}   
      \State return abort(i);
    \EndIf \Comment{$(\gstat_k == live)$}
  \EndIf \Comment{$(\tltl_k \geq \tutl_i)$}
\EndFor {$(T_k \in \srl)$}

\State \Comment{At this point $T_i$ can't abort.}
\State $\tltl_i = \tutl_i$;  \label{lin:ti-updt} 
\State \Comment{Since $T_i$ can't abort, we can update $T_k$'s \tutl}
\ForAll {$(T_k \in \srl)$}  
  \If {$(\isab(T_k))$} 
    \State \Comment{Transaction $T_k$ can be ignored since it is already aborted or about to be aborted}
    \State continue;
  \EndIf
  
  \State /* The following line ensure that $\tltl_k \leq \tutl_k < \tltl_i$. Note that this does not cause the limits of $T_k$ to cross each other because of the check in \Lineref{tk-check}.*/  
  \State $\tutl_k = min(\tutl_k, \tltl_i - 1)$;
\EndFor
}

\ForAll {$T_k \in \abl$} \Comment{Abort all the transactions in \abl}
  \State $\gval_k =  F$;
\EndFor

\cmnt
{
\algstore{tryc-break2}

\end{algorithmic}
\end{algorithm}

\begin{algorithm}[H]  
\label{alg:tryc-cont2} 
\caption{STM $\tryc()$: Continued Again}
\begin{algorithmic}[1]
\algrestore{tryc-break2}
}

\State \Comment{Having completed all the checks, $T_i$ can be committed}
\ForAll {$(x \in \ws_i)$} 
  
  \State $newTuple = \langle \gcts_i, \ws_i[x].val, nil \rangle$; \Comment { Create new v\_tuple: \gcts, val, \rl for $x$}
  \If {($|x.vl| > k$)}
    \State replace the oldest tuple in $x.\vl$ with $newTuple$; \Comment{$x.\vl$ is ordered by timestamp}
  \Else
    \State add a $newTuple$ to $x.vl$ in sorted order;
  \EndIf
\EndFor \Comment{$x \in \ws_i$}

\State $\gstat_i$ = \texttt{commit};
\State unlock all variables;
\State return $\mathcal{C}$;
 
\end{algorithmic}
\end{algorithm}

\begin{algorithm}[H]  
\label{alg:isab} 
\caption{$\isab(T_k)$: Verifies if $T_i$ is already aborted or its \gval flag is set to false implying that $T_i$ will be aborted soon}
\begin{algorithmic}[1]

\If {$(\gval_k == F) \lor (\gstat_k == \texttt{abort}) \lor (T_k \in \abl)$} 
  \State return $T$;
\Else
  \State return $F$;
\EndIf
\end{algorithmic}
\end{algorithm}

\begin{algorithm}[H] 
\label{alg:ap-abort} 
\caption{$abort(i)$: Invoked by various STM methods to abort transaction $T_i$. It returns $\mathcal{A}$}
\begin{algorithmic}[1]
\State $\gval_i = F$; $\gstat_i$ = \texttt{abort};
\State unlock all variables locked by $T_i$;      
\State return $\mathcal{A}$;
\end{algorithmic}
\end{algorithm}

We have the following property on the correctness of \pkto. 


\begin{property}
\label{prop:pmvto-correct} 
Any history generated by \pkto is strict-serializable.
\end{property}
Consider a history $H$ generated by \pkto. Let the \emph{committed} \ssch of $H$ be $CSH = \shist{\comm{H}}{H}$. It can be shown that $CSH$ is \opq with the equivalent serialized history $SH'$ is one in which all the transactions of $CSH$ are ordered by their \cts{s}. Hence, $H$ is \stsble.


\cmnt{
	The read \op{} described in \stref{read} is a common idea used in \mvstm{s} which are based on \ts{s} such as \mvto{} \cite{Kumar+:MVTO:ICDCN:2014}, Generic Multi-Version STM \cite{LuScott:GMV:DISC:2013} etc. Unlike these STMs which have infinite versions, a read \op{} in \pmvto{} can abort. Similarly, the $\tryc{}$ \op{} described in \stref{tryc} is very similar to \mvto, but modified for finite versions. 
}


\noindent
\textbf{Possibility of Starvation in \pkto:}
As discussed above, \pkto gives priority to transactions having lower \its. But a transaction $T_i$ having the lowest \its could still abort due to one of the following reasons: (1) Upon executing $\tread(x)$ \mth if it does not find any other version of $x$ to read from. This can happen if all the versions of $x$ present have a timestamp greater than $\tcts{i}$. (2) While executing \stref{verify}(i), of the $\tryc$ \mth, if $T_i$ wishes to create a version of $x$ with timestamp $i$. But some other transaction, say $T_k$ has read from a version with timestamp $j$ and $j<i<k$. In this case, $T_i$ has to abort if $T_k$ has already \emph{committed}. 

This issue is not restricted only to \pkto. It can occur in \pmvto (and \pmvtogc) due to the point (2) described above.  

\begin{figure}[H]
	\centering
	\scalebox{0.5}{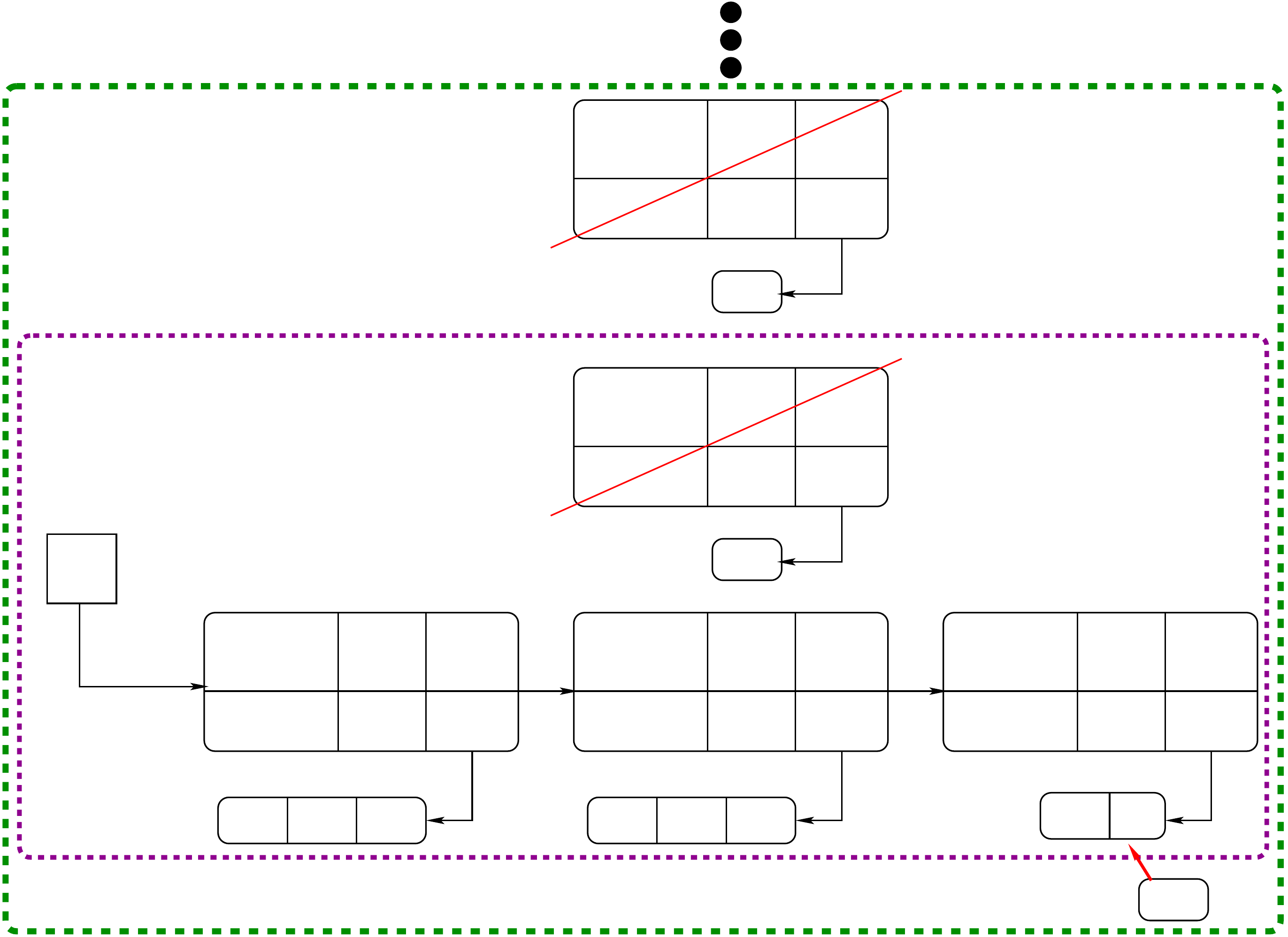}
	\captionsetup{justification=centering}
	\caption{Pictorial representation of execution under \pkto}
	\label{fig:kstmvl}
\end{figure}
We illustrate this problem in \pkto with \figref{kstmvl}. Here transaction $T_{26}$, with \its 26 is the lowest among all the live transactions, starves due to \stref{verify}.(i) of the $\tryc$. \emph{First time}, $T_{26}$ gets \emph{aborted}  due to higher timestamp transaction $T_{29}$ in the \rlist of $x[25]$ has \emph{committed}. We have denoted it by a `(C)' next to the version. The \emph{second time}, $T_{26}$ retries with same \its 26 but new \cts 33. Now when $T_{33}$ comes for commit, suppose another transaction $T_{34}$ in the \rlist of $x[25]$ has already \emph{committed}. So this will cause $T_{33}$ (another incarnation of $T_{26}$) to abort again. Such scenario can possibly repeat again and again and thus causing no \inc of  $T_{26}$ to ever commit leading to its starvation. 

\noindent
\textbf{Garbage Collection in \mvsftmgc and \pmvtogc:} Having multiple versions to increase the performance and to decrease the number of aborts, leads to creating too many versions which are not of any use and hence occupying space. So, such garbage versions need to be taken care of. Hence we come up with a garbage collection over these unwanted versions. This technique help to conserve memory space and increases the performance in turn as no more unnecessary traversing of garbage versions by transactions is necessary. We have used a global, i.e., across all transactions a list that keeps track of all the live transactions in the system. We call this list as \livel. Each transaction at the beginning of its life cycle creates its entry in this \livel. Under the optimistic approach of STM, each transaction in the shared memory performs its updates in the $\tryc$ phase. In this phase, each transaction performs some validations, and if all the validations are successful then the transaction make changes or in simple terms creates versions of the corresponding \tobj{} in the shared memory. While creating a version every transaction, check if it is the least timestamp live transaction present in the system by using \livel{} data structure, if yes then the current transaction deletes all the version of that \tobj{} and create one of its own. Else the transaction does not do any garbage collection or delete any version and look for creating a new version of next \tobj{} in the write set, if at all. \figref{ksftm} and \figref{pkto} show that both \mvsftmgc and \pmvtogc performs better than \mvsftm and \pmvto across all workloads.

\ignore{
\begin{enumerate}
	\item \textbf{read rule:} $T_i$ on invoking $r_i(x)$ reads the value $v$, where $v$ is the value written by a transaction $T_j$ that commits before $r_i(x)$ and $j$ is the largest timestamp $\leq i$.
	\item \textbf{write rule:} $T_i$ writes into local memory.
	\item \textbf{commit rule:} $T_i$ on invoking \tryc{} \op{} checks for each \tobj{} $x$, in its \Wset:
	\begin{enumerate}
		\item If a transaction $T_k$ has read $x$ from $T_j$, i.e. $r_k(x, v) \in \evts{T_k}$ and $w_j(x, v) \in \evts{T_j}$ and $j < i < k$, then $\tryc_i$ returns abort, 
		\item otherwise, the transaction is allowed to commit.
	\end{enumerate}
\end{enumerate}
}
\subsection{Modifying \pkto to Obtain \sfkv: Trading Correctness for \emph{Starvation-Freedom}}
\label{subsec:sfmvto}

Our goal is to revise \pkto algorithm to ensure that \emph{\stfdm} is satisfied. Specifically, we want the transaction with the lowest \its to eventually commit. Once this happens, the next non-committed transaction with the lowest \its will commit. Thus, from induction, we can see that every transaction will eventually commit. 

\noindent
\textbf{Key Insights For Eliminating Starvation in \pkto:} To identify the necessary revision, we first focus on the effect of this algorithm on two transactions, say $T_{50}$ and $T_{60}$ with their \cts values being 50 and 60 respectively. Furthermore, for the sake of discussion, assume that these transactions only read and write \tobj $x$. Also, assume that the latest version for $x$ is with $ts$ $40$. Each transaction first reads $x$ and then writes $x$ (as part of the $\tryc{}$ operation). We use $r_{50}$ and $r_{60}$ to denote their read operations while $w_{50}$ and $w_{60}$ to denote their $\tryc$ \op{s}. Here, a read operation will not fail as there is a previous version present.  


Now, there are six possible permutations of these statements. We identify these permutations and the action that should be taken for that permutation in Table \ref{tbl:sfillus}. In all these permutations, the read \op{s} of a transaction come before the write \op{s} as the writes to the shared memory occurs only in the $\tryc$ \op (due to optimistic execution) which is the final \op of a transaction. 


\begin{table}[ht] \centering
	\begin{tabular}{|c|l|l|}
		\hline
		\multicolumn{1}{|c|}{S. No} & \multicolumn{1}{c|}{Sequence} & \multicolumn{1}{c|}{Action} \\
		\hline
		1. & $r_{50}, w_{50}, r_{60}, w_{60}$ & $T_{60}$ reads the version written by $T_{50}$. No conflict. \\
		\hline
		2. & $r_{50}, r_{60}, w_{50}, w_{60}$ & Conflict detected at $w_{50}$. Either abort $T_{50}$ or $T_{60}$. \\
		\hline
		3. & $r_{50}, r_{60}, w_{60}, w_{50}$ & Conflict detected at $w_{50}$. Hence, abort $T_{50}$. \\
		\hline
		4. & $r_{60}, r_{50}, w_{60}, w_{50}$ & Conflict detected at $w_{50}$. Hence, abort $T_{50}$. \\
		\hline
		5. & $r_{60}, r_{50}, w_{50}, w_{60}$ & Conflict detected at $w_{50}$. Either abort $T_{50}$ or $T_{60}$. \\
		\hline
		6. & $r_{60}, w_{60}, r_{50}, w_{50}$ & Conflict detected at $w_{50}$. Hence, abort $T_{50}$.\\
		\hline
	\end{tabular}
	\caption{Permutations of operations}
	\label{tbl:sfillus}
\end{table}

From this table, it can be seen that when a conflict is detected, in some cases, algorithm \pkto \textit{must} abort $T_{50}$. In case both the transactions are live, \pkto has the option of aborting either transaction depending on their \its. If $T_{60}$ has lower \its then in no case, \pkto is required to abort $T_{60}$. In other words, it is possible to ensure that the transaction with lowest \its and the highest \cts is never aborted. Although in this example, we considered only one \tobj, this logic can be extended to cases having multiple \op{s} and \tobj{s}. 

Next, consider \stref{notfound} of \pkto algorithm. Suppose a transaction $T_i$ wants to read a \tobj but does not find a version with a timestamp smaller than $i$. In this case, $T_i$ has to abort. But if $T_i$ has the highest \cts, then it will certainly find a version to read from. This is because the timestamp of a version corresponds to the timestamp of the transaction that created it. If $T_i$ has the highest \cts value then it implies that all versions of all the \tobj{s} have a timestamp smaller than \cts of $T_i$. This reinforces the above observation that a transaction with lowest \its and highest \cts is not aborted.


To summarize the discussion, algorithm $\pkto$ has an in-built mechanism to protect transactions with lowest \its and highest \cts value. However, this is different from what we need. Specifically, we want to protect a transaction $T_i$, with lowest $\its$ value. One way to ensure this: if transaction $T_i$ with lowest \its keeps getting aborted, eventually it will achieve the highest \cts. Once this happens, \pkto  ensures that $T_i$ cannot be further aborted. In this way, we can ensure the liveness of all transactions. 


\ignore{

\color{blue}
I propose to rename 3.2 as Key Insights For Eliminating Starvation in PMVTO

Create a new section here with title
Modifying PMVTO to Obtain SFMVTO: Trading Correctness for Starvation-freedom

Change title of 3.3 (new 3.4)  to 
Design of KFTM: Regaining Correctness while Preserving Starvation Freedom

\color{black}
}

\noindent
\textbf{The working of \emph{\stf}  algorithm:} To realize this idea and achieve \emph{\stfdm}, we consider another variation of \mvto, \emph{Starvation-Free MVTO} or \emph{\sfmv}. We specifically consider \sfmv with $K$ versions, denoted as \emph{\sfkv}. 

A transaction $T_i$ instead of using the current time as $\tcts{i}$, uses a potentially higher timestamp, \emph{Working Timestamp - \wts} or $\twts{i}$. Specifically, it adds $C * (\tcts{i} - \tits{i})$ to $\tcts{i}$, i.e., 
\begin{equation}
\label{eq:wtsf}
\twts{i} = \tcts{i} + C * (\tcts{i} - \tits{i});
\end{equation}
where, $C$ is any constant greater than 0. In other words, when the transaction $T_i$ is issued for the first time, $\twts{i}$ is same as $\tcts{i}(= \tits{i})$. However, as transaction keeps getting aborted, the drift between $\tcts{i}$ and $\twts{i}$ increases. The value of $\twts{i}$ increases with each retry. 

Furthermore, in \sfkv algorithm, \cts is replaced with \wts for $\tread$, $\twrite$ and $\tryc$ \op{s} of \pkto. In \sfkv, a transaction $T_i$ uses $\twts{i}$ to read a version in $\tread$. Similarly, $T_i$ uses $\twts{i}$ in $\tryc$ to find the appropriate previous version (in \stref{notfound}) and to verify if $T_i$ has to be aborted (in \stref{verify}). Along the same lines, once $T_i$ decides to commit and create new versions of $x$, the timestamp of $x$ will be same as its $\twts{i}$ (in \stref{commit}). Thus the timestamp of all the versions in $\vlist$ will be \wts of the transactions that created them. 


\noindent Now, we have the following property about \sfkv algorithm. 

\begin{property}
\label{prop:sfmv-live}
\sfkv algorithm ensures \stfdm. 
\end{property}
While the proof of this property is somewhat involved, the key idea is that the transaction with lowest \its value, say $T_{low}$, will eventually have highest \wts value than all the other transactions in the system. Moreover, after a certain duration, any \textit{new} transaction arriving in the system (i.e., whose $\its$ value sufficiently higher than that of $T_{low}$) will have a lower $\wts$ value than $T_{low}$. This will ensure that $T_{low}$ will not be aborted. In fact, this property can be shown to be true of \sfmv as well. 

\noindent
\textbf{The drawback of \sfkv:} Although \sfkv satisfies starvation-freedom, it, unfortunately, does not satisfy strict-serializability. Specifically, it violates the real-time requirement. \pkto uses \cts for its working while \sfkv uses \wts. It can be seen that \cts is close to the \rt execution of transactions whereas \wts of a transaction $T_i$ is artificially inflated based on its \its and might be much larger than its \cts. 
\begin{figure}
	\centerline{
		\scalebox{0.6}{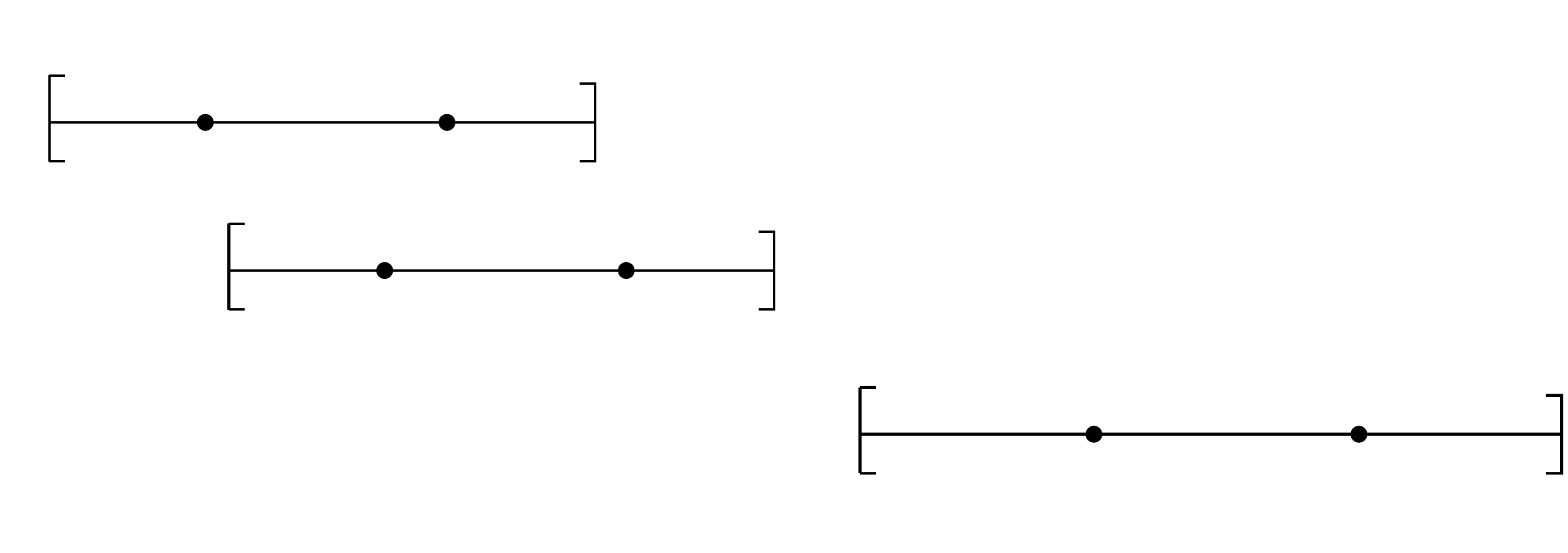}}
	\captionsetup{justification=centering}
	\caption{Correctness of \sfkv Algorithm}
	\label{fig:sfmv-correct}
\end{figure}
We illustrate this with an example. Consider the history $H1$ as shown in \figref{sfmv-correct}: $r_1(x,0) r_2(y,0) w_1(x, 10)\\ C_1 w_2(x, 20) C_2 r_3(x, 10) r_3(z, 25) C_3$ with \cts as 50, 60 and 80 and \wts as 50, 100 and 80 for $T_1, T_2, T_3$ respectively. Here $T_1, T_2$ are ordered before $T_3$ in \rt with $T_1 \prec_{H1}^{RT} T_3$ and $T_2 \prec_{H1}^{RT} T_3$ although $T_2$ has a higher \wts than $T_3$. 

Here, as per \sfkv algorithm, $T_3$ reads $x$ from $T_1$ since $T_1$ has the largest \wts (50) smaller than $T_3$'s \wts (80). It can be verified that it is possible for \sfkv to generate such a history. But this history is not \stsble. The only possible serial order equivalent to $H1$ and \legal is $T_1 T_3 T_2$. But this violates \rt order as $T_3$ is serialized before $T_2$ but in $H1$, $T_2$ completes before $T_3$ has begun. Since $H1$ is not \stsble, it is not \lopq as well. Naturally, this drawback extends to \sfmv as well.


\subsection{Design of \ksftm: Regaining Correctness while Preserving \emph{Starvation-Freedom}}
\label{subsec:ksftm}

In this section, we discuss how principles of \pkto and \sfkv can be combined to obtain \ksftm that provides both correctness (strict-serializability and \lopq) as well as \emph{\stfdm}.
To achieve this, we first understand why the initial algorithm, \pkto satisfies strict-serializability. This is because \cts was used to create the ordering among committed transactions. \cts is closely associated with real-time. In contrast, \sfkv uses \wts which may not correspond to the real-time, as \wts may be significantly larger than \cts as shown by $H1$ in \figref{sfmv-correct}. 


One straightforward way to modify \sfkv is to delay a committing transaction, say $T_i$ with \wts value $\twts{i}$ until the real-time (\gtcnt) catches up to $\twts{i}$. This will ensure that value of \wts will also become same as the real-time thereby guaranteeing \stsbty. However, this is unacceptable, as in practice, it would require transaction $T_i$ locking all the variables it plans to update and wait. This will adversely affect the performance of the STM system. 

We can allow the transaction $T_i$ to commit before its $\twts{i}$ has caught up with the actual time if it does not violate the \rt ordering. Thus, to ensure that the notion of \rt order is respected by transactions in the course of their execution in \sfkv, we add extra time constraints. We use the idea of timestamp ranges. This notion of timestamp ranges was first used by Riegel et al. \cite{Riegel+:LSA:DISC:2006} in the context of multi-version STMs. Several other researchers have used this idea since then such as Guerraoui et al. \cite{Guer+:disc:2008}, Crain et al. \cite{Crain+:RI_VWC:ICA3PP:2011},  Aydonat \& Abdelrahman \cite{AydAbd:RCC:TPDS:2012}.

\ignore{





Towards, this end, first, we understand why the initial algorithm, \pkto is correct.  As mentioned in \propref{pmvto-correct}, any history generated by it is \stsble with the equivalent serial history being one in which transactions are ordered by their \cts{s}. We wanted to achieve the same principle with \sfkv using \wts{s}. But the serial order of \wts{s} of transactions does not respect \rt order as shown by $H1$ in \figref{sfmv-correct}. 

Thus, to ensure that the notion of \rt order is respected by transactions in course of their execution in \sfkv, we add extra time constraints. We use the idea of timestamp ranges. This notion of timestamp ranges was first used by Riegel et al. \cite{Riegel+:LSA:DISC:2006} in the context of multi-version STMs. Several other researchers have used this idea since then such as Guerraoui et al. \cite{Guer+:disc:2008}, Crain et al. \cite{Crain+:RI_VWC:ICA3PP:2011}, Aydonat \& Abdelrahman \cite{AydAbd:RCC:TPDS:2012}.

}

Thus, in addition to \its, \cts and \wts, each transaction $T_i$ maintains a timestamp range: \emph{Transaction Lower Timestamp Limit} or $\ttltl{i}$, and \emph{Transaction Upper Timestamp Limit} or $\ttutl{i}$. When a transaction $T_i$ begins, $\ttltl{i}$ is assigned $\tcts{i}$ and $\ttutl{i}$ is assigned a largest possible value which we denote as infinity. When $T_i$ executes a \mth $m$ in which it reads a version of a \tobj $x$ or creates a new version of $x$ in $\tryc$, $\ttltl{i}$ is incremented while $\ttutl{i}$ gets decremented \footnote{Technically $\infty$, which is assigned to $\ttutl{i}$, cannot be decremented. But here as mentioned earlier, we use $\infty$ to denote the largest possible value that can be represented in a system.}. 

We require to serialize all the transactions based on their \wts while maintaining their \rt order. On executing $m$, $T_i$ is ordered w.r.t to other transactions that have created a version of $x$ based on increasing order of \wts. For all transactions $T_j$ which also have created a version of $x$ and whose $\twts{j}$ is less than $\twts{i}$, $\ttltl{i}$ is incremented such that $\ttutl{j}$ is less than $\ttltl{i}$. Note that all such $T_j$ are serialized before $T_i$.  Similarly, for any transaction $T_k$ which has created a version of $x$ and whose $\twts{k}$ is greater than $\twts{i}$, $\ttutl{i}$ is decremented such that it becomes less than $\ttltl{k}$. Again, note that all such $T_k$ is serialized after $T_i$. 

Note that in the above discussion, $T_i$ need not have created a version of $x$. It could also have read the version of $x$ created by $T_j$. After the increments of $\ttltl{i}$ and the decrements of $\ttutl{i}$, if $\ttltl{i}$ turns out to be greater than $\ttutl{i}$ then $T_i$ is aborted. Intuitively, this implies that $T_i$'s \wts and \rt orders are out of \emph{sync} and cannot be reconciled. 


Finally, when a transaction $T_i$ commits: (1) $T_i$ records its commit time (or $\ct_i$) by getting the current value of \gtcnt and incrementing it by $\incv$ which is any value greater than or equal to 1. Then $\ttutl{i}$ is set to $\ct_i$ if it is not already less than it. Now suppose $T_i$ occurs in \rt before some other transaction, $T_k$ but does not have any conflict with it. This step ensures that $\ttutl{i}$ remains less than  $\ttltl{k}$ (which is initialized with $\tcts{k}$); (2) Ensure that $\ttltl{i}$ is still less than $\ttutl{i}$. Otherwise, $T_i$ is aborted. 

We illustrate this technique with the history $H1$ shown in \figref{sfmv-correct}. When $T_1$ starts its $\tcts{1} = 50, \ttltl{1} = 50, \ttutl{1}=\infty$. Now when $T_1$ commits, suppose $\gtcnt$ is 70. Hence, $\ttutl{1}$ reduces to 70. Next, when $T_2$ commits, suppose $\ttutl{2}$ reduces to 75 (the current value of $\gtcnt$). As $T_1, T_2$ have accessed a common \tobj $x$ in a conflicting manner, $\ttltl{2}$ is incremented to a value greater than $\ttutl{1}$, say 71. Next, when $T_3$ begins, $\ttltl{3}$ is assigned $\tcts{3}$ which is 80 and $\ttutl{3}$ is initialized to $\infty$. When $T_3$ reads 10 from $T_1$, which is $r_3(x, 10)$, $\ttutl{3}$ is reduced to a value less than $\ttltl{2} (= 71)$, say 70. But $\ttltl{3}$ is already at 80. Hence, the limits of $T_3$ have crossed and thus causing $T_3$ to abort. The resulting history consisting of only committed transactions $T_1 T_2$ is \stsble. 

Based on this idea, we next develop a variation of \sfkv, \emph{K-version Starvation-Free STM System} or \emph{\ksftm}. To explain this algorithm, we first describe the structure of the version of a \tobj used. It is a slight variation of the \tobj used in \pkto algorithm. It consists of: (1) timestamp, $ts$  which is the \wts of the transaction that created this version (and not \cts like \pkto); (2) the value of the version; (3) a list, called \rlist{}, consisting of transactions ids (could be \cts as well) that read from this version; (4) version \rt timestamp or \vt which is the \utl of the transaction that created this version. Thus a version has information of \wts and \utl of the transaction that created it. 

Now, we describe the main idea behind $\begt$, $\tread$, $\twrite$ and $\tryc{}$ \op{s} of a transaction $T_i$ which is an extension of \pkto. Note that as per our notation $i$ represents the \cts of $T_i$. 

\noindent
\textbf{$\begt(t)$:} A unique timestamp $ts$ is allocated to $T_i$ which is its \cts ($i$ from our assumption) which is generated by atomically incrementing the global counter $\gtcnt$. If the input $t$ is null then $\tcts{i} = \tits{i} = ts$ as this is the first \inc of this transaction. Otherwise, the non-null value of $t$ is assigned to $\tits{i}$. Then, \wts is computed by \eqnref{wtsf}. Finally, \ltl and \utl are initialized: $\ttltl{i} = \tcts{i}$, $\ttutl{i} = \infty$. 

\noindent
\textbf{$\tread(x)$:} Transaction $T_i$ reads from a version of $x$ with timestamp $j$ such that $j$ is the largest timestamp less than $\twts{i}$ (among the versions $x$), i.e. there exists no version $k$ such that $j<k<\twts{i}$ is true. If no such $j$ exists then $T_i$ is aborted. Otherwise, after reading this version of $x$, $T_i$ is stored in $j$'s $rl$. Then we modify \ltl, \utl as follows: 
\begin{enumerate}
\item The version $x[j]$ is created by a transaction with $\twts{j}$ which is less than $\twts{i}$. Hence, $\ttltl{i} = max(\ttltl{i}, x[j].$\vt$  + 1)$.
\item Let $p$ be the timestamp of smallest version larger than $i$. Then $\ttutl{i} = min(\ttutl{i}, x[p].\vt - 1)$.
\item After these steps, abort $T_i$ if \ltl and \utl have crossed, i.e., $\ttltl{i} > \ttutl{i}$. 
\end{enumerate}


\noindent
\textbf{$\twrite(x,v)$:} $T_i$ stores this write to value $x$ locally in its $\ws_i$. 

\noindent
\textbf{$\tryc:$} This \op{} consists of multiple steps: 
\begin{enumerate}
	\item 
	Before $T_i$ can commit, we need to verify that any version it creates is updated consistently. $T_i$ creates a new version with timestamp $\twts{i}$. Hence, we must ensure that any transaction that read a previous version is unaffected by this new version. Additionally, creating this version would require an update of \ltl and \utl of $T_i$ and other transactions whose read-write set overlaps with that of $T_i$. Thus, $T_i$ first validates each \tobj{} $x$ in its $\ws{}$ as follows: \label{step:kverify}
	\begin{enumerate}
		\item $T_i$ finds a version of $x$ with timestamp $j$ such that $j$ is the largest timestamp less than $\twts{i}$ (like in $\tread$). If there exists no version of $x$ with a timestamp less than $\twts{i}$ then $T_i$ is aborted. This is similar to \stref{notfound} of the $\tryc$ of \pkto algorithm. \label{step:k-look}
		
		\item Among all the transactions that have previously read from $j$ suppose there is a transaction $T_k$ such that $j<\twts{i}<\twts{k}$. Then (i) if $T_k$ has already committed then $T_i$ is aborted; (ii) Suppose $T_k$ is live, and $\tits{k}$ is less than $\tits{i}$. Then again $T_i$ is aborted; (iii) If $T_k$ is still live with $its_i$ less than $its_k$ then $T_k$ is aborted. 
		
		This step is similar to \stref{verify} of the $\tryc$ of \pkto algorithm. \label{step:k-verify}
			    
		
				
		\item Next, we must ensure that $T_i$'s \ltl and \utl are updated correctly w.r.t to other concurrently executing transactions. To achieve this, we adjust \ltl, \utl as follows: (i) Let $j$ be the $ts$ of the largest version smaller than $\twts{i}$. Then $\ttltl{i} = max(\ttltl{i}, x[j].\vt + 1)$. Next, for each reading transaction, $T_r$ in $x[j].\rlist$, we again set, $\ttltl{i} = max(\ttltl{i}, \ttutl{r} + 1)$. (ii) Similarly, let $p$ be the $ts$ of the smallest version larger than $\twts{i}$. Then, $\ttutl{i} = min(\ttutl{i}, x[p].\vt - 1)$. (Note that we don't have to check for the transactions in the \rlist of $x[p]$ as those transactions will have \ltl higher than $x[p].\vt$ due to $\tread$.) (iii) Finally, we get the commit time of this transaction from \gtcnt: $\ct_i = \gtcnt.add\&Get(\incv)$ where $\incv$ is any constant $\geq 1$. Then, $\ttutl{i} = min(\ttutl{i}, \ct_i)$. After performing these updates, abort $T_i$ if \ltl and \utl have crossed, i.e., $\ttltl{i} > \ttutl{i}$. \label{step:ktk-upd}
	\end{enumerate}
	
	\item After performing the tests of \stref{kverify} over each \tobj{s} $x$ in $T_i$'s $\ws$, if $T_i$ has not yet been aborted, we proceed as follows: for each $x$ in $\ws_i$ create a \vtup $\langle \twts{i}, \ws_i.x.v, null,\\ \ttutl{i} \rangle$. In this tuple, $\twts{i}$ is the timestamp of the new version; $\ws_i.x.v$ is the value of $x$ is in $T_i$'s $\ws$; the \rlist of the $\vtup$ is $null$; $\vt$ is  $\ttutl{i}$ (actually it can be any value between $\ttltl{i}$ and $\ttutl{i}$). Update the $\vlist$ of each \tobj $x$ similar to \stref{updt} of $\tryc$ of \pkto. 	
	
	\ignore{		
	\begin{enumerate}
		\item $T_i$ creates a \vtup $\langle i, x.v, null \rangle$. In this tuple, $i$ (\cts of $T_i$) is the timestamp of the new version; $x.v$ is the value of $x$ is in $T_i$'s \ws and the \rlist of the \vtup is $null$.
		\item Suppose the total number of versions of $x$ is $K$. Then among all the versions of $x$, $T_i$ replaces the version with the smallest timestamp with \vtup $\langle i, x.v, null \rangle$. Otherwise, the \vtup is added to $x$'s \vlist. 
	\end{enumerate}
	}
	\item Transaction $T_i$ is then committed. \label{step:kcommit}
\end{enumerate}

\noindent \stref{ktk-upd}.(iii) of $\tryc$ ensures that \rt order between transactions that are not in conflict. It can be seen that locks have to be used to ensure that all these \mth{s} to execute in a \lble manner (i.e., atomically).
\subsection{Data Structures and Pseudocode of \ksftm}
\label{sec:code}


The STM system consists of the following methods: $\init(), \begt(), read(i, x), write_i(i, x, v)$ and $\tryc(i)$. We assume that all the \tobj{s} are ordered as $x_1, x_2, ...x_n$ and belong to the set $\mathcal{T}$. We describe the data-structures used by the algorithm. 

We start with structures that local to each transaction. Each transaction $T_i$ maintains a $\rset{i}$ and $\wset{i}$. In addition it maintains the following structures (1) $\ct_i$: This is value given to $T_i$ when it terminates which is assigned a value in \tryc \mth. (2) A series of lists: \srl, \lrl, \allrl, \pvl, \nvl, \relll, \abl. The meaning of these lists will be clear with the description of the pseudocode. In addition to these local structures, the following shared global structures are maintained that are shared across transactions (and hence, threads). We name all the shared variable starting with `G'. 
\begin{itemize}
	\item $\gtcnt$ (counter): This a numerical valued counter that is incremented when a transaction begins and terminates. 
\end{itemize}
\noindent For each transaction $T_i$ we maintain the following shared time-stamps:
\begin{itemize}
	\item $\glock_i$: A lock for accessing all the shared variables of $T_i$.
	\item $\gits_i$ (initial timestamp): It is a time-stamp assigned to $T_i$ when it was invoked for the first time without any aborts. The current value of $\gtcnt$ is atomically assigned to it and then incremented. If $T_i$ is aborted and restarts later then the application assigns it the same \gits. 
	
	\item $\gcts_i$ (current timestamp): It is a time-stamp when $T_i$ is invoked again at a later time after an abort. Like \gits,the current value of $\gtcnt$ is atomically assigned to it and then incremented. When $T_i$ is created for the first time, then its \gcts{} is same as its \gits. 
	
	\item $\gwts_i$ (working timestamp): It is the time-stamp that $T_i$ works with. It is either greater than or equal to $T_i$'s \gcts. It is computed as follows: $\gwts_i = \gcts_i + C * (\gcts_i - \gits_i)$.
	
	\item $\gval_i$: This is a boolean variable which is initially true. If it becomes false then $T_i$ has to be aborted.
	
	\item $\gstat_i$: This is a variable which states the current value of $T_i$. It has three states: \texttt{live}, \texttt{committed} or \texttt{aborted}. 
	
	\item $\tltl_i, \tutl_i$ (transaction lower \& upper time limits): These are the time-limits described in the previous section used to keep the transaction \wts and \rt orders in sync. $\tltl_i$ is \gcts{} of $T_i$ when transaction begins and is a non-decreasing value. It continues to increase (or remains same) as $T_i$ reads \tobj{s} and later terminates. $\tutl_i$ on the other hand is a non-increasing value starting with $\infty$ when the $T_i$ is created. It reduces (or remains same) as $T_i$ reads \tobj{s} and later terminates. If $T_i$ commits then both $\tltl_i$ \& $\tutl_i$ are made equal.	
\end{itemize}
Two transactions having the same \its are said to be \inc{s}. No two transaction can have the same \cts. For simplicity, we assume that no two transactions have the same \wts as well. In case, two transactions have the same \wts, one can use the tuple $\langle$\wts, \cts$\rangle$ instead of \wts. But we ignore such cases. For each \tobj $x$ in $\mathcal{T}$, we maintain:
\begin{itemize}
	\item $x.\vl$ (version list): It is a list consisting of version tuples or \emph{\vtup} of the form $\langle \ts, val, \rl, \vt \rangle$. The details of the tuple are explained below. 
	
	\item $\ts$ (timestmp): Here $\ts$ is the $\gwts_i$ of a committed transaction $T_i$ that has created this version. 
	
	\item $val$: The value of this version. 
	
	\item $\rl$ (readList): $rl$ is the read list consists of all the transactions that have read this version. Each entry in this list is of the form $\langle rts \rangle$ where $rts$ is the $\gwts_j$ of a transaction $T_j$ that read this version.
	
	\item $\vt$ (version real-time timestamp): It is the \tutl value (which is same as \tltl) of the transaction $T_i$ that created this version at the time of commit of $T_i$.
\end{itemize}

\begin{algorithm}[H]
	\caption{STM $\init()$: Invoked at the start of the STM system. Initializes all 
		the \tobj{s} used by the STM System}  \label{algo:init}
		\begin{algorithmic}[1]
			\State $\gtcnt$ = 1; \Comment{Global Transaction Counter}
			\ForAll {$x$ in $\mathcal{T}$} \Comment{All the \tobj{s} used by the STM System}
			\State /* $T_0$ is creating the first version of $x$:  $\ts= 0, val = 0, \rl 
			= nil, \vt = 0$ */
			\State add $\langle 0, 0, nil, 0 \rangle$ to $x.\vl$; \label{lin:t0-init1} 
			\EndFor;
		\end{algorithmic}
\end{algorithm}

\begin{algorithm} [H]
	
	\caption{STM $\begt(its)$: Invoked by a thread to start a new transaction $T_i$. Thread can pass a parameter $its$ which is the initial timestamp when this transaction was invoked for the first time. If this is the first invocation then $its$ is $nil$. It returns the tuple $\langle id, \gwts, \gcts \rangle$} \label{algo:begin} 
			
		\begin{algorithmic}[1]
			\State $i$ = unique-id; \Comment{An unique id to identify this transaction. 
				It could be same as \gcts}
			\State \Comment{Initialize transaction specific local \& global variables}  
			\If {($its == nil$)}
			\State $\gits_i = \gwts_i = \gcts_i = \gtcnt.get\&Inc()$; \Comment{$\gtcnt.get\&Inc()$ returns the current value of \gtcnt and atomically increments it}
			\label{lin:ti-ts-init}

			\Else 
			\State $\gits_i = its$;
		
			\State $\gcts_i = \gtcnt.get\&Inc()$; 
			\State $\gwts_i = \gcts_i + C * (\gcts_i - \gits_i)$; \Comment{$C$ is any 
				constant greater or equal to than 1}

			\EndIf
			
			\State $\tltl_i = \gcts_i$; $\tutl_i = \ct_i = \infty$; \label{lin:lts-init}
		
			\State $\gstat_i$ = \texttt{live}; $\gval_i = T$;
			
			\State $\rs_i = \ws_i = nil$;			
			\State return $\langle i, \gwts_i, \gcts_i\rangle$
		\end{algorithmic}
\end{algorithm}
\begin{algorithm}[H]
	
	\caption{STM $read(i, x)$: Invoked by a transaction $T_i$ to read \tobj{} $x$. It returns either the value of $x$ or $\mathcal{A}$} \label{algo:read}
	
		\begin{algorithmic}[1]  
			\If {($x \in \ws_i$)} \Comment{Check if the \tobj{} $x$ is in 
						$\ws_i$}
						\State return $\ws_i[x].val$;
			\ElsIf{($x \in \rs_i$)} \Comment{Check if the \tobj{} $x$ is in 
						$\rs_i$}
						\State return $\rs_i[x].val$; 
			\Else \Comment{\tobj{} $x$ is not in $\rs_i$ and $\ws_i$}         
			\State lock $x$; lock $\glock_i$;
			\If {$(\gval_i == F)$} return abort(i); \label{lin:rd-chk}
			\EndIf    			
			\State /* \findls: From $x.\vl$, returns the largest \ts value less than $\gwts_i$. If no such version exists, it returns $nil$ */
			\State $curVer  = \findls(\gwts_i,x)$; 	\label{lin:rd-curver10} 					
			\If {$(curVer == nil)$} return abort(i); \Comment{Proceed only if $curVer$ is not nil} \label{lin:rd-cvnil}
			\EndIf
			\State /* \findsl: From $x.\vl$, returns the smallest \ts value greater than $\gwts_i$. If no such version exists, it returns $nil$ */			
			\State $nextVer  = \findsl(\gwts_i,x)$; 
\algstore{myalg}
\end{algorithmic}
\end{algorithm}
\begin{algorithm}[H]
\begin{algorithmic}
\algrestore{myalg}
			\If {$(nextVer \neq nil)$}
				\State \Comment{Ensure that $\tutl_i$ remains smaller than $nextVer$'s \vltl}
				\State $\tutl_i = min(\tutl_i, x[nextVer].\vt-1)$; \label{lin:rd-ul-dec}
			\EndIf
							
			\State \Comment{$\tltl_i$ should be greater than $x[curVer].\vltl$}
			
			\State $\tltl_i = max(\tltl_i, x[curVer].\vt + 1)$; \label{lin:rd-tltl-inc}
			\If {($\tltl_i > \tutl_i$)} \Comment{If the limits have crossed each other, then $T_i$ is aborted}
				\State return abort(i); \label{lin:rd-lts-cross}
			\EndIf    
			
			\State $val = x[curVer].v$; add $\langle x, val \rangle$ to $\rs_i$;
			\State add $T_i$ to $x[curVer].rl$; 
			\State unlock $\glock_i$; unlock $x$;
			\State return $val$;
			\EndIf
		\end{algorithmic}
\end{algorithm}

\begin{algorithm}  
	\caption{STM $write_i(x,val)$: A Transaction $T_i$ writes into local memory} \label{algo:write} 
	\begin{algorithmic}[1]
		\State Append the $d\_tuple \langle x,val \rangle$ to $\ws_i$.
		\State return $ok$;
	\end{algorithmic}
\end{algorithm}
\begin{algorithm} [H]
	\caption{STM $\tryc()$: Returns $ok$ on commit else return Abort} \label{algo:tryc} 
		\begin{algorithmic}[1]
			\State \Comment{The following check is an optimization which needs to be 
				performed again later}
			
			\State lock $\glock_i$;
			
			\If {$(\gval_i == F)$} return abort(i); \label{lin:init-tc-chk}
			\EndIf
			\State unlock $\glock_i$;
			
			\State \Comment{Initialize smaller read list (\srl), larger read list (\lrl), 
				all read list (\allrl) to nil}
			\State $\srl = \lrl = \allrl = nil$; \label{lin:init-rls}
			
			\State \Comment{Initialize previous version list (\pvl), next version list (\nvl) to nil}
			\State $\pvl = \nvl = nil$; \label{lin:init-vls}
			
			\ForAll {$x \in \ws_i$}
			\State lock $x$ in pre-defined order; \label{lin:lockxs} 
			\State /* \findls: returns the version of $x$ with the largest \ts less than $\gwts_i$. If no such version exists, it returns $nil$. */
			
			\State $\prevv = \findls(\gwts_i, x)$; \Comment{\prevv: largest version smaller than $\gwts_i$}
			\If {$(\prevv == nil)$} \Comment{There exists no version with \ts value less 
				than $\gwts_i$}
			\State lock $\glock_i$; return abort(i); \label{lin:prev-nil}
			\EndIf																
			\State $\pvl = \pvl \cup \prevv$; \Comment{\pvl stores the previous 
			version in sorted order }  
			
			\State $\allrl = \allrl \cup x[\prevv].rl$; \Comment{Store the read-list of the previous version}
			
			\State \Comment{\textbf{\getl}: obtain the list of reading 
			transactions of $x[\prevv].rl$ whose $\gwts$ is greater than 
			$\gwts_i$}
		
			\State $\lrl = \lrl \cup \getl(\gwts_i, $
			\Statex $x[\prevv].rl)$; \label{lin:lar-coll}
			
			\State \Comment{\textbf{\getsm}: obtain the list of reading transactions of $x[\prevv].rl$ whose $\gwts$ is smaller than $\gwts_i$}
			\State $\srl = \srl \cup \getsm(\gwts_i, $ 
			\Statex $x[\prevv].rl)$; \label{lin:lar-sml}  
\algstore{myalg}
\end{algorithmic}
\end{algorithm}
\begin{algorithm}
\begin{algorithmic}[H]
\algrestore{myalg}			
			\State /* \findsl: returns the version with the smallest \ts value greater than $\gwts_i$. If no such version exists, it returns $nil$. */
			\State $\nextv = \findsl(\gwts_i, x)$; \Comment{\nextv: smallest version larger than $\gwts_i$} \label{lin:get-nextv}
			
			\If {$(\nextv \neq nil)$)}
				\State $\nvl = \nvl \cup \nextv$; \Comment{\nvl stores the next version in sorted order}  \label{lin:nvl-coll}
			\EndIf
		
			\EndFor \Comment{$x \in \ws_i$}  
									
			\State $\relll = \allrl \cup T_i$; \Comment{Initialize relevant Lock List (\relll)}
				
			\ForAll {($T_k \in \relll$)}
				\State lock $\glock_k$ in pre-defined order; \Comment{Note: Since $T_i$ is also in $\relll$, $\glock_i$ is also locked} \label{lin:lockall}
			\EndFor
	
			\State \Comment{Verify if $\gval_i$ is false}
			\If {$(\gval_i == F)$} 
				return abort(i); \label{lin:mid-tc-chk}
			\EndIf
			
			\State $\abl = nil$ \Comment{Initialize abort read list (\abl)}
								
			\State \Comment{Among the transactions in $T_k$ in $\lrl$, either $T_k$ or $T_i$ has to be aborted}
			
			\ForAll {$(T_k \in \lrl)$}  
			\If {$(\isab(T_k))$} 
				\State \Comment{Transaction $T_k$ can be ignored since it is already aborted or about to be aborted}
				\State continue;
			\EndIf
			
\If {$(\gits_i < \gits_k) \land (\gstat_k == \texttt{live})$}   
				\State \Comment{Transaction $T_k$ has lower priority and is not yet committed. So it needs to be aborted}
				\State $\abl = \abl \cup T_k$; \Comment{Store $T_k$ in \abl} \label{lin:addAbl-lar}
			\Else \Comment{Transaction $T_i$ has to be aborted}
				\State return abort(i); \label{lin:its-chk1}
			\EndIf   
						
			\EndFor
			
			\State \Comment{Ensure that $\tltl_i$ is greater than \vltl of the versions in $\pvl$} 	
								
			\ForAll {$(ver \in \pvl)$}
				\State $x$ = \tobj of $ver$;
		
				\State $\tltl_i = max(\tltl_i, x[ver].\vt + 1)$; \label{lin:tryc-tltl-inc}
			\EndFor
			
			\State \Comment{Ensure that $\vutl_i$ is less than \vltl of versions in $\nvl$} 
				
			\ForAll {$(ver \in \nvl)$}
			
				\State $x$ = \tobj of $ver$;
							
				\State $\tutl_i = min(\tutl_i, x[ver].\vltl - 1)$; \label{lin:tryc-ul-dec}
				
			\EndFor 
				
			\State \Comment{Store the current value of the global counter as commit time and increment it}	
					
			\State $\ct_i = \gtcnt.add\&Get(\incv)$; \Comment{$\incv$ can be any constant $\geq$ 1} \label{lin:tryc-cmt-mod}
			
			\State $\tutl_i = min(\tutl_i, \ct_i)$; \Comment{Ensure that $\tutl_i$ is less than or equal to $\ct$} \label{lin:tryc-ul-cmt}

			\State \Comment{Abort $T_i$ if its limits have crossed}
			\If {$(\tltl_i > \tutl_i)$} 
				return abort(i); \label{lin:tc-lts-cross}
			\EndIf
\algstore{myalg}
\end{algorithmic}
\end{algorithm}
\begin{algorithm}
\begin{algorithmic}[H]
\algrestore{myalg}	
			\ForAll {$(T_k \in \srl)$}  
			\If {$(\isab(T_k))$} 
				\State continue;
				\EndIf
			
			\If {$(\tltl_k \geq \tutl_i)$} \label{lin:tk-check} \Comment{Ensure that the limits do not cross for both $T_i$ \& $T_k$}
			
			
			
				\If {$(\gstat_k == live)$}  \Comment{Check if $T_k$ is live}        

					\If {$(\gits_i < \gits_k)$}
						\State \Comment{Transaction $T_k$ has lower priority and is not yet committed. So it needs to be aborted}
						\State $\abl = \abl \cup T_k$; \Comment{Store $T_k$ in \abl} \label{lin:addAbl-sml}     
					\Else \Comment{Transaction $T_i$ has to be aborted}
						\State return abort(i); \label{lin:its|lar-sml}
					\EndIf \Comment{$(\gits_i < \gits_k)$}    \label{lin:its-ik}
				\Else \Comment{($T_k$ is committed. Hence, $T_i$ has to be aborted)}   
					\State return abort(i); \label{lin:its-chk2}
				\EndIf \Comment{$(\gstat_k == live)$}
			\EndIf \Comment{$(\tltl_k \geq \tutl_i)$}			
			\EndFor {$(T_k \in \srl)$}
			
			\State \Comment{After this point $T_i$ can't abort.}
			\State $\tltl_i = \tutl_i$;  \label{lin:ti-updt} 
			
			\State \Comment{Since $T_i$ can't abort, we can update $T_k$'s \tutl}
			\ForAll {$(T_k \in \srl)$}  
			\If {$(\isab(T_k))$} 
				\State continue;
			\EndIf
			       
            \State /* The following line ensure that $\tltl_k \leq \tutl_k < \tltl_i$. Note that this does not cause the limits of $T_k$ to cross each other because  of the check in  \Lineref{tk-check}.*/  			
			\State $\tutl_k = min(\tutl_k, \tltl_i - 1)$; \label{lin:tk-updt}
			\EndFor
			
			\ForAll {$T_k \in \abl$} \Comment{Abort all the transactions in \abl since 
				$T_i$ can't abort}
			\State $\gval_k =  F$; \label{lin:gval-set}
			\EndFor
			
			\State \Comment{Having completed all the checks, $T_i$ can be committed}
			\ForAll {$(x \in \ws_i)$} 
			\State /* Create new v\_tuple: $\ts, val, \rl, \vt$ for $x$ */
			\State $newTuple = \langle \gwts_i, \ws_i[x].val, nil, \tltl_i \rangle$; \label{lin:new-tup} 
			
			\If {($|x.vl| > k$)}
				\State replace the oldest tuple in $x.\vl$ with $newTuple$; \Comment{$x.\vl$ is ordered by $\ts$}
			\Else
				\State add a $newTuple$ to $x.vl$ in sorted order;
			\EndIf
			\EndFor \Comment{$x \in \ws_i$}
			
			\State $\gstat_i$ = \texttt{commit};
			\State unlock all variables;
			\State return $\mathcal{C}$;
			
		\end{algorithmic}
\end{algorithm}

\begin{algorithm}[H]
	\caption{$\isab(T_k)$: Verifies if $T_i$ is already aborted or its \gval flag 
		is set to false implying that $T_i$ will be aborted soon} \label{algo:isab} 
		\begin{algorithmic}[1]			
			\If {$(\gval_k == F) \lor (\gstat_k == \texttt{abort}) \lor (T_k \in \abl)$} 
			\State return $T$;
			\Else
			\State return $F$;
			\EndIf
		\end{algorithmic}
\end{algorithm}

\begin{algorithm}[H]
	\caption{$abort(i)$: Invoked by various STM methods to abort transaction $T_i$. 
		It returns $\mathcal{A}$} \label{algo:abort} 
		\begin{algorithmic}[1]
			\State $\gval_i = F$; $\gstat_i$ = \texttt{abort};
			\State unlock all variables locked by $T_i$;      
			\State return $\mathcal{A}$;
		\end{algorithmic}
\end{algorithm}
We get the following nice properties on \ksftm. For simplicity, we assumed $C$ and $\incv$ to be 0.1 and 1 respectively in our analysis. But the proof and the analysis holds for any value greater than 0. 

\ignore{
\begin{property}
\label{prop:ksftm-sble}
Any history generated by \pkto is \stsble.
\end{property}
}

\begin{theorem}
\label{thm:ksftm-lo}
Any history generated by \ksftm is strict-serializable and \lopq.
\end{theorem}



\begin{theorem}
\label{thm:ksftm-live}
\ksftm algorithm ensures \stfdm. 
\end{theorem}

\noindent As explained in the description \propref{sfmv-live}, the proof of this property is somewhat involved. As expected, this proof can be extended to \mvsftm as well. 

\noindent
\textbf{Garbage Collection:} Having described the \emph{\stf} algorithm, we now describe how garbage collection can be performed on the unbounded variant, \mvsftm to achieve \mvsftmgc. This is achieved by deleting non-latest version (i.e., there exists a version with greater $ts$) of each \tobj whose timestamp, $ts$ is less than the \cts of smallest live transaction. It must be noted that \mvsftm (\ksftm) works with \wts which is greater or equal to \cts for any transaction. Interestingly, the same garbage collection principle can be applied for \pmvto to achieve \pmvtogc. 

To identify the transaction with the smallest \cts among live transactions, we maintain a set of all the live transactions, \livel. When a transaction $T_i$ begins, its \cts is added to this \livel. And when $T_i$ terminates (either commits or aborts), $T_i$ is deleted from this \livel. 



\section{Experimental Evaluation}
\label{sec:exp}

For performance evaluation of \ksftm with the state-of-the-art STMs, we implemented the the algorithms \pkto, \svsftm \cite{Gramoli+:TM2C:Eurosys:2012, WaliSten:Starve:CC:2009, Spear+:CCM:PPoPP:2009} along with \ksftm in C++ \footnote{Code is available here: https://github.com/PDCRL/KSFTM}. We used the available implementations of NOrec STM \cite{Dalessandro:2010:NSS:1693453.1693464}, and ESTM \cite{Felber:2017:jpdc} developed in C++. Although, only \ksftm and \svsftm provide \stfdm, we compared with other STMs as well, to see its performance in practice. 

\noindent
\textbf{Experimental system:} The experimental system is a 2-socket Intel(R) Xeon(R) CPU E5-2690 v4 @ 2.60GHz with 14 cores per socket and 2 hyper-threads (HTs) per core, for a total of 56 threads. Each core has a private 32KB L1 cache and 256 KB L2 cache. The machine has 32GB of RAM and runs Ubuntu 16.04.2 LTS. In our implementation, all threads have the same base priority and we use the default Linux scheduling algorithm. This satisfies the \asmref{bdtm} (\bdtm) about the scheduler. We ensured that there no parasitic transactions \cite{BushkovandGuerraoui:2015:COST} in our experiments.
	
\noindent
\textbf{Methodology:} Here we have considered two different applications:\textbf{(1)} Counter application - In this, each thread invokes a single transaction which performs 10 reads/writes \op{s} on randomly chosen \tobj{s}. A thread continues to invoke a transaction until it successfully commits. To obtain high contention, we have taken large number of threads ranging from 50-250 where each thread performs its read/write operation over a set of 5 \tobj{s}. We have performed our tests on three workloads stated as: (W1) Li - Lookup intensive: 90\% read, 10\% write, (W2) Mi - Mid intensive: 50\% read, 50\% write and (W3) Ui - Update intensive: 10\% read, 90\% write. This application is undoubtedly very flexible as it allows us to examine performance by tweaking different parameters (refer to \subsecref{countercode} for details).  \textbf{(2)} Two benchmarks from STAMP suite \cite{stamp123} - (a) We considered KMEANS which has low contention with short running transactions. The number of data points as 2048 with 16 dimensions and total clusters as 5. (b) We then considered LABYRINTH which has high contention with long running transactions. We considered the grid size as 64x64x3 and paths to route as 48.

To study starvation in the various algorithms, we considered \emph{\wct}, which is the maximum time taken by a transaction among all the transactions in a given experiment to commit from its first invocation. This includes time taken by all the aborted \inc{s} of the transaction to execute as well. To reduce the effect of outliers, we took the average of \wct in ten runs as the final result for each application.

	\cmnt{
	and second application is labyrinth from STAMP benchmark \cite{stamp} which has long running transactions with high contention. \cmnt{In our counter application, each thread invokes a single transaction and within a transaction, a thread performs 10 reads/writes \op{s} on randomly chosen \tobj{s} from a set of 5 \tobj{s} to increase the contention level. A thread continues to invoke a transaction until it successfully commits.To increase the contention we considered multiple threads. We have the control of greater flexibility with counter application. So, }We have performed our tests on three workloads stated as: (W1) Li - Lookup intensive (90\% read, 10\% write), (W2) Mi - Mid intensive (50\% read, 50\% write) and (W3) Ui - Update intensive (10\% read, 90\% write). \cmnt{We have also checked the performance on labyrinth which is showing similar results. Due to brevity of space, the comparison of algorithms are in \apnref{apn-result}.} For accurate results, we took an average of ten runs as the final result for each algorithm.}
\begin{figure}
	\captionsetup{justification=centering}
	\includegraphics[width=\linewidth]{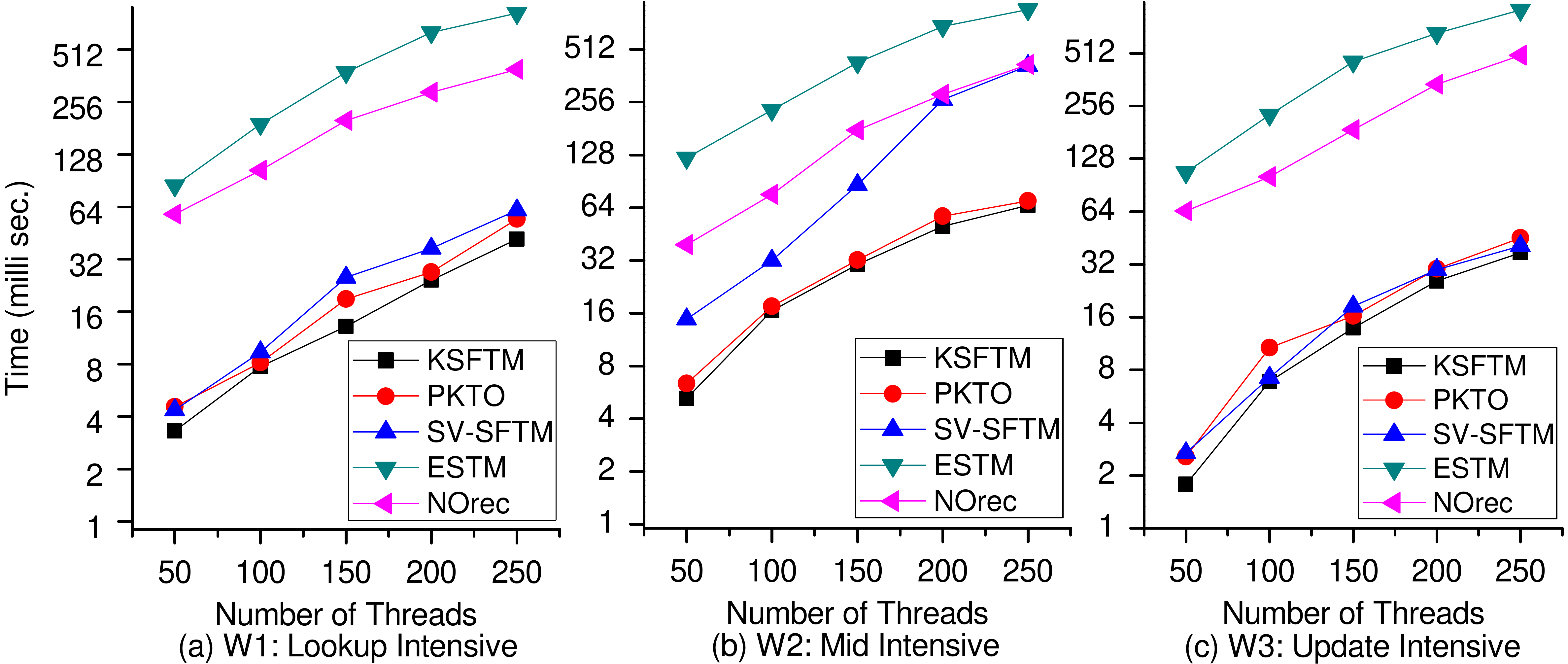}
	\vspace{-.2cm} \caption{Performance analysis on workload $W1$, $W2$, $W3$}\label{fig:wcts}
\end{figure}

\noindent
\textbf{Results Analysis:} \figref{wcts} illustrates \wct analysis of \ksftm over the above mentioned STMs for the counters application under the workloads $W1$, $W2$ and $W3$ while varying the number of threads from 50 to 250. For \ksftm and \pkto, we chose the value of K as 5 and C as 0.1 as the best results were obtained with these parameters. We can see that \ksftm performs the best for all the three workloads. \ksftm gives an average speedup on \wct by a factor of 1.22, 1.89, 23.26 and 13.12 over \pkto, \svsftm, NOrec STM and ESTM respectively.
	
\ignore{
Under high contention the time taken by the longest running transaction to commit is considered as the worst case time. 
	
We implemented 3 variants of \ksftm (\mvsftm, \mvsftmgc, and \ksftm) and \pkto (\pmvto, \pmvtogc, and \pkto) and tested on all the workloads $W1$  $W2$ and $W3$. \ksftm outperforms \mvsftm and \mvsftmgc by a factor of 2.1 and 1.5. Similarly, \pkto outperforms \pmvto and \pmvtogc by a factor of 2 and 1.35. These results show that maintaining finite versions corresponding to each \tobj performs better than maintaining infinite versions and garbage collection on infinite versions corresponding to each \tobj. We identified the best value of K as 5 and optimal value of $C$ as 0.1 for \ksftm on counter application.
we ran our experiment, varying value of K and keeping the number of threads as 64 on workload $W1$ and obtained the optimal value of $K$ in \ksftm is 5 as shown in \figref{worstcase}.(b). Similarly, we calculate the best value of $K$ as 5 for \pkto on the same parameters. The optimal value of $C$ as 0.1 for \ksftm on the above parameters (details described in \apnref{apn-result} and technical report \cite{DBLP:journals/corr/abs-1709-01033}).
}
    
\figref{stamp}(a) shows analysis of \wct for KMEANS while \figref{stamp}(b) shows for LABYRINTH. In this analysis we have not considered ESTM as the integrated STAMP code for ESTM is not publicly available. For KMEANS, \ksftm performs 1.5 and 1.44 times better than \pkto and \svsftm. But, NOrec is performing 1.09 times better than \ksftm. This is because KMEANS has short running transactions have low contention. As a result, the commit time of the transactions is also low. 

On the other hand for LABYRINTH, \ksftm again performs the best. It performs 1.14, 1.4 and 2.63 times better than \pkto, \svsftm and NOrec respectively. This is because LABYRINTH  has high contention with long running transactions. This result in longer commit times for transactions. 
    
\figref{stamp}(c) shows the stability of \ksftm algorithm over time for the counter application. Here we fixed the number of threads to 32, $K$ as 5, $C$ as 0.1, \tobj{s} as 1000, along with 5 seconds warm-up period on $W1$ workload. Each thread invokes transactions until its time-bound of 60 seconds expires. We performed the experiments on number of transactions committed over time in the increments 5 seconds. The experiment shows that over time \ksftm is stable which helps to hold the claim that \ksftm's performance will continue in same manner if time is increased to higher orders. 
    

\vspace{1mm}	
	\begin{figure}
	\captionsetup{justification=centering}
	\includegraphics[width=\linewidth]{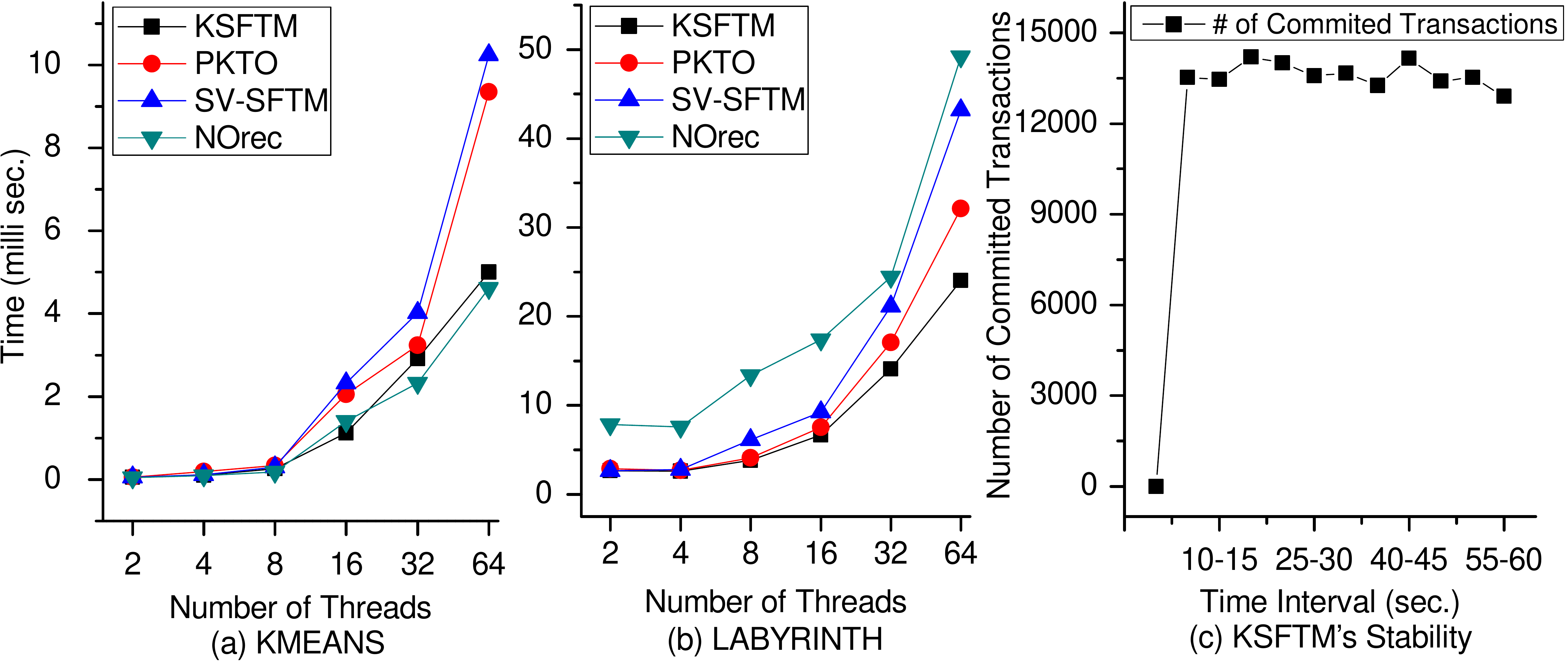}
\vspace{-.2cm} 	\caption{Performance analysis on KMEANS, LABYRINTH and KSFTM's Stability }\label{fig:stamp}
\end{figure}

\cmnt{
	\begin{figure}
		\captionsetup{justification=centering}
		\includegraphics[width=12cm, height=8cm]{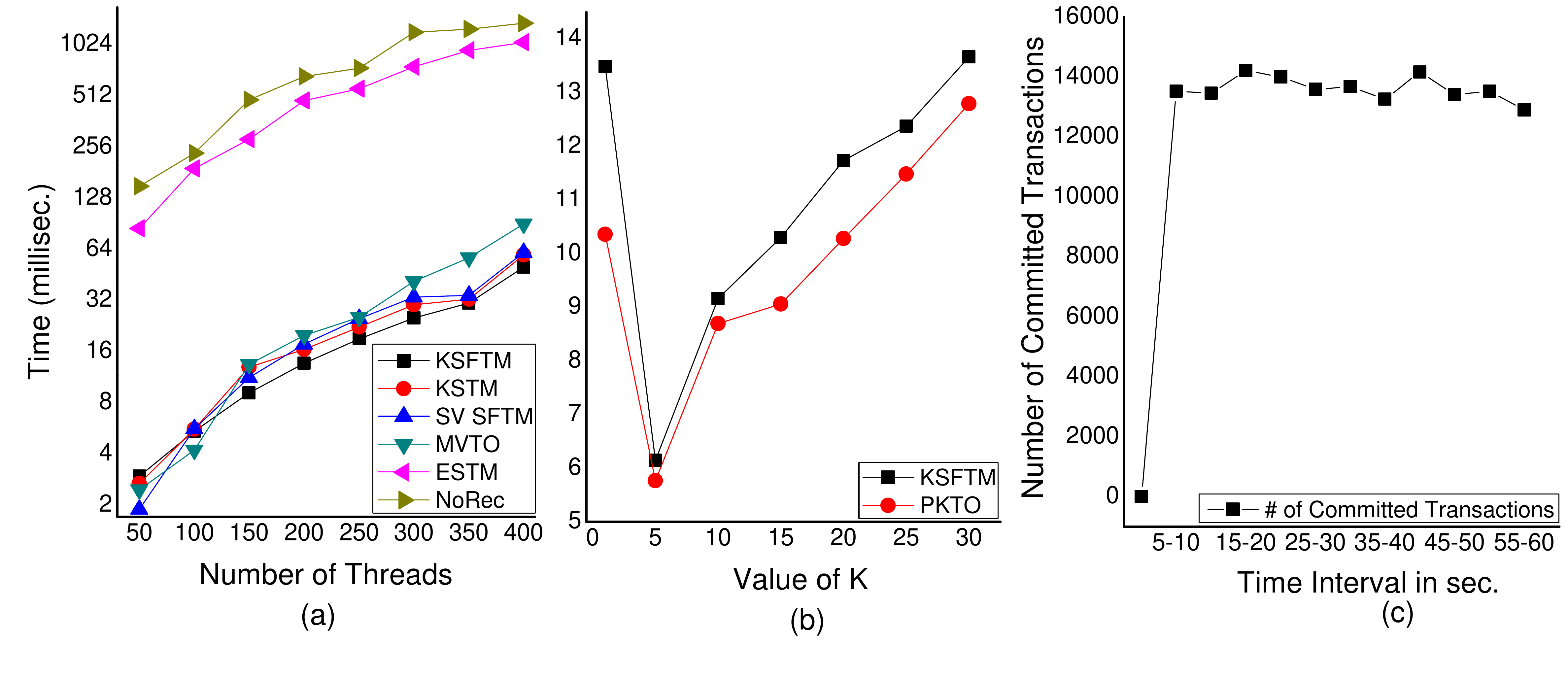}
		\vspace{-.5cm}
		\caption{Worst case time analysis, Optimal value of K and \ksftm stability}\label{fig:worstcase}
	\end{figure}

	Our experimental goals are to: (G1) identify the best value of K (i.e., number of versions) in the propose algorithm \ksftm; (G2) to evaluate the performance of all our proposed algorithms (\ksftm, \mvsftm, \mvsftmgc, \pkto \pmvto, \pmvtogc); (G3) to evaluate the overhead of starvation-freedom by comparing the performance of \ksftm  and \emph{non-\stf} \pkto STM, and  (G4) to compare the performance of \ksftm with state-of-the-art STMs (NOrec, ESTM, MVTO, and \svsftm) on different workloads. 
	
	\ignore{
	\color{red}
	It can be seen that achieving \stfdm involves some overhead. We want to understand how costly is this overhead by comparing the performance of \ksftm with finite version non-starvation free STM - \pkto.
	\color{black} 
	}
	\begin{figure}
		\captionsetup{justification=centering}
		\includegraphics[width=12cm, height=6cm]{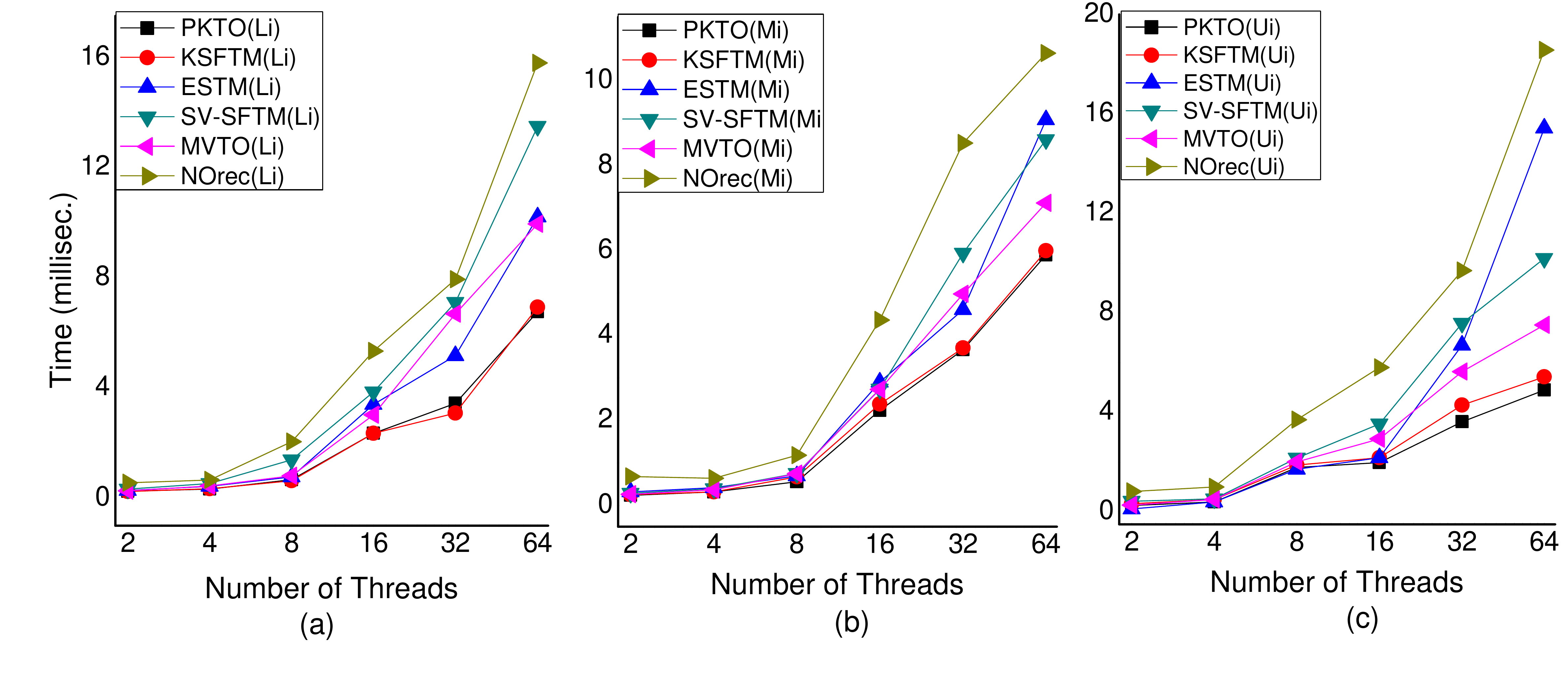}
		\caption{Performance on workload $W1$, $W2$, $W3$}\label{fig:performance}
	\end{figure}
	
	\noindent
	\textbf{Experimental system:} The experimental system is a 2-socket Intel(R) Xeon(R) CPU E5-2690 v4 @ 2.60GHz with 14 cores per socket and 2 hyper-threads (HTs) per core, for a total of 56 threads. Each core has a private 32KB L1 cache and 256 KB L2 cache. The machine has 32GB of RAM and runs Ubuntu 16.04.2 LTS. 
	
	In our implementation, all threads have the same base priority and we use the default Linux scheduling algorithm. This satisfies the \asmref{bdtm} (\bdtm), about the scheduler and no parasitic transaction \cite{BushkovandGuerraoui:2015:COST} exist.
	
	
	
	\ignore{
	\noindent
	\textbf{STM implementations:} The application first creates N-threads, each thread, in turn, invokes a transaction. Each transaction has two phases, the first phase, i.e., the read-write phase consists of read (from the shared memory) and write (to the transaction's local buffer) operations followed by the second phase, i.e., the commit phase, where the actual writes are made visible in the shared memory, and the transaction tries to commit. We have specifically chosen such a test application to achieve our experimental goals. Also to test the starvation of transactions, we wanted to vary the various parameters and our chosen application suffice this need too.
	
	We get the inspiration of single version starvation-freedom, \svsftm from the literature by Gramoli et al. \cite{Gramoli+:TM2C:Eurosys:2012}, Waliullah and Stenstrom \cite{WaliSten:Starve:CC:2009}, Spear et al. \cite{Spear+:CCM:PPoPP:2009}. 
	}
	
	\noindent
	\textbf{Methodology:} To test the algorithms, we considered a test-application in which each thread invokes a single transaction. Within a transaction, a thread performs 10 reads/writes \op{s} on randomly chosen \tobj{s} from a set of 1000 \tobj{s}. A thread continues to invoke a transaction until it successfully commits. We have considered three types of workloads: (W1) Li - Lookup intensive (90\% read, 10\% write), (W2) Mi - Mid intensive (50\% read, 50\% write) and (W3) Ui - Update intensive (10\% read, 90\% write). For accurate results, we took an average of ten runs as the final result for each algorithm. 
	
	\noindent
	\textbf{Results Analysis:} 
	We implemented 3 variants of \ksftm (\mvsftm, \mvsftmgc, and \ksftm) and \pkto (\pmvto, \pmvtogc, and \pkto) to do the performance analysis on different workloads $W1$ (Li), $W2$ (Mi) and $W3$ (Ui) respectively. \ksftm outperforms \mvsftm and \mvsftmgc by a factor of 2.1 and 1.5. Similarly, \pkto outperforms \pmvto and \pmvtogc by a factor of 2 and 1.35. These results show that maintaining finite versions corresponding to each \tobj performs better than maintaining infinite versions and garbage collection on infinite versions corresponding to each \tobj.
	
	To identify the best value of K for \ksftm on counter application, we ran our experiment, varying value of K and keeping the number of threads as 64 on workload $W1$ and obtained the optimal value of $K$ in \ksftm is 5 as shown in \figref{worstcase}.(b). Similarly, we calculate the best value of $K$ as 5 for \pkto on the same parameters. The optimal value of $C$ as 0.1 for \ksftm on the above parameters (details described in \apnref{apn-result} and technical report \cite{DBLP:journals/corr/abs-1709-01033}).
	
	\figref{performance} (a), (b) and (c) shows the performance of all the algorithms (proposed as well as state-of-the-art STMs) for workloads $W1$, $W2$ and $W3$ respectively. For workload $W1$, the graph shows that \ksftm outperforms \svsftm, NOrec, ESTM, and MVTO by a factor of 2.5, 3, 1.7 and 1.5. \ignore{ESTM performs better than \ksftm in low contention when the thread count is 16 and less, while this is not the case for high contention when the thread count is increased to 32 and more.} For workload $W2$, \ksftm exceeds \svsftm, NOrec, ESTM, and MVTO by a factor of 1.5, 2, 1.6 and 1.3 respectively. For workload $W3$, \ksftm again beats \svsftm, NOrec, ESTM, and MVTO by 1.7, 3.3, 3 and 1.4 at thread count 64. So, \ksftm outperforms all the other STM algorithms in low as well as high contention.
	
	\ksftm's performance is comparable to \pkto in all workloads. In fact, in all workloads, the performance of \ksftm is 2\% less than \pkto. But as discussed in \subsecref{mvto}, that the transactions can possibly starve with \pkto while this is not the case with \ksftm. We believe that this is the overhead that one has to pay to achieve \stfdm. On the positive side, the overhead is very small, being only around 2\% as compared to \pkto. On the other hand, the performance of \ksftm is much better than single-version \stf algorithm \svsftm.
	
	We analyzed time to achieve the \stfdm for \ksftm in worst-case and compared \svsftm, NOrec, ESTM, MVTO and \pkto. High contention certainly increases probability of transactions being aborted. So we have considered varying the threads from 50 to 400, each thread invoking 1000 transaction with $K$ as 5, $C$ as 0.1, \tobj{s} as 1000 and with $W1$ workload. \figref{worstcase}.(a) shows that the worst-case time for the commit of a transaction in \ksftm is consistently better.
	
	\begin{figure}
		\captionsetup{justification=centering}
		\includegraphics[width=12cm, height=6cm]{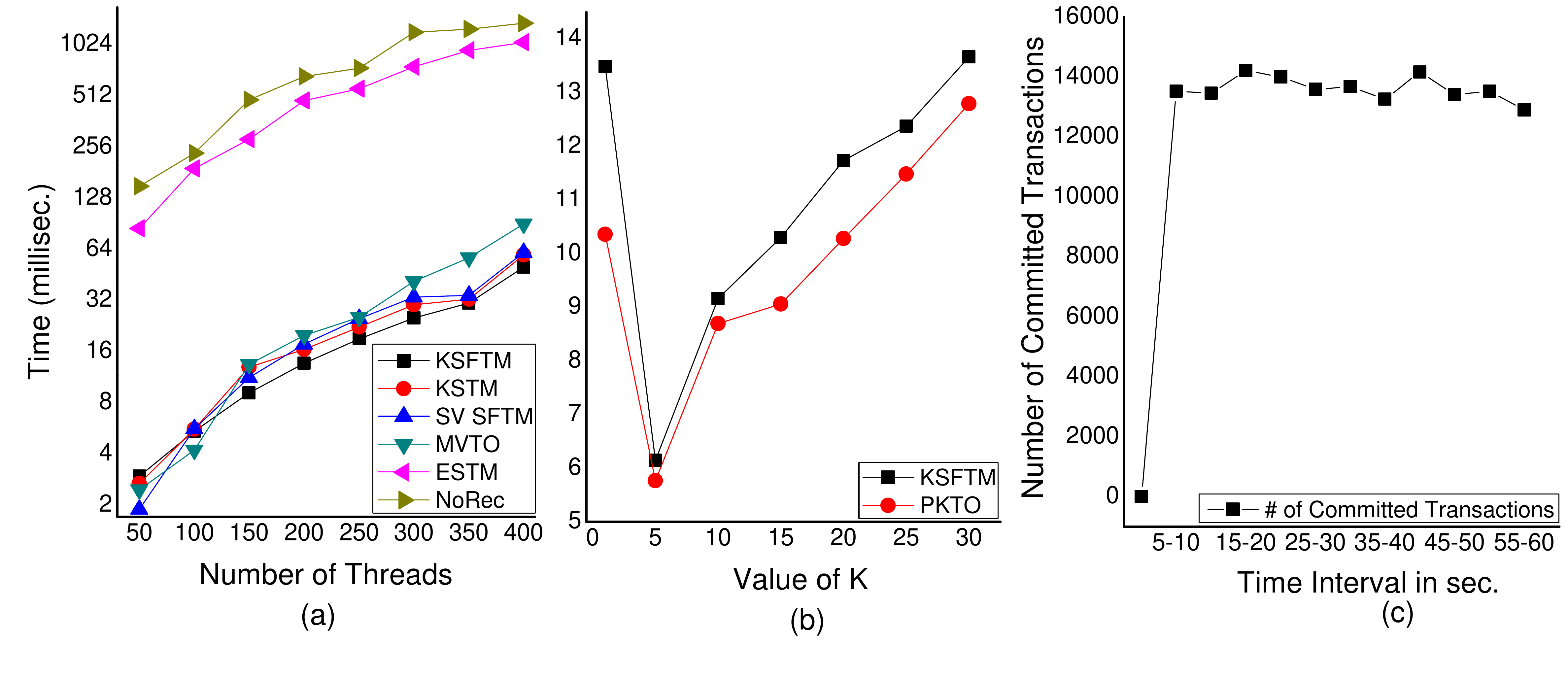}
		\vspace{-.5cm}
		\caption{Worst case time analysis, Optimal value of K and \ksftm stability}\label{fig:worstcase1}
	\end{figure}
	\vspace{-.05cm}
	
	\ignore{
	\begin{figure}
		\centering
		\begin{minipage}[b]{0.49\textwidth}
			\centering
			\includegraphics[width=6cm, height=5cm]{figs/ctime.pdf}
		\caption{Worst-Case Time Comparison}\label{fig:worst-case time1}
		\end{minipage}   
		\hfill
		\begin{minipage}[b]{0.49\textwidth}
	\includegraphics[width=7cm, height=4cm]{figs/ctrans.pdf}
			\centering
		\caption{KSFTM Stability}\label{fig:stability1}
		\end{minipage}
	\end{figure}
	}
	\cmnt{
	\begin{figure}
		\captionsetup{justification=centering}
			\includegraphics[width=12cm, height=9cm]{figs/ctime.pdf}
			\caption{\Large Worst-Case Time Comparison}\label{fig:worst-case time2}
	\end{figure}
	
	\begin{figure}
		\captionsetup{justification=centering}
			\includegraphics[width=\linewidth]{figs/ctrans.pdf}
			\caption{\Large KSFTM Stability}\label{fig:stability2}
	\end{figure}
	}
	To check the stability of our proposed algorithm \ksftm over time, we performed an experiment and shown in \figref{worstcase}.(c) where we fixed the number of threads to 32, $K$ as 5, $C$ as 0.1, \tobj{s} as 1000, along with 5 seconds warm-up period on $W1$ workload. Each thread keeps on invoking and executing transactions until $TIME$-$BOUND$ is over, which in our case is 60 seconds. We have done the experiments on number of transactions committed with increase in time as $t = t+ \epsilon$, where $\epsilon$ is fixed to 5 seconds. The experiment shows that over time \ksftm is stable which helps to hold the claim that \ksftm's performance will continue in same manner if time is increased to higher order.
	
	To show the benefits of starvation freedom utilizing the power of multiple versions in \ksftm, we tested on an application called Labyrinth provided by STAMP benchmark which has long running transactions along with high contention. \ksftm shows better performance than \pkto and \svsftm.
}

Maintaining multiple versions to increase the performance and to decrease the number of aborts, leads to creating too many versions which are not of any use and hence occupying space. So, such garbage versions need to be taken care of. Hence we come up with a garbage collection over these unwanted versions. This technique help to conserve memory space and increases the performance in turn as no more unnecessary traversing of garbage versions by transactions is necessary. We have used a global, i.e., across all transactions a list that keeps track of all the live transactions in the system. We call this list as \livel. Each transaction at the beginning of its life cycle creates its entry in this \livel. Under the optimistic approach of STM, each transaction in the shared memory performs its updates in the $\tryc$ phase. In this phase, each transaction performs some validations, and if all the validations are successful then the transaction make changes or in simple terms creates versions of the corresponding \tobj{} in the shared memory. While creating a version every transaction, check if it is the least timestamp live transaction present in the system by using \livel{} data structure, if yes then the current transaction deletes all the version of that \tobj{} and create one of its own. Else the transaction does not do any garbage collection or delete any version and look for creating a new version of next \tobj{} in the write set, if at all. 

\figref{ksftm} represents three variants of \ksftm (\mvsftm, \mvsftmgc, and \ksftm) and \figref{pkto} shows the three variants of \pkto (\pmvto, \pmvtogc, and \pkto) on all the workloads $W1$  $W2$ and $W3$. \ksftm outperforms \mvsftm and \mvsftmgc by a factor of 2.1 and 1.5. Similarly, \pkto outperforms \pmvto and \pmvtogc by a factor of 2 and 1.35. These results show that maintaining finite versions corresponding to each \tobj performs better than maintaining infinite versions and garbage collection on infinite versions corresponding to each \tobj.

\begin{figure}
	\captionsetup{justification=centering}
	\includegraphics[width=\linewidth]{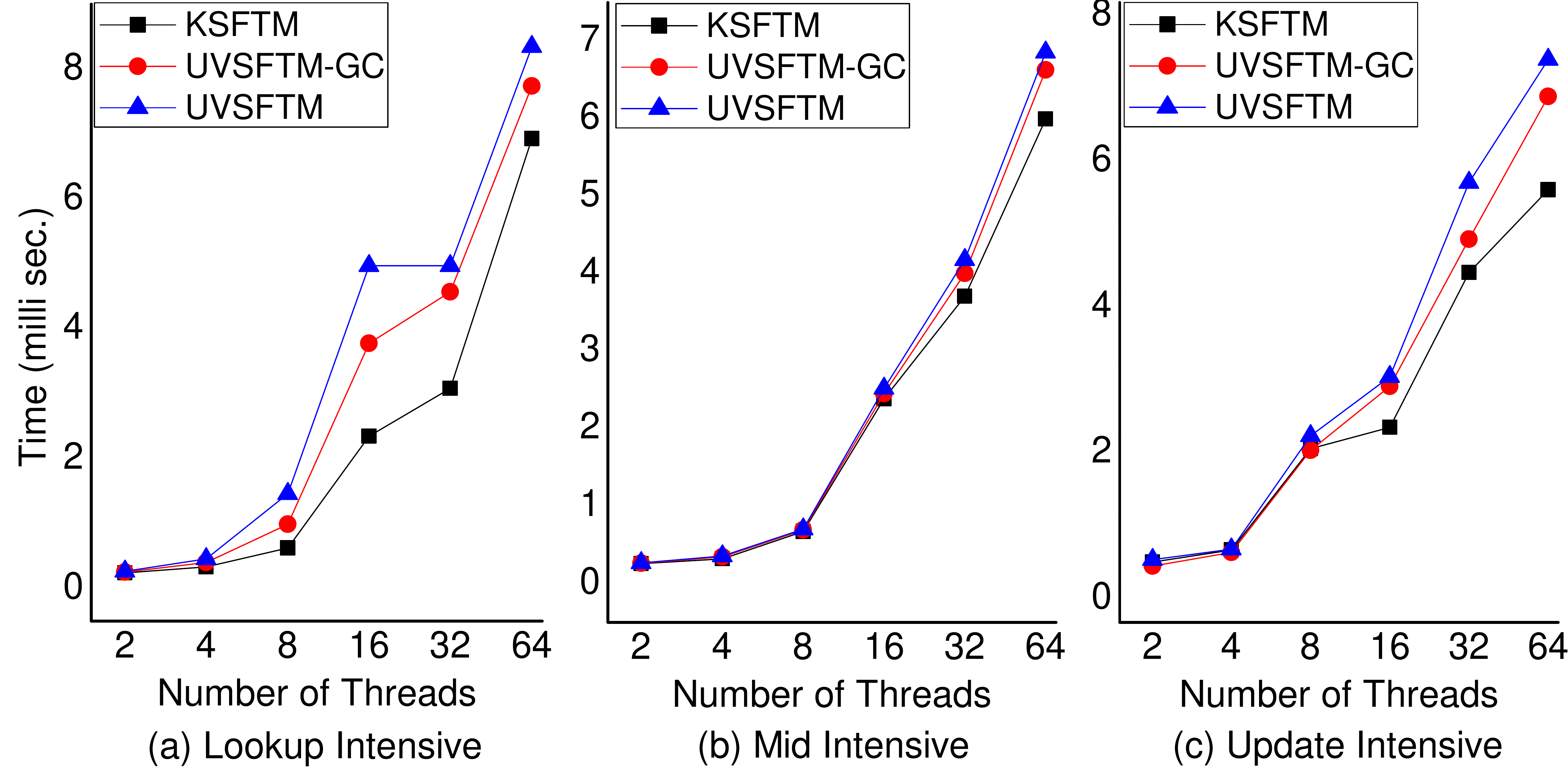}
	\caption{Time comparison among variants of \ksftm}\label{fig:ksftm}
\end{figure}
\begin{figure}
	\captionsetup{justification=centering}
	\includegraphics[width=\linewidth]{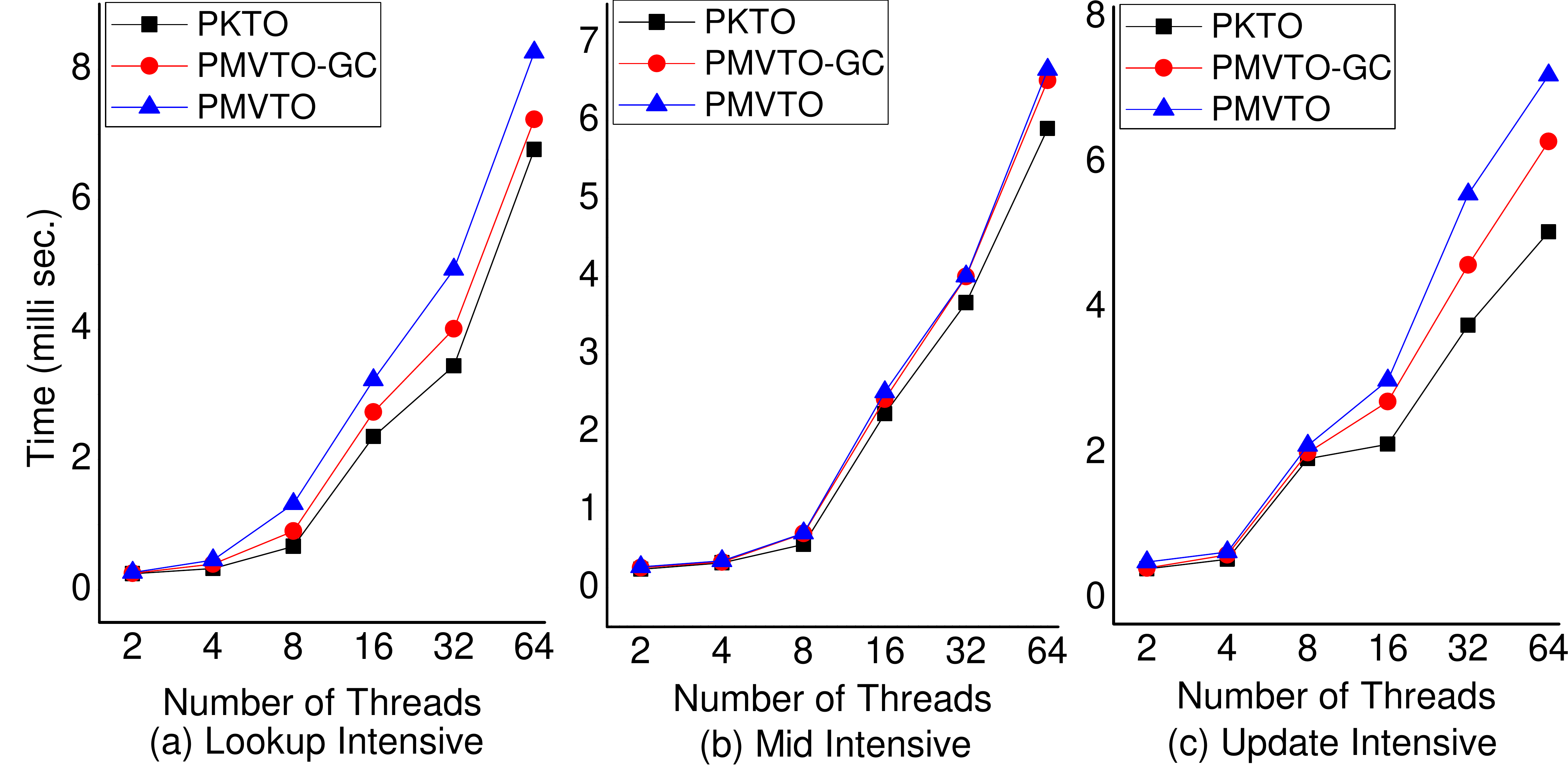}
	\caption{Time comparison among variants of \pkto}\label{fig:pkto}
\end{figure}
\begin{figure}
	\captionsetup{justification=centering}
	\includegraphics[width=\linewidth]{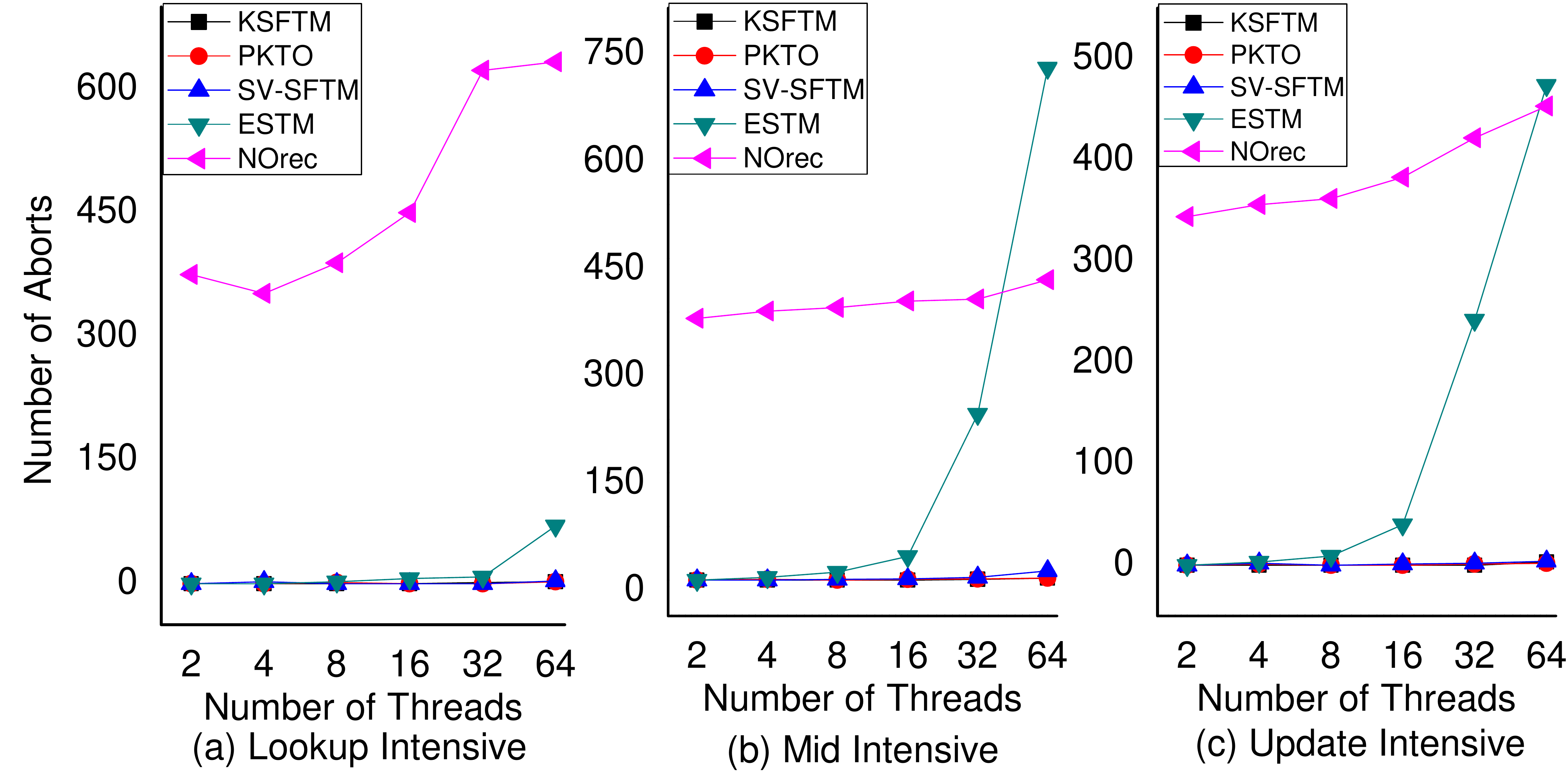}
	\caption{Abort Count on workload $W1, W2, W3$}\label{fig:aborts}
\end{figure}

\cmnt{
	\begin{figure}
		\captionsetup{justification=centering}
		\includegraphics[width=13cm, height=18cm]{figs/ksftmpkto.pdf}
		\caption{\Large Execution under variants of $KSFTM$ and $PKTO$}\label{fig:ksftmpkto}
	\end{figure}
}
\begin{figure}[H]
	\captionsetup{justification=centering}
	\includegraphics[width=\linewidth]{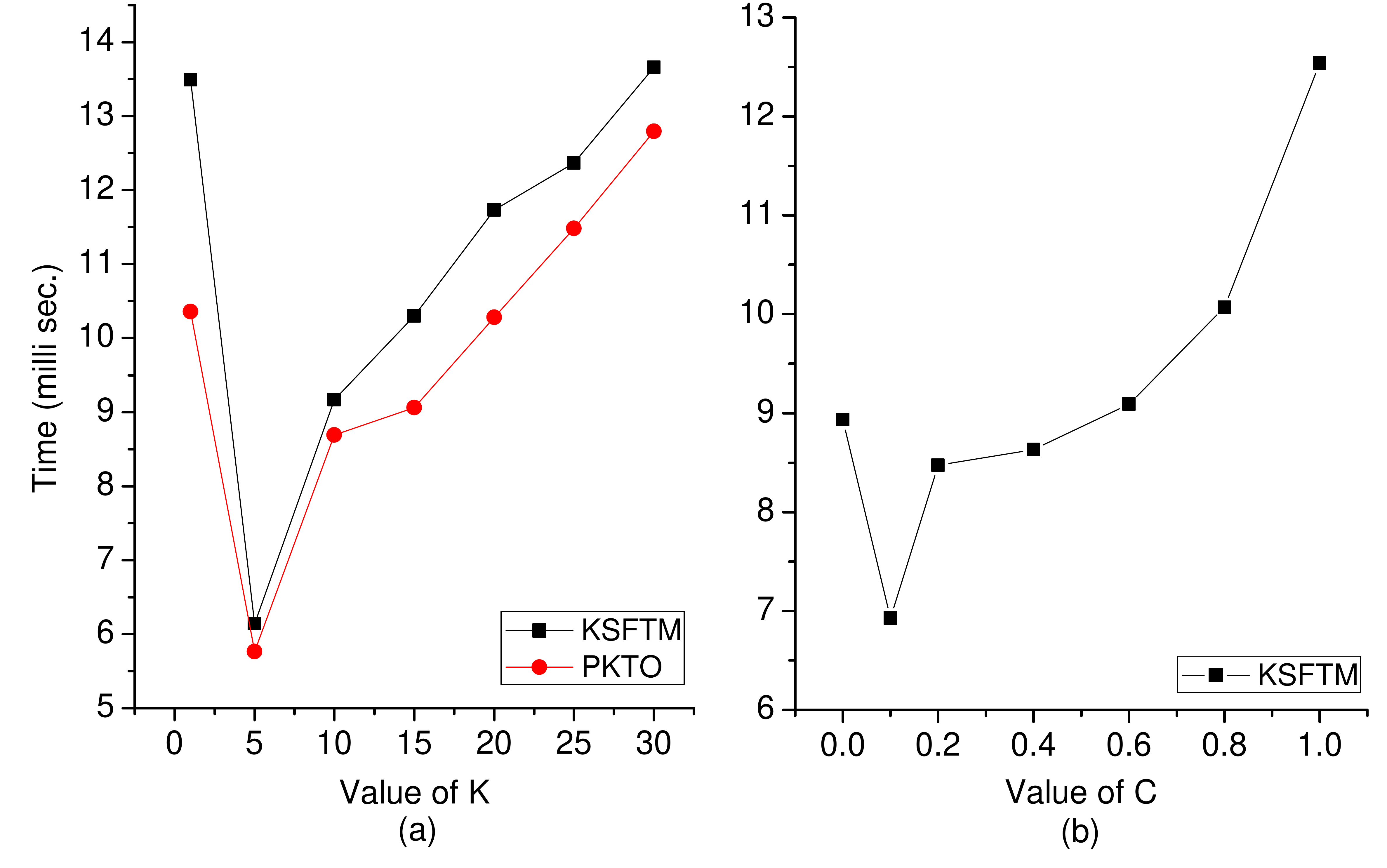}
	\caption{Best value of K and optimal value of $C$ for \ksftm}\label{fig:optimalC}
\end{figure}

\noindent
\textbf{Comparison on the basis of Abort count:}
\figref{aborts} shows the abort count comparisons of \ksftm with \pkto, ESTM,  NOrec, MVTO, and \svsftm across all workloads ($W1$, $W2$, and $W3$). The number of aborts in ESTM and NOrec are high as compared to all other STM algorithms while all other algorithms (\ksftm, \pkto, MVTO, \svsftm) have marginally small differences among them. 
\vspace{.4cm}

\noindent
\textbf{Best value of $K$ and optimal value of constant \textit{C}:} To identify the best value of K for \ksftm, we ran our experiment, varying value of K and keeping the number of threads as 64 on workload $W1$ and obtained the optimal value of $K$ in \ksftm is 5 as shown in \figref{optimalC}.(a) for counter application. Similarly, we calculate the best value of $K$ as 5 for \pkto on the same parameters. $C$, is a constant that is used to calculate $WTS$ of a transaction. i.e., $\twts{i} = \tcts{i} + C * (\tcts{i} - \tits{i});$ where, $C$ is any constant greater than 0. We run or experiments across load $W1$, for 64 threads and other parameters are same as defined in the methodology of \secref{exp}, we achieve the best value of $C$ as 0.1 for counter application. Experimental results are shown in \figref{optimalC} (b).

\cmnt{
	\textbf{Benefit of Starvation freedom and multi-versioning:}Our proposed algorithm provides the dual benefit of starvation freedom utilizing the power of multiple versions. \figref{lib} depicts that proposed algorithm KSFTM performs better than non-starvation free finite version algorithm PKVTO when tested on an application called Labyrinth provided by STAMP benchmark which has long running transactions and has very high contention. For the same application we experimented our algorithm against single version starvation freedom SV-SFTM algorithm, results shown in the \figref{lib} which supports the fact that multi-versioning provides better performance than single versioning. The experiment shows the worst case time a transaction can take over high contention and long running transactions environment, where grid size is 128 x 128 x 3 and paths to route are 64.}
\cmnt{
	\begin{figure}
		\captionsetup{justification=centering}
		\includegraphics[width=\linewidth]{figs/Lib.pdf}
		\caption{Execution under Labyrinth provided by STAMP benchmark }\label{fig:lib}
	\end{figure}
}

\cmnt{
	\begin{figure}
		\captionsetup{justification=centering}
		\includegraphics[width=\linewidth]{figs/abortc.pdf}
		\caption{ Aborts on workload $W1, W2, W3$ and Optimal value of C as 0.1}\label{fig:optimalC}
	\end{figure}
}
\subsection{Pseudo code of Counter Application}
\label{subsec:countercode}
OP\_LT\_SEED is defined as number of operations per transaction, T\_OBJ\_SEED is defined as number of transaction objects in the system, TRANS\_LT defines the total number of transactions to be executed in the system, and READ\_PER is the percentage of read operation which is used to define various workloads.
\label{apn:conters}
\begin{algorithm}  
	\caption{$main()$: The main procedure invoked by counter application} \label{algo:main} 
	\begin{algorithmic}[1]
	\State \Comment{To log abort counts by each thread}
		\State $abort\_count[$NUMTHREADS$]$ 
		\State \Comment{To log average time taken by each transaction to commit}
		\State $time\_taken[$NUMTHREADS$]$ 
		\State \Comment{To log the time of longest running transaction by each thread, worst case time}
		\State $worst\_time[$NUMTHREADS$]$
		
		\For{(i = 0 : NUMTHREADS)}
		    \State pthread\_create(\&threads[i], NULL, testFunc\_helper,(void$\ast$)args)
		\EndFor
		\For{(i = 0 : NUMTHREADS)}
		    \State pthread\_join(threads[i], \&status)
		\EndFor
		
		\State $max\_worst\_time = 0.0$
		\State $total\_abort\_count = 0$
		\State $average\_time_taken = 0$
		\For{(i = 0 : NUMTHREADS)}
		    \If{($max\_worst\_time < worst\_time[i]$)}
		    \State $max\_worst\_time = worst\_time[i]$
		    \EndIf
		    \State $total\_abort\_count += abort\_count[i]$
		    \State $average\_time\_taken += time\_taken[i]$
		\EndFor
	\end{algorithmic}
\end{algorithm}

\vspace{1mm}
\begin{algorithm} [H] 
	\caption{$testFunc\_helper()$:Function invoked by threads} \label{algo:testFunc} 
	\begin{algorithmic}[1]
	    \State $transaction\_count = 0$
		\While{(TRANS\_LT)}
		    \State\Comment{Log the time at the start of every transaction}
		    \State $begin\_time = time\_request()$
		    \State\Comment{Invoke the test function to execute a transaction}
		    \State $abort\_count[thread\_id] = test\_function()$
		    \State $transaction\_count++$
		    			\algstore{myalg}
		\end{algorithmic}
	\end{algorithm}
	\begin{algorithm}
		\begin{algorithmic}
			\algrestore{myalg}
		    \State\Comment{Log the time at the end of every transaction}
		    \State $end\_time = time\_request()$
		    \State $time\_taken[thread\_id] += (end\_time - begin\_time)$
		    \If{($worst\_time[thread_id] < (end\_time - begin\_time)$)}
		    \State $worst\_time[thread_id] = (end\_time - begin\_time)$
		    \EndIf
		    \State TRANS\_LT -= 1
		\EndWhile
		\State $time\_taken[thread\_id]$ /= $transaction\_count$
	\end{algorithmic}
\end{algorithm}

\vspace{1mm}
\begin{algorithm} [H] 
	\caption{$test\_function()$:main test function while executes a transaction} \label{algo:testFunx} 
	\begin{algorithmic}[1]
    \State Transaction $\ast$T = new Transaction;
    \State $T\rightarrow g\_its$ = NIL
    \State $local\_abort\_count$ = 0
    \State label:
    \While{(true)}
    \If{($T\rightarrow g\_its$ != $NIL$)}
    \State $its = T\rightarrow g\_its$
    \State $T = lib\rightarrow stm$-$begin(its)$
	\Else
	\State $T = lib\rightarrow stm$-$begin(T\rightarrow g\_its)$
    \EndIf
    \ForAll{(OP\_LT\_SEED)}
    \State $t\_obj = rand()\%T\_OBJ\_SEED$
    \State $randVal = rand()\%OP\_SEED$
    \If{($randVal <= READ\_PER $)}
     \State $stm$-$read(t\_obj, value)$
        \If{(value == $ABORTED$)}
        \State $local\_abort\_count$++
        \State goto label
        \EndIf
    \Else
    \State $stm$-$write(t\_obj, value)$
    \EndIf
    \EndFor
    \If{($lib\rightarrow stm$-$tryC() == ABORTED$)}
    \State $local\_abort\_count$++
    \State continue
    \EndIf
    \State break
    \EndWhile
	\end{algorithmic}
\end{algorithm}
\section{Graph Characterization of Local Opacity \& \ksftm Correctness}
\label{sec:ap-gphchar}


To prove correctness of STM systems, it is useful to consider graph characterization of histories. In this section, we describe the graph characterization developed by Kumar et al \cite{Kumar+:MVTO:ICDCN:2014} for proving \opty which is based on characterization by Bernstein and Goodman \cite{BernGood:1983:MCC:TDS}.  We extend this characterization for \lo. 

Consider a history $H$ which consists of multiple versions for each \tobj. The graph characterization uses the notion of \textit{version order}. Given $H$ and a \tobj{} $x$, we define a version order for $x$ as any (non-reflexive) total order on all the versions of $x$ ever created by committed transactions in $H$. It must be noted that the version order may or may not be the same as the actual order in which the version of $x$ are generated in $H$. A version order of $H$, denoted as $\ll_H$ is the union of the version orders of all the \tobj{s} in $H$. 

Consider the history $H2: r_1(x, 0) r_2(x, 0) r_1(y, 0) r_3(z, 0) w_1(x, 5) w_3(y, 15)  w_2(y, 10) w_1(z, 10) \\
c_1 c_2 r_4(x, 5) r_4(y, 10) w_3(z, 15) c_3 r_4(z, 10)$. Using the notation that a committed transaction $T_i$ writing to $x$ creates a version $x_i$, a possible version order for $H2$ $\ll_{H2}$ is: $\langle x_0 \ll x_1 \rangle, \langle y_0 \ll y_2 \ll y_3 \rangle, \langle z_0 \ll z_1 \ll z_3 \rangle $. 

We define the graph characterization based on a given version order. Consider a history $H$ and a version order $\ll$. We then define a graph (called opacity graph) on $H$ using $\ll$, denoted as $\opg{H}{\ll} = (V, E)$. The vertex set $V$ consists of a vertex for each transaction $T_i$ in $\overline{H}$. The edges of the graph are of three kinds and are defined as follows:

\begin{enumerate}	
	\item \textit{\rt}(real-time) edges: If $T_i$ commits before $T_j$ starts in $H$, then there is an edge from $v_i$ to $v_j$. This set of edges are referred to as $\rtx(H)$.
	
	\item \textit{\rf}(reads-from) edges: If $T_j$ reads $x$ from $T_i$ in $H$, then there is an edge from $v_i$ to $v_j$. Note that in order for this to happen, $T_i$ must have committed before $T_j$ and $c_i <_H r_j(x)$. This set of edges are referred to as $\rf(H)$.
	
	\item \textit{\mv}(multiversion) edges: The \mv{} edges capture the multiversion relations and is based on the version order. Consider a successful read \op{} $r_k(x,v)$ and the write \op{} $w_j(x,v)$ belonging to transaction $T_j$ such that $r_k(x,v)$ reads $x$ from $w_j(x,v)$ (it must be noted $T_j$ is a committed transaction and $c_j <_H r_k$). Consider a committed transaction $T_i$ which writes to $x$, $w_i(x, u)$ where $u \neq v$. Thus the versions created $x_i, x_j$ are related by $\ll$. Then, if $x_i \ll x_j$ we add an edge from $v_i$ to $v_j$. Otherwise ($x_j \ll x_i$), we add an edge from $v_k$ to $v_i$. This set of edges are referred to as $\mv(H, \ll)$.
\end{enumerate}

Using the construction, the $\opg{H2}{\ll_{H2}}$ for history $H2$ and $\ll_{H2}$ is shown in \figref{opg}. The edges are annotated. The only \mv{} edge from $T4$ to $T3$ is because of \tobj{s} $y, z$. $T4$ reads value 5 for $z$ from $T1$ whereas $T3$ also writes 15 to $z$ and commits before $r_4(z)$. 

\begin{figure}[tbph]
	\centerline{\scalebox{0.7}{\input{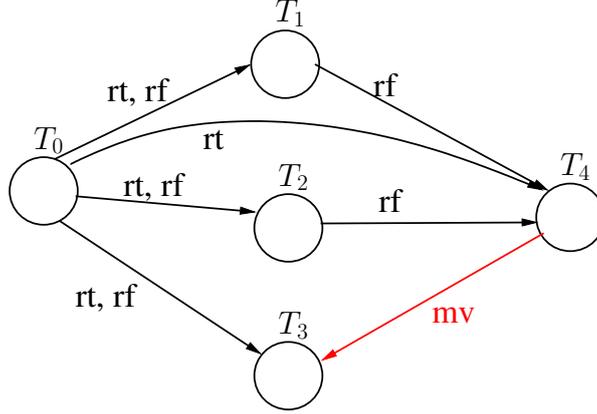}}}
	\captionsetup{justification=centering}
	\caption{$\opg{H2}{\ll_{H2}}$}
	\label{fig:opg}
\end{figure}

Kumar et al \cite{Kumar+:MVTO:ICDCN:2014} showed that if a version order $\ll$ exists for a history $H$ such that $\opg{H}{\ll_H}$ is acyclic, then $H$ is \opq. This is captured in the following result.

\ignore{
	\begin{result}
		\label{res:main-opg}
		A \valid{} history $H$ is opaque iff there exists a version order $\ll_H$ such that $\opg{H}{\ll_H}$ is acyclic.
	\end{result}
	
	\noindent This result can be extended to characterize \lo using graphs with the following theorem. The proof is in Appendix \thmref{log}. 
	
	\begin{theorem}
		\label{thm:main-log}
		A \valid{} history $H$ is \lopq iff for each sub-history $sh$ in $\shset{H}$ there exists a version order $\ll_{sh}$ such that $\opg{sh}{\ll_{sh}}$ is acyclic. Formally, $\langle (H \text{ is \lopq}) \Leftrightarrow (\forall sh \in \shset{H}, \exists \ll_{sh}: \opg{sh}{\ll_{sh}} \text{ is acyclic}) \rangle$. 
	\end{theorem}
}

\begin{result}
	\label{res:opg}
	A \valid{} history $H$ is opaque iff there exists a version order $\ll_H$ such that $\opg{H}{\ll_H}$ is acyclic.
\end{result}

\noindent This result can be easily extended to prove \lo as follows

\begin{theorem}
	\label{thm:log}
	A \valid{} history $H$ is \lopq iff for each sub-history $sh$ in $\shset{H}$ there exists a version order $\ll_{sh}$ such that $\opg{sh}{\ll_{sh}}$ is acyclic. Formally, $\langle (H \text{ is \lopq}) \Leftrightarrow (\forall sh \in \shset{H}, \exists \ll_{sh}: \opg{sh}{\ll_{sh}} \text{ is acyclic}) \rangle$. 
\end{theorem}

\begin{proof}
	To prove this theorem, we have to show that each sub-history $sh$ in $\shset{H}$ is \valid. Then the rest follows from \resref{opg}. Now consider a sub-history $sh$. Consider any read \op $r_i(x, v)$ of a transaction $T_i$. It is clear that $T_i$ must have read a version of $x$ created by a previously committed transaction. From the construction of $sh$, we get that all the transaction that committed before $r_i$ are also in $sh$. Hence $sh$ is also \valid. 
	
	Now, proving $sh$ to be \opq iff there exists a version order $\ll_{sh}$ such that $\opg{sh}{\ll_{sh}}$ is acyclic follows from \resref{opg}. 
\end{proof}

\ignore {

Using this theorem, we can give the proof sketch of \ksftm algorithms. Here for simplicity, we assume that the history generated is sequential. 

\begin{theorem}
	\label{thm:ap1-ksftm-lo} 
	Any history generated by \ksftm{} is \lopq.
\end{theorem}

\begin{proof}
For proving this, we consider a sequential history $H$ generated by \ksftm. We define the version order $\ll_{\vt}$: for two versions $v_i, v_j$ it is defined as

$(v_i \ll_{\vt} v_j) \equiv (v_i.\vt < v_j.\vt)$

\noindent Using this version order $\ll_{\vt}$, we can show that all the sub-histories in $\shset{H}$ are acyclic.

\end{proof}

}
\begin{lemma}
	\label{lem:tltl-edge}
	Consider a history $H$ in $\gen{\ksftm}$ with two transactions $T_i$ and 
	$T_j$ such that both their G\_valid flags are true.  there is an edge from 
	$T_i$ $\rightarrow$ $T_j$ then $\tltl_i$ $<$ $\tltl_j$.		
\end{lemma}

\begin{proof}
	There are three types of possible edges  in MVSG.
	\begin{enumerate}
		\item Real-time edge: Since, transaction $T_i$ and $T_j$  are in real 
		time 
		order so $\ct_i$ $<$ $\gcts_j$. As we know from \lemref{ti|tltl-comt} 
		$(\tltl_i 
		\leq \ct_i)$. So, $(\tltl_i \leq \cts_j)$.\\
		We know from STM $\begt(its)$ method, $\tltl_j = \gcts_j$.\\
		Eventually, $\tltl_i$ $<$ $\tltl_j$.	
		\item Read-from edge: Since, transaction $T_i$ has been committed and $T_j$ is reading from 
		$T_i$ so, 
		from \Lineref{new-tup} 
		$\tryc(T_i)$, $\tltl_i = \vltl_i$.\\ 
		  and from \Lineref{rd-tltl-inc} STM $read(j, x)$, $\tltl_j = 
		max(\tltl_j,\\ x[curVer].
		\vltl + 1)$ $\Rightarrow$ $(\tltl_j > \vltl_i)$ $\Rightarrow$ $(\tltl_j 
		> \tltl_i)$  \\
		Hence, $\tltl_i$ $<$ $\tltl_j$.
		\item Version-order edge: Consider a triplet $w_j(x_j) r_k(x_j) 
		w_i(x_i)$ in 
		which there are two possibilities of version order:
		\begin{enumerate}
			\item i $\ll$ j $\Longrightarrow$ $\gwts_i < \gwts_j$ 
			\\ There are two possibilities of commit order:
			\begin{enumerate}
				\item $\ct_i <_H \ct_j$: Since, $T_i$ has been committed before 
				$T_j$ so 
				$\tltl_i = \vltl_i$. From \Lineref{tryc-tltl-inc} of 
				$\tryc(T_j)$, $\vltl_i < 
				\tltl(j)$.\\
				Hence, $\tltl_i$ $<$ $\tltl_j$.
				\item  $\ct_j <_H \ct_i$: Since, $T_j$ has been committed 
				before $T_i$ so 
				$\tltl_j = \vltl_j$. From 
				\Lineref{tryc-ul-dec} of $\tryc(T_i)$, $\tutl_i < \vltl_j$. As 
				we have assumed 
				$\gval_i$ is true so definitely it will execute the 
				\Lineref{ti-updt} 
				$\tryc(T_i)$ i.e. $\tltl_i = \tutl_i$.\\
				Hence, $\tltl_i$ $<$ $\tltl_j$.
			\end{enumerate}
			\item j $\ll$ i $\Longrightarrow$ $\gwts_j < \gwts_i$ 
			\\ Again, there are two possibilities of commit order:
			\begin{enumerate}
				\item $\ct_j <_H \ct_i$: Since, $T_j$ has been committed before 
				$T_i$ and $T_k$ 
				read from $T_j$. There can be two possibilities $\gwts_k$.
				\begin{enumerate}
					\item $\gwts_k > \gwts_i$: That means $T_k$ is in largeRL 
					of $T_i$. From 
					\Lineref{addAbl-lar} to \Lineref{its-chk1}of $\tryc(i)$, either transaction $T_k$ or $T_i$,   $\gval$ 
					flag is set to be false. If $T_i$ returns abort then this case will not be considered in 
					\lemref{tltl-edge}. Otherwise, as $T_j$ has 	already been committed and later $T_i$ will execute the \Lineref{new-tup} $\tryc(T_i)$, 									
					Hence, $\tltl_j < \tltl_i$.\\
										
					\item $\gwts_k < \gwts_i$: That means $T_k$ is in smallRL 
					of $T_i$. From 
					\Lineref{rd-ul-dec} of $read(k, x)$, $\tutl_k < \vltl_i$ 
					and from 
					\Lineref{rd-tltl-inc} of $read(k, x)$, $\tltl_k > \vltl_j$. 
					Here, $T_j$ has 
					already been committed so, $\tltl_j = \vltl_j$. As we have 
					assumed 
					$\gval_i$ is true so definitely it will execute the 
					\Lineref{new-tup} 
					$\tryc(T_i)$, $\tltl_i = \vltl_i$.\\ 
					So, $\tutl_k < \tltl_i $ and $\tltl_k > \tltl_j$.
					While considering $\gval_k$ flag is true $\rightarrow$ 
					$\tltl_k < \tutl_k$.\\
					Hence, $\tltl_j < \tltl_k < \tutl_k < \tltl_i$.\\
					Therefore, $\tltl_j < \tltl_k < \tltl_i$.
				\end{enumerate}
				\item $\ct_i <_H \ct_j$: Since, $T_i$ has been committed before 
				$T_j$ so, $\tltl_i = \vltl_i$. From 
				\Lineref{tryc-ul-dec} of $\tryc(T_j)$, $\tutl_j < \vltl_i$ i.e. 
				$\tutl_j < \tltl_i$. Here, $T_k$ 
				read from $T_j$. So, From \Lineref{rd-ul-dec} of $read(k, x)$, 
				$\tutl_k < \vltl_i$ $\rightarrow$ $\tutl_k < \tltl_i$ from 
				\Lineref{rd-tltl-inc} of $read(k, x)$, $\tltl_k > \vltl_j$. As 
				we have assumed 
				$\gval_j$ is true so definitely it will execute the 
				\Lineref{new-tup} 
				$\tryc(T_j)$, $\tltl_j = \vltl_j$.\\
				Hence, $\tltl_j < \tltl_k < \tutl_k < \tltl_i$.\\
				Therefore, $\tltl_j < \tltl_k < \tltl_i$.
			\end{enumerate}
		\end{enumerate}
		\cmnt{Due to acquiring the lock on each dataitem before 
			creating a version So, let say $T_i$ created a version then release 
			it and 
			return commit. For both the transactions G\_valid 
			flags are true. As we know from \lemref{tltl commit} $(\tltl_i 
			\leq \ct_i)$. After creating a version by $T_i$, transaction $T_j$ 
			wants to 
			create a version of the same dataitem then definitely, $(\ct_i < 
			\ct_j)$. Again 
			from \lemref{tltl commit} on $T_j$, $(\tltl_j \leq \ct_j)$. \\
			So, $\tltl_i$ $<$ $\tltl_j$.
		} 
	\end{enumerate} 
	
\end{proof}

\cmnt{            
\begin{theorem}
	\label{thm:trans-com|abt}
	Transaction with lowest $\its$ value will eventually have the highest 
	$\gwts$ value.
\end{theorem}
}

\begin{theorem}
Any history H gen(KSFTM) is local opaque iff for a given version order $\ll$ H, 
 MVSG(H,$\ll$) is acyclic. 
\end{theorem}
\begin{proof}
We are proving it by contradiction, so Assuming MVSG(H,$\ll$) has cycle.
From \lemref{tltl-edge}, For any two transactions $T_i$ and 
	$T_j$ such that both their G\_valid flags are true and if there is an edge 
	from $T_i$ $\rightarrow$ $T_j$ then $\tltl_i$ $<$ $\tltl_j$. While 
	considering transitive case for k transactions $T_1, T_2, T_3...T_k$ such 
	that G\_valid flags of all the transactions are true. if there is an edge 
		from $T_1$ $\rightarrow$ $T_2$ $\rightarrow$ $T_3$ 
		$\rightarrow$....$\rightarrow$ $T_k$ then $\tltl_1 $ $<$ $\tltl_2$ $<$ 
		$\tltl_3$ $<$ ....$<$ $\tltl_k$.\\
		Now, considering our assumption, MVSG(H,$\ll$) has cycle so, $T_1$ 
		$\rightarrow$ $T_2$ $\rightarrow$ $T_3$ $\rightarrow$....$\rightarrow$ 
		$T_k$ $\rightarrow$ $T_1$ that implies $\tltl_1 $ $<$ $\tltl_2$ $<$ 
		$\tltl_3$ $<$ ....$<$ $\tltl_k$ $<$ $\tltl_1$.\\
		Hence from above assumption, $\tltl_1$ $<$ $\tltl_1$ but this is 
		impossible. So, our assumption is wrong.\\
		Therefore, MVSG(H,$\ll$) produced by KSFTM is acyclic.
\end{proof}

\textbf{\textit{M\_Order$_H$:}} It stands for method order of history H in which methods of transactions are interval (consists of invocation and response of 
a method) instead of dot (atomic). Because of having method as an interval, methods of different transactions can overlap. To prove the correctness \textit{(local 
opacity)} of our algorithm, we need to order the overlapping methods.

Let say, there are two transactions $T_i$ and $T_j$ either accessing common (t-objects/$\glock$) or $\gtcnt$ through operations $op_i$ and $op_j$ respectively. If res($op_i$) $<_H$ 
inv($op_j$) then $op_i$ and $op_j$ are in real-time order in H. So, the 
\textit{M\_Order$_H$} is $op_i \rightarrow op_j$.

If operations are overlapping and either accessing common t-objects or sharing $\glock$: 
\begin{enumerate}
\item $read_i(x)$ and $read_j(x)$: If $read_i(x)$ acquires the lock on x  
before $read_j(x)$ then the \textit{M\_Order$_H$} is $op_i \rightarrow 
op_j$.
\item $read_i(x)$ and $\tryc_j()$: If they are accessing common t-objects then, let say $read_i(x)$ acquires the lock on x  before 
$\tryc_j()$ then the \textit{M\_Order$_H$} is $op_i \rightarrow 
op_j$. Now if they are not accessing common t-objects but sharing $\glock$ then, let say $read_i(x)$ acquires the lock on $\glock_i$  before 
$\tryc_j()$ acquires the lock on $\relll$ (which consists of $\glock_i$ and $\glock_j$) then the \textit{M\_Order$_H$} is $op_i \rightarrow 
op_j$.
\item  $\tryc_i()$ and $\tryc_j()$: If they are accessing common t-objects then, let say $\tryc_i()$  acquires the lock on x  before 
$\tryc_j()$ then the \textit{M\_Order$_H$} is $op_i \rightarrow 
op_j$. Now if they are not accessing common t-objects but sharing $\glock$ then, let say $\tryc_i()$  acquires the lock on $\relll_i$  before 
$\tryc_j()$ then the \textit{M\_Order$_H$} is $op_i \rightarrow 
op_j$.

\end{enumerate}

If operations are overlapping and accessing different t-objects but sharing $\gtcnt$ counter:
\begin{enumerate}
	\item $\begt_i$ and $\begt_j$: Both the $\begt$ are accessing shared 
	counter variable $\gtcnt$. If $\begt_i$ executes $\gtcnt.get\&Inc()$ 
	before $\begt_j$ then the \textit{M\_Order$_H$} is $op_i \rightarrow op_j$.
	\item $\begt_i$ and $\tryc(j)$: If $\begt_i$ executes $\gtcnt.get\&Inc()$ 
	before $\tryc(j)$ then the \textit{M\_Order$_H$} is $op_i \rightarrow op_j$.
	
\end{enumerate}

\textit{Linearization:} The history generated by STMs are generally not sequintial because operations of the transactions are overlapping. The correctness of STMs is defined on sequintial history, inorder to show history generated by our algorithm is correct we have to consider sequintial history. We have enough information to order the overlapping methods, after ordering the operations will have equivalent sequintial history, the total order of the operation is called linearization of the history.\\

\textit{Operation graph (OPG):}  Consider each operation as a vertex and edges 
as below:
\begin{enumerate}
	
    \item Real time edge: If response of operation $op_i$ happen before the invocation of operation $op_j$ i.e. rsp($op_i$) $<_H$ inv($op_j$) then there exist real time edge between $op_i$ $\rightarrow$ $op_j$.
    
    \item Conflict edge: It is based on $L\_Order_H$ which depends on three conflicts:
    \begin{enumerate}
    		\item Common \textit{t-object}: If two operations $op_i$ and $op_j$  are overlapping and accessing common \textit{t-object x}. Let	say $op_i$ acquire lock first on x then $L\_Order.op_i$(x) $<_H$ $L\_Order.op_j$(x) so, conflict edge is $op_i$ $\rightarrow$ $op_j$.
    		
    		\item Common $\gval$ flag: If two operation $op_i$ and $op_j$  are 
    		overlapping but accessing common $\gval$ flag instead of 
    		\textit{t-object}. Let	say $op_i$ acquire lock first on $\gval_i$ 
    		then $L\_Order.op_i$(x) $<_H$ $L\_Order.op_j$(x) so, conflict edge 
    		is $op_i$ $\rightarrow$ $op_j$.
    		\end{enumerate}
    		
    		\item Common $\gtcnt$ counter: If two operation $op_i$ and $op_j$  are overlapping but accessing common $\gtcnt$ counter instead of \textit{t-object}. Let	say $op_i$ access  $\gtcnt$ counter before $op_j$ then $L\_Order.op_i$(x) $<_H$ $L\_Order.op_j$(x) so, conflict edge is $op_i$ $\rightarrow$ $op_j$.

\end{enumerate}

\cmnt{
\begin{lemma}
	Any history H gen(KSFTM) follows strict partial order of all the locks in H 
	($lockOrder_H$) so, operation graph (OPG(H)) is acyclic. i.e.
\end{lemma}
\begin{enumerate}
\item \textrm{If ($p_i$, $p_j$) is an edge in OPG, then $Fpu_i$($\alpha$) $<$ 
$Lpl_j$($\alpha$) for some data item $\alpha$ and two operations 
$p_i$($\alpha$), $p_j$($\alpha$) in conflict.}
\item \textrm{If ($p_1$,$p_2$, ... ,$p_n$) is a path in OPG, n $\geq$ 1, then 
$Fpu_1$($\alpha$) $<$ $Lpl_n$($\gamma$) for two data items $\alpha$ and 
$\gamma$ as well as operations $p_1$($\alpha$) and $p_n$($\gamma$).}
\item \textrm{OPG is acyclic.}
\end{enumerate}

\begin{proof} 
\textrm{We assume variables $\alpha$, $\beta$ and $\gamma$ can be data item or 
G\_Valid.}

\begin{enumerate}
\item \textrm{If ($p_i$, $p_j$) is an edge in operation graph OPG, then OPG 
comprises two steps $p_i$($\alpha$) and $p_j$($\alpha$) in conflict such that 
$p_i$($\alpha$) $<$ $p_j$($\alpha$). According to (If $o_i$(x)(o $\in$\{r, w\}) 
occurs in graph, then so do $ol_i$($\alpha$) and $ou_i$($\alpha$) with the 
sequencing $ol_i$($\alpha$) $<$ $o_i$($\alpha$) $<$ $ou_i$($\alpha$).),  this 
implies $pl_i$($\alpha$) $<$ $p_i$($\alpha$) $<$ $pu_i$($\alpha$) and 
$pl_j$($\alpha$) $<$ $p_j$($\alpha$) $<$ $pu_j$($\alpha$). According to (If 
some steps $p_i$($\alpha$) and $p_j$($\alpha$) from graph are in conflict, then 
either $Fpu_i$($\alpha$) $<$ $Lpl_j$($\alpha$) or $Fpu_j$($\alpha$) $<$ 
$Lpl_i$($\alpha$) holds.), we moreover find (a) $Fpu_i$($\alpha$) $<$  
$Lpl_j$($\alpha$) or (b) $Fpu_j$($\alpha$) $<$  $Lpl_i$($\alpha$). Case (b) 
means $Lpl_j$($\alpha$) $<$ $p_j$($\alpha$) $<$ $pu_j$($\alpha$) $<$ 
$pl_i$($\alpha$) $<$ $p_i$($\alpha$) $<$ $Fpu_i$($\alpha$) and hence 
$p_j$($\alpha$) $<$ $p_i$($\alpha$), a contradiction to $p_i$($\alpha$) $<$ 
$p_j$($\alpha$). Thus, $Fpu_i$($\alpha$) $<$  $Lpl_j$($\alpha$).}

\item \textrm{In this case we are considering transitive order. In order to 
prove it, we are taking induction on n (1): If ($p_1$, $p_2$) is an edge in 
OPG, 
there is a conflict between $p_1$ and $p_2$. Thus, $Fpu_1$($\alpha$) $<$ 
$Lpl_2$($\alpha$), i.e., $p_1$ unlocks $\alpha$ before $p_2$ locks $\alpha$. In 
other words, when $p_2$ sets a lock, $p_1$ has already released one. Now we are 
assuming it is true for n transactions on the path ($p_1$,$p_2$, ... ,$p_n$) in 
OPG and we need to prove, it holds for path n+1. The inductive assumption now 
tells us that there are data items $\alpha$ and $\beta$ such that 
$Fpu_1$($\alpha$) $<$  $Lpl_n$($\beta$) in S . Since ($p_n$, $p_{n+1}$) is an 
edge in OPG, it follows from (1) above that for operations $p_n$($\gamma$) and 
$p_{n+1}$($\gamma$) in conflict we have $Fpu_n$($\gamma$) $<$ 
$Lpl_{n+1}$($\gamma$). According to (If $p_i$($\alpha$) and $p_i$($\gamma$) are 
in graph, then $pl_i$($\alpha$) $<$ $pu_i$($\gamma$), i.e., every lock 
operation occurs before every unlock operation of the same transaction), this 
implies $pl_n$($\beta$) $<$ $pu_n$($\gamma$) and hence $Fpu_1$($\alpha$) $<$  
$Lpl_{n+1}$($\gamma$).}

\item \textrm{Proof by contradiction: Assuming OPG is cyclic. So there exists a 
cycle ($p_1$,$p_2$, ... ,$p_n$, $p_1$), n $\geq$ 1. By using (2), 
$Fpu_1$($\alpha$) 
$<$ $Lpl_1$(($\gamma$) for operations $p_1$($\alpha$), $p_1$($\gamma$), a 
contradiction to the KSFTM (If $p_i$($\alpha$) and $p_i$(($\gamma$) are 
in OPG, then $pl_i$($\alpha$) $<$ $pu_i$(($\gamma$), i.e., every lock operation 
occurs before every unlock operation of the same transaction).}
\end{enumerate}	
\end{proof}
}

\begin{lemma}
All the locks in history H ($L\_Order_H$) gen(KSFTM) follows strict partial 
order. So, operation graph (OPG(H)) is acyclic. If ($op_i$$\rightarrow$$op_j$) 
in OPG, then atleast one of them will definitely true: ($Fpu_i$($\alpha$) 
$<$ 
$Lpl\_op_j$($\alpha$)) $\cup$ ($access.\gtcnt_i$ $<$ $access.\gtcnt_j$) $\cup$
($Fpu\_op_i$($\alpha$) $<$ $access.\gtcnt_j$) $\cup$ ($access.\gtcnt_i$ $<$ 
$Lpl\_op_j$($\alpha$)). Here, $\alpha$ can either be t-object or $\gval$.
\end{lemma}
\begin{proof} we consider proof by induction, So we assummed there exist a path from $op_1$ to $op_n$ and there is an edge between $op_n$ to $op_{n+1}$. As we described, while constructing OPG(H) we need to consider 
three types of edges. We are considering one by one:
\begin{enumerate}
\item Real time edge between $op_n$ to $op_{n+1}$:

\begin{enumerate}
	\item $op_{n+1}$ is a locking method: In this we are considering all the possible path between $op_1$ to $op_n$:
\begin{enumerate}
	\item ($Fu\_op_1$($\alpha$) $<$ $Ll\_op_n$($\alpha$)): Here, ($Fu\_op_n$($\alpha$) $<$ $Ll\_op_{n+1}$($\alpha$)).\\
	So, ($Fu\_op_1$($\alpha$) $<$ $Ll\_op_n$($\alpha$)) $<$ ($Fu\_op_n$($\alpha$) $<$ $Ll\_op_{n+1}$($\alpha$))\\
	Hence, ($Fu\_op_1$($\alpha$) $<$ $Ll\_op_{n+1}$($\alpha$))
	
	\item ($Fu\_op_1$($\alpha$) $<$ $Ll\_op_n$($\alpha$)): Here, ($access.\gtcnt_n$ $<$	$Ll\_op_{n+1}$($\alpha$)). As we know if any method is locking as well as accessing common counter then locking tobject first then accessing the counter after that unlocking tobject i.e.\\
	So, ($Ll\_op_n$($\alpha$)) $<$ ($access.\gtcnt_n$) $<$ ($Fu\_op_n$($\alpha$)).\\
	Hence, ($Fu\_op_1$($\alpha$) $<$ $Ll\_op_{n+1}$($\alpha$))
	
	\item ($access.\gtcnt_{1}$) $<$ ($access.\gtcnt_{n}$): Here, ($access.\gtcnt_{n}$) $<$ $Ll\_op_{n+1}$($\alpha$)).\\
	So, ($access.\gtcnt_{1}$) $<$ ($access.\gtcnt_{n}$) $<$ $Ll\_op_{n+1}$($\alpha$)).\\
	Hence, ($access.\gtcnt_{1}$) $<$ $Ll\_op_{n+1}$($\alpha$)).
	
	\item ($Fu\_op_1$($\alpha$) $<$ ($access.\gtcnt_{n}$): Here, ($access.\gtcnt_{n}$) $<$ $Ll\_op_{n+1}$($\alpha$)).\\
	So, ($Fu\_op_1$($\alpha$) $<$ ($access.\gtcnt_{n}$) $<$ $Ll\_op_{n+1}$($\alpha$)).\\
	Hence, ($Fu\_op_1$($\alpha$) $<$ $Ll\_op_{n+1}$($\alpha$))
	
	\item ($access.\gtcnt_{1}$) $<$ $Ll\_op_{n}$($\alpha$)): Here, ($Fu\_op_n$($\alpha$) $<$ $Ll\_op_{n+1}$($\alpha$)).\\
	So, ($access.\gtcnt_{1}$) $<$ $Ll\_op_{n}$($\alpha$)) $<$ ($Fu\_op_n$($\alpha$) $<$ $Ll\_op_{n+1}$($\alpha$)).\\
	Hence, ($access.\gtcnt_{1}$) $<$ $Ll\_op_{n+1}$($\alpha$)).
	
	\item ($access.\gtcnt_{1}$) $<$ $Ll\_op_{n}$($\alpha$)): Here, ($access.\gtcnt_n$ $<$	$Ll\_op_{n+1}$($\alpha$)). As we know if any method is locking as well as accessing common counter then locking tobject first then accessing the counter after that unlocking tobject i.e.\\
	So, ($Ll\_op_n$($\alpha$)) $<$ ($access.\gtcnt_n$) $<$ ($Fu\_op_n$($\alpha$)).\\
	Hence, ($access.\gtcnt_{1}$) $<$ $Ll\_op_{n+1}$($\alpha$)).
\end{enumerate}
\item $op_{n+1}$ is a non-locking method: Again, we are considering all the possible path between $op_1$ to $op_n$:
\begin{enumerate}
	\item ($Fu\_op_1$($\alpha$) $<$ $Ll\_op_n$($\alpha$)): Here, ($access.\gtcnt_n$) $<$ ($access.\gtcnt_{n+1}$).\\
	As we know if any method is locking as well as accessing common counter then locking tobject first then accessing the counter after that unlocking tobject i.e.\\
	So, ($Ll\_op_n$($\alpha$)) $<$ ($access.\gtcnt_n$) $<$ ($Fu\_op_n$($\alpha$)).\\
	Hence, ($Fu\_op_1$($\alpha$) $<$ ($access.\gtcnt_{n+1}$)
	
	\item ($Fu\_op_1$($\alpha$) $<$ $Ll\_op_n$($\alpha$)): Here, ($Fu\_op_n$($\alpha$) $<$ ($access.\gtcnt_{n+1}$). \\	So, ($Fu\_op_1$($\alpha$) $<$ $Ll\_op_n$($\alpha$)) $<$ ($Fu\_op_n$($\alpha$) $<$ ($access.\gtcnt_{n+1}$)\\
	Hence, ($Fu\_op_1$($\alpha$) $<$ ($access.\gtcnt_{n+1}$))
	
	\item ($access.\gtcnt_{1}$) $<$ ($access.\gtcnt_{n}$): Here, ($access.\gtcnt_{n}$) $<$ ($access.\gtcnt_{n+1}$).\\
	So, ($access.\gtcnt_{1}$) $<$ ($access.\gtcnt_{n}$) $<$ ($access.\gtcnt_{n+1}$).\\
	Hence, ($access.\gtcnt_{1}$) $<$ ($access.\gtcnt_{n+1}$).
	
	\item ($Fu\_op_1$($\alpha$) $<$ ($access.\gtcnt_{n}$): Here, ($access.\gtcnt_{n}$) $<$ ($access.\gtcnt_{n+1}$).\\
	So, ($Fu\_op_1$($\alpha$) $<$ ($access.\gtcnt_{n}$) $<$ ($access.\gtcnt_{n+1}$).\\
	Hence, ($Fu\_op_1$($\alpha$) $<$ ($access.\gtcnt_{n+1}$)
	
	\item ($access.\gtcnt_{1}$) $<$ $Ll\_op_{n}$($\alpha$)): Here, ($access.\gtcnt_n$) $<$ ($access.\gtcnt_{n+1}$).\\
	As we know if any method is locking as well as accessing common counter then locking tobject first then accessing the counter after that unlocking tobject i.e.\\
	So, ($Ll\_op_n$($\alpha$)) $<$ ($access.\gtcnt_n$) $<$ ($Fu\_op_n$($\alpha$)).\\
	Hence, ($access.\gtcnt_{1}$) $<$ ($access.\gtcnt_{n+1}$).
	
	\item ($access.\gtcnt_{1}$) $<$ $Ll\_op_{n}$($\alpha$)): Here, ($Fu\_op_n$($\alpha$) $<$ ($access.\gtcnt_{n+1}$). \\	
	So, ($access.\gtcnt_{1}$) $<$ $Ll\_op_{n}$($\alpha$)) $<$ ($Fu\_op_n$($\alpha$) $<$ ($access.\gtcnt_{n+1}$). \\
	Hence, ($access.\gtcnt_{1}$) $<$ ($access.\gtcnt_{n+1}$).
	
\end{enumerate}
\end{enumerate}

\item Conflict edge between $op_n$ to $op_{n+1}$:
\begin{enumerate}
	\item ($Fu\_op_1$($\alpha$) $<$ $Ll\_op_n$($\alpha$)): Here, ($Fu\_op_n$($\alpha$) $<$ $Ll\_op_{n+1}$($\alpha$)). Ref 1.(a).i.
	
	\item ($access.\gtcnt_{1}$) $<$ ($access.\gtcnt_{n}$): Here, ($Fu\_op_n$($\alpha$) $<$ $Ll\_op_{n+1}$($\alpha$)). As we know if any method is locking as well as accessing common counter then locking tobject first then accessing the counter after that unlocking tobject i.e.\\
	So, ($Ll\_op_n$($\alpha$)) $<$ ($access.\gtcnt_n$) $<$ ($Fu\_op_n$($\alpha$)).\\
	Hence, ($access.\gtcnt_{1}$) $<$ $Ll\_op_{n+1}$($\alpha$)).
	
	\item ($Fu\_op_1$($\alpha$) $<$ ($access.\gtcnt_{n}$): Here, ($Fu\_op_n$($\alpha$) $<$ $Ll\_op_{n+1}$($\alpha$)). As we know if any method is locking as well as accessing common counter then locking tobject first then accessing the counter after that unlocking tobject i.e.\\
	So, ($Ll\_op_n$($\alpha$)) $<$ ($access.\gtcnt_n$) $<$ ($Fu\_op_n$($\alpha$)).\\
	Hence, ($Fu\_op_1$($\alpha$) $<$ $Ll\_op_{n+1}$($\alpha$)).
	
	\item ($access.\gtcnt_{1}$) $<$ $Ll\_op_{n}$($\alpha$)): Here, ($Fu\_op_n$($\alpha$) $<$ $Ll\_op_{n+1}$($\alpha$)).\\ Ref 1.(a).v.
\end{enumerate}

\item Common counter edge between $op_n$ to $op_{n+1}$:
\begin{enumerate}
	\item ($Fu\_op_1$($\alpha$) $<$ $Ll\_op_n$($\alpha$)): Here, ($access.\gtcnt_{n}$) $<$ ($access.\gtcnt_{n+1}$). As we know if any method is locking as well as accessing common counter then locking tobject first then accessing the counter after that unlocking tobject i.e.\\
	So, ($Ll\_op_n$($\alpha$)) $<$ ($access.\gtcnt_n$) $<$ ($Fu\_op_n$($\alpha$)).\\
	Hence, ($Fu\_op_1$($\alpha$) $<$ ($access.\gtcnt_{n+1}$).
	
	\item ($access.\gtcnt_{1}$) $<$ ($access.\gtcnt_{n}$): Here, ($access.\gtcnt_{n}$) $<$ ($access.\gtcnt_{n+1}$). Ref 1.(b).iii.
	
	\item ($Fu\_op_1$($\alpha$) $<$ ($access.\gtcnt_{n}$): Here, ($access.\gtcnt_{n}$) $<$ ($access.\gtcnt_{n+1}$). Ref 1.(b).iv.
	
	\item ($access.\gtcnt_{1}$) $<$ $Ll\_op_{n}$($\alpha$)): Here, ($access.\gtcnt_n$) $<$ ($access.\gtcnt_{n+1}$). Ref 1.(b).v
\end{enumerate}
\end{enumerate}

\cmnt{
then $Fpu_i$($\alpha$) $<$ 
$Lpl_j$($\alpha$). Conflict edge is based on $M\_Order_H$.\\
	 
We are proving by contradiction while assuming OPG(H) is cyclic. So, $op_1$ 
$\rightarrow$ $op_2$ $\rightarrow$ $op_3$ $\rightarrow$....$\rightarrow$ 
$op_k$ $\rightarrow$ $op_1$ that implies either ($Fpu_1$($\alpha$) $<$ 
$Lpl_2$($\alpha$) $<$ $Fpu_2$($\alpha$) $<$ 
$Lpl_3$($\alpha$) $<$ .... $<$ $Fpu_{k-1}$($\alpha$) $<$ 
$Lpl_k$($\alpha$) $<$ $Fpu_k$($\alpha$) $<$ 
$Lpl_1$($\alpha$)) or ($access.\gtcnt_1$ $<$ $access.\gtcnt_2$ $<$ 
$access.\gtcnt_3$ $<$ .... $<$ $access.\gtcnt_{k-1}$ $<$ $access.\gtcnt_k$ $<$ 
$access.\gtcnt_1$). Hence from above assumption, either ($Fpu_1$($\alpha$) $<$ 
$Lpl_1$($\alpha$)) or 
($access.\gtcnt_1$ $<$ $access.\gtcnt_1$). But ($Fpu_1$($\alpha$) 
$<$ $Lpl_1$($\alpha$)) is impossible because methods are follwing 2PL order of 
locking and ($access.\gtcnt_1$ $<$ $access.\gtcnt_1$) is never be true because 
of  same method $op_1$. Hence, all the above cases are impossible. So, our 
assumption is wrong.
}

Therefore, OPG(H, $M\_Order$) produced by KSFTM is acyclic.
\end{proof}

\begin{lemma}
\label{lem:val-hs}
Any history H gen(KSFTM) with $\alpha$ linearization such that it respects $M\_Order_H$ then (H, $\alpha$) is valid.
\end{lemma}

\begin{proof}
From the definition of \textit{valid history}: If all the read operations of H is reading from the previously committed transaction $T_j$ then H is valid.\\
In order to prove H is valid, we are analyzing the read(i,x). so, from \Lineref{rd-curver10}, it returns the largest \ts value less than $\gwts_i$ that has already been committed and return the value successfully. If such version created by transaction $T_j$ found then $T_i$ read from $T_j$. Otherwise, if there is no version whose \wts is less than $T_i$'s \wts, then $T_i$ returns abort.\\
Now, consider the base case read(i,x) is the first transaction $T_1$ and none of the transactions has been created a version then as we have assummed, there always exist $T_0$ by default that has been created a version for all t-objects. Hence, $T_1$ reads from committed transaction $T_0$.\\
So, all the reads are reading from largest \ts value less than $\gwts_i$ that 
has already been committed. Hence, (H, $\alpha$) is valid.
\end{proof}

\begin{lemma}
\label{lem:rt-hs}
Any history H gen(KSFTM) with $\alpha$ and $\beta$ linearization such that both 
respects $M\_Order_H$ i.e. $M\_Order_H \subseteq  \alpha$ and $M\_Order_H 
\subseteq  \beta$ then  $\prec_{(H, {\alpha})} ^{RT}$= $\prec_{(H, {\beta})}^{RT}$.
\end{lemma}

\begin{proof}
Consider a history H gen(KSFTM) such that two transactions $T_i$ and $T_j$ are in 
real time order which respects $M\_Order_H$ i.e. $\tryc_i$ $<$ 
$\begt_j$. As $\alpha$ and $\beta$ are linearizations of H so, $\tryc_i$ 
$<_{(H,{\alpha})}$ $\begt_j$ and $\tryc_i$ $<_{(H,{\beta})}$ $\begt_j$. Hence in both 
the cases of linearizations, $T_i$ committed before begin of $T_j$. So, 
$\prec_{(H, {\alpha})} ^{RT}$= $\prec_{(H, {\beta})}^{RT}$. 
\end{proof}

\begin{lemma}
	Any history H gen(KSFTM) with $\alpha$ and $\beta$ linearization such that both respects $M\_Order_H$ i.e. $M\_Order_H \subseteq  \alpha$ and $M\_Order_H	\subseteq  \beta$ then $(H, {\alpha})$ is local opaque iff $(H, {\beta})$ is local opaque.
\end{lemma}

\begin{proof}
	As $\alpha$ and $\beta$ are linearizations of history H gen(KSFTM) so, from \lemref{val-hs} (H, $\alpha$) and (H, $\beta$) are valid histories.
	
	Now assuming (H, $\alpha$) is local opaque so we need to show (H, $\beta$) is also local opaque. Since (H, $\alpha$) is local opaque so there exists legal t-sequential history S (with respect to each aborted transactions and last committed transaction while considering only committed transactions) which is equivalent to ($\overline{H}$, $\alpha$). As we know $\beta$ is a linearization of H so ($\overline{H}$, $\beta$) is equivalent to some legal t-sequential history S. From the definition of local opacity $\prec_{(H, {\alpha})} ^{RT} \subseteq \prec_S^{RT}$. From \lemref{rt-hs}, $\prec_{(H, {\alpha})} ^{RT}$= $\prec_{(H, {\beta})}^{RT}$ that implies $\prec_{(H, {\beta})}^{RT} \subseteq \prec_S^{RT}$. Hence, $(H, {\beta})$ is local opaque.\\
	
	Now consider the other way in which (H, $\beta$) is local opaque and we need to show (H, $\alpha$) is also local opaque. We can prove it while giving the same argument as above, by exchanging  $\alpha$ and $\beta$.\\
	
	Hence, $(H, {\alpha})$ is local opaque iff $(H, {\beta})$ is local opaque.
\end{proof}

\begin{theorem}
	\label{thm:ap-ksftm-lo} 
	Any history generated by \ksftm{} is \lopq.
\end{theorem}

\begin{proof}
	For proving this, we consider a sequential history $H$ generated by \ksftm. We define the version order $\ll_{\vt}$: for two versions $v_i, v_j$ it is defined as
	
	$(v_i \ll_{\vt} v_j) \equiv (v_i.\vt < v_j.\vt)$
	
	\noindent Using this version order $\ll_{\vt}$, we can show that all the sub-histories in $\shset{H}$ are acyclic.	
\end{proof}

\noindent Since the histories generated by \ksftm are \lopq, we get that they are also \stsble. 

\begin{corollary}
	\label{thm:ap-ksftm-stsble} 
	Any history generated by \ksftm{} is \stsble.
\end{corollary}

\section{Proof of Liveness}
\label{sec:ap-liveness}

\paragraph{Proof Notations:} Let $\gen{\ksftm}$ consist of all the histories accepted by \ksftm{} algorithm. In the follow sub-section, we only consider histories that are generated by \ksftm unless explicitly stated otherwise. For simplicity, we only consider sequential histories in our discussion below. 

Consider a transaction $T_i$ in a history $H$ generated by \ksftm. Once it executes \begt \mth, its \its, \cts, \wts values do not change. Thus, we denote them as $\tits{i}, \tcts{i}, \twts{i}$ respectively for $T_i$. In case the context of the history $H$ in which the transaction executing is important, we denote these variables as $\htits{i}{H}, \htcts{i}{H}, \htwts{i}{H}$ respectively. 

The other variables that a transaction maintains are: \ltl, \utl, \lock, \val, \stat. These values change as the execution proceeds. Hence, we denote them as: $\htltl{i}{H}, \htutl{i}{H}, \htlock{i}{H}, \htval{i}{H}, \htstat{i}{H}$. These represent the values of \ltl, \utl, \lock, \val, \stat after the execution of last event in $H$. Depending on the context, we sometimes ignore $H$ and denote them only as: $\tlock{i}, \tval{i}, \tstat{i}, \ttltl{i}, \ttutl{i}$.

We approximate the system time with the value of $\tcntr$. We denote the \syst of history $H$ as the value of $\tcntr$ immediately after the last event of $H$. Further, we also assume that the value of $C$ is 1 in our arguments. But, it can be seen that the proof will work for any value greater than 1 as well. 

The application invokes transactions in such a way that if the current $T_i$ transaction aborts, it invokes a new transaction $T_j$ with the same \its. We say that $T_i$ is an \emph{\inc} of $T_j$ in a history $H$ if $\htits{i}{H} = \htits{j}{H}$. Thus the multiple \inc{s} of a transaction $T_i$ get invoked by the application until an \inc finally commits. 

To capture this notion of multiple transactions with the same \its, we define \emph{\incset} (incarnation set) of $T_i$ in $H$ as the set of all the transactions in $H$ which have the same \its as $T_i$ and includes $T_i$ as well. Formally, 
\begin{equation*}
\incs{i}{H} = \{T_j|(T_i = T_j) \lor (\htits{i}{H} = \htits{j}{H})\}  
\end{equation*}

Note that from this definition of  \incset, we implicitly get that $T_i$ and all the transactions in its \incset of $H$ also belong to $H$. Formally, $\incs{i}{H} \in \txns{H}$. 

The application invokes different incarnations of a transaction $T_i$ in such a way that as long as an \inc is live, it does not invoke the next \inc. It invokes the next \inc after the current \inc has got aborted. Once an \inc of $T_i$ has committed, it can't have any future \inc{s}. Thus, the application views all the \inc{s} of a transaction as a single \emph{\aptr}. 

We assign \emph{\incn{s}} to all the transactions that have the same \its. We say that a transaction $T_i$ starts \emph{afresh}, if $\inum{i}$ is 1. We say that $T_i$ is the \ninc of $T_i$ if $T_j$ and $T_i$ have the same \its and $T_i$'s \incn is $T_j$'s \incn + 1. Formally, $\langle (\nexti{i} = T_j) \equiv (\tits{i} = \tits{j}) \land (\inum{i} = \inum{j} + 1)\rangle$


As mentioned the objective of the application is to ensure that every \aptr eventually commits. Thus, the applications views the entire \incset as a single \aptr (with all the transactions in the \incset having the same \its). We can say that an \aptr has committed if in the corresponding \incset a transaction in eventually commits. For $T_i$ in a history $H$, we denote this by a boolean value \incct (incarnation set committed) which implies that either $T_i$ or an \inc of $T_i$ has committed. Formally, we define it as $\inct{i}{H}$

\begin{equation*}
\inct{i}{H} = \begin{cases}
True    & (\exists T_j: (T_j \in \incs{i}{H}) \land (T_j \in \comm{H})) \\
False	& \text{otherwise}
\end{cases}
\end{equation*}

\noindent From the definition of \incct we get the following observations \& lemmas about a transaction $T_i$

\begin{observation}
\label{obs:inct-term}
Consider a transaction $T_i$ in a history $H$ with its \incct being true in $H$. Then $T_i$ is terminated (either committed or aborted) in $H$. Formally, $\langle H, T_i: (T_i \in \txns{H}) \land (\inct{i}{H}) \implies (T_i \in \termed{H}) \rangle$. 
\end{observation}

\begin{observation}
\label{obs:inct-fut}
Consider a transaction $T_i$ in a history $H$ with its \incct being true in $H1$. Let $H2$ be a extension of $H1$ with a transaction $T_j$ in it. Suppose $T_j$ is an \inc of $T_i$. Then $T_j$'s \incct is true in $H2$. Formally, $\langle H1, H2, T_i, T_j: (H1 \sqsubseteq H2) \land (\inct{i}{H1}) \land (T_j \in \txns{H2}) \land (T_i \in \incs{j}{H2})\implies (\inct{j}{H2}) \rangle$. 
\end{observation}

\begin{lemma}
\label{lem:inct-diff}
Consider a history $H1$ with a strict extension $H2$. Let $T_i$ \& $T_j$ be two transactions in $H1$ \& $H2$ respectively. Let $T_j$ not be in $H1$. Suppose $T_i$'s \incct is true. Then \its of $T_i$ cannot be the same as \its of $T_j$. Formally, $\langle H1, H2, T_i, T_j: (H1 \sqsubset H2) \land (\inct{i}{H1}) \land (T_j \in \txns{H2}) \land (T_j \notin \txns{H1}) \implies (\htits{i}{H1} \neq \htits{j}{H2}) \rangle$.
\end{lemma}

\begin{proof}
Here, we have that $T_i$'s \incct is true in $H1$. Suppose $T_j$ is an \inc of $T_i$, i.e., their \its{s} are the same. We are given that $T_j$ is not in $H1$. This implies that $T_j$ must have started after the last event of $H1$. 

We are also given that $T_i$'s \incct is true in $H1$. This implies that an \inc of $T_i$ or $T_i$ itself has committed in $H1$. After this commit, the application will not invoke another transaction with the same \its as $T_i$. Thus, there cannot be a transaction after the last event of $H1$ and in any extension of $H1$ with the same \its of $T_1$. Hence, $\htits{i}{H1}$ cannot be same as $\htits{j}{H2}$. 
\end{proof}

Now we show the liveness with the following observations, lemmas \& theorems. We start with two observations about that histories of which one is an extension of the other. The following states that for any history, there exists an extension. In other words, we assume that the STM system runs forever and does not terminate. This is required for showing that every transaction eventually commits. 

\begin{observation}
\label{obs:hist-future}
Consider a history $H1$ generated by \gen{\ksftm}. Then there is a history $H2$ in \gen{\ksftm} such that $H2$ is a strict extension of $H1$. Formally, $\langle \forall H1: (H1 \in \gen{ksftm}) \implies (\exists H2: (H2 \in \gen{ksftm}) \land (H1 \sqsubset H2) \rangle$. 
\end{observation}

\noindent The follow observation is about the transaction in a history and any of its extensions. 

\begin{observation}
\label{obs:hist-subset}
Given two histories $H1$ \& $H2$ such that $H2$ is an extension of $H1$. Then, the set of transactions in $H1$ are a subset equal to  the set of transaction in $H2$. Formally, $\langle \forall H1, H2: (H1 \sqsubseteq H2) \implies (\txns{H1} \subseteq \txns{H2}) \rangle$. 
\end{observation}

In order for a transaction $T_i$ to commit in a history $H$, it has to compete with all the live transactions and all the aborted that can become live again as a different \inc. Once a transaction $T_j$ aborts, another \inc of $T_j$ can start and become live again. Thus $T_i$ will have to compete with this \inc of $T_j$ later. Thus, we have the following observation about aborted \& committed transactions. 

\begin{observation}
\label{obs:abort-retry}
Consider an aborted transaction $T_i$ in a history $H1$. Then there is an extension of $H1$, $H2$ in which an \inc of $T_i$, $T_j$ is live and has $\tcts{j}$ is greater than $\tcts{i}$. Formally, $\langle H1, T_i: (T_i \in \aborted{H1}) \implies(\exists T_j, H2: (H1 \sqsubseteq H2) \land (T_j \in \live{H2}) \land (\htits{i}{H2} = \htits{j}{H2}) \land (\htcts{i}{H2} < \htcts{j}{H2})) \rangle$. 
\end{observation}

\begin{observation}
\label{obs:cmt-noinc}
Consider an committed transaction $T_i$ in a history $H1$. Then there is no extension of $H1$, in which an \inc of $T_i$, $T_j$ is live. Formally, $\langle H1, T_i: (T_i \in \comm{H1}) \implies(\nexists T_j, H2: (H1 \sqsubseteq H2) \land (T_j \in \live{H2}) \land (\htits{i}{H2} = \htits{j}{H2})) \rangle$. 
\end{observation}

\begin{lemma}
\label{lem:cts-wts}
Consider a history $H1$ and its extension $H2$. Let $T_i, T_j$ be in $H1, H2$ respectively such that they are \inc{s} of each other. If \wts of $T_i$ is less than \wts of $T_j$ then \cts of $T_i$ is less than \cts $T_j$.
Formally, $\langle H1, H2, T_i, T_j: (H1 \sqsubset H2) \land (T_i \in \txns{H1}) \land (T_j \in \txns{H2}) \land (T_i \in \incs{j}{H2}) \land (\htwts{i}{H1} < \htwts{j}{H2})\implies (\htcts{i}{H1} < \htcts{j}{H2}) \rangle$
\end{lemma}

\begin{proof}
Here we are given that 
\begin{equation}
\label{eq:wts-ij}
\htwts{i}{H1} < \htwts{j}{H2}
\end{equation}

The definition of \wts of $T_i$ is: $\htwts{i}{H1} = \htcts{i}{H1} + C * (\htcts{i}{H1} - \htits{i}{H1})$. Combining this \eqnref{wts-ij}, we get that 


$(C + 1) * \htcts{i}{H1} - C * \htits{i}{H1} < (C + 1) * \htcts{j}{H2} - C * \htits{j}{H2} \xrightarrow[\htits{i}{H1} = \htits{j}{H2}]{T_i \in \incs{j}{H2}} \htcts{i}{H1} < \htcts{j}{H2}$. 
\end{proof}

\begin{lemma}
\label{lem:wts-great}
Consider a live transaction $T_i$ in a history $H1$ with its $\twts{i}$ less than a constant $\alpha$. Then there is a strict extension of $H1$, $H2$ in which an \inc of $T_i$, $T_j$ is live with \wts greater than $\alpha$. Formally, $\langle H1, T_i: (T_i \in \live{H1}) \land (\htwts{i}{H1} < \alpha) \implies(\exists T_j, H2: (H1 \sqsubseteq H2) \land (T_i \in \incs{j}{H2}) \land ((T_j \in \comm{H2}) \lor ((T_j \in \live{H2}) \land (\htwts{j}{H2} > \alpha)))) \rangle$. 
\end{lemma}

\begin{proof}
The proof comes the behavior of an \aptr. The application keeps invoking a transaction with the same \its until it commits. Thus the transaction $T_i$ which is live in $H1$ will eventually terminate with an abort or commit. If it commits, $H2$ could be any history after the commit of $T_2$. 

On the other hand if $T_i$ is aborted, as seen in \obsref{abort-retry} it will be invoked again or reincarnated with another \cts and \wts. It can be seen that \cts is always increasing. As a result, the \wts is also increasing. Thus eventually the \wts will become greater $\alpha$. Hence, we have that either an \inc of $T_i$ will get committed or will eventually have \wts greater than or equal to $\alpha$.
\end{proof}

\noindent Next we have a lemma about \cts of a transaction and the \syst of a history. 

\begin{lemma}
\label{lem:cts-syst}
Consider a transaction $T_i$ in a history $H$. Then, we have that \cts of $T_i$ will be less than or equal to \syst of $H$. Formally, $\langle T_i, H1: (T_i \in \txns{H}) \implies (\htcts{i}{H} \leq \hsyst{H}) \rangle$. 
\end{lemma}

\begin{proof}
We get this lemma by observing the \mth{s} of the STM System that increment the \tcntr which are \begt and \tryc. It can be seen that \cts of $T_i$ gets assigned in the \begt \mth. So if the last \mth of $H$ is the \begt of $T_i$ then we get that \cts of $T_i$ is same as \syst of $H$. On the other hand if some other \mth got executed in $H$ after \begt of $T_i$ then we have that \cts of $T_i$ is less than \syst of $H$. Thus combining both the cases, we get that \cts of $T_i$ is less than or equal to as \syst of $H$, i.e., $(\htcts{i}{H} \leq \hsyst{H})$
\end{proof}

\noindent From this lemma, we get the following corollary which is the converse of the lemma statement

\begin{corollary}
\label{cor:cts-syst}
Consider a transaction $T_i$ which is not in a history $H1$ but in an strict extension of $H1$, $H2$. Then, we have that \cts of $T_i$ is greater than the \syst of $H$. Formally, $\langle T_i, H1, H2: (H1 \sqsubset H2) \land (T_i \notin \txns{H1}) \land (T_i \in \txns{H2}) \implies (\htcts{i}{H2} > \hsyst{H1}) \rangle$. 
\end{corollary}

\noindent Now, we have lemma about the \mth{s} of \ksftm completing in finite time. 

\begin{lemma}
\label{lem:mth-fdm}
If all the locks are fair and the underlying system scheduler is fair then all the \mth{s} of \ksftm will eventually complete.
\end{lemma}

\begin{proof}
It can be seen that in any \mth, whenever a transaction $T_i$ obtains multiple locks, it obtains locks in the same order: first lock relevant \tobj{s} in a pre-defined order and then lock relevant \glock{s} again in a predefined order. Since all the locks are obtained in the same order, it can be seen that the \mth{s} of \ksftm will not deadlock. 

It can also be seen that none of the \mth{s} have any unbounded while loops. All the loops in \tryc \mth iterate through all the \tobj{s} in the write-set of $T_i$. Moreover, since we assume that the underlying scheduler is fair, we can see that no thread gets swapped out infinitely. Finally, since we assume that all the locks are fair, it can be seen all the \mth{s} terminate in finite time.
\end{proof}

\begin{theorem}
\label{thm:trans-com|abt}
Every transaction either commits or aborts in finite time.
\end{theorem}

\begin{proof}
This theorem comes directly from the \lemref{mth-fdm}. Since every \mth of \ksftm will eventually complete, all the transactions will either commit or abort in finite time.
\end{proof}

\noindent From this theorem, we get the following corollary which states that the maximum \emph{lifetime} of any transaction is $L$. 

\begin{corollary}
\label{cor:cts-L}
Any transaction $T_i$ in a history $H$ will either commit or abort before the \syst of $H$ crosses $\tcts{i} + L$. 
\end{corollary}

\noindent The following lemma connects \wts and \its of two transactions, $T_i, T_j$. 

\begin{lemma}
\label{lem:wts-its}
Consider a history $H1$ with two transactions $T_i, T_j$. Let $T_i$ be in $\live{H1}$. Suppose $T_j$'s \wts is greater or equal to $T_i$' s \wts. Then \its of $T_j$ is less than $\tits{i} + 2*L$. Formally, $\langle H, T_i, T_j : (\{ T_i, T_j\} \subseteq \txns{H}) \land ( T_i \in \live{H}) \land (\htwts{j}{H} \geq \htwts{i}{H}) \Longrightarrow (\htits{i}{H} + 2L \geq \htits{j}{H}) \rangle$.
\end{lemma}

\begin{proof}
Since $T_i$ is live in $H1$, from \corref{cts-L}, we get that it terminates before the system time, $\tcntr$ becomes $\tcts{i} + L$. Thus, \syst of history $H1$ did not progress beyond $\tcts{i} + L$. Hence, for any other transaction $T_j$ (which is either live or terminated) in $H1$, it must have started before \syst has crossed $\tcts{i} + L$. Formally $\langle \tcts{j} \leq \tcts{i} + L \rangle$.

Note that we have defined \wts of a transaction $T_j$ as: $\twts{j} = (\tcts{j} + C * (\tcts{j} - \tits{j}))$. Now, let us consider the difference of the \wts{s} of both the transactions. 

\noindent 
\begin{math}
\twts{j} - \twts{i}  = (\tcts{j} + C * (\tcts{j} - \tits{j})) - (\tcts{i} + C * (\tcts{i} - \tits{i})) \\ 
= (C + 1)(\tcts{j} - \tcts{i}) - C(\tits{j} - \tits{i}) \\
\leq (C + 1)L - C(\tits{j} - \tits{i}) \qquad [\because \tcts{j} \leq \tcts{i} + L] \\
= 2*L + \tits{i} - \tits{j}  \qquad [\because C = 1] \\
\end{math}

\noindent Thus, we have that: $ \langle (\tits{i} + 2L - \tits{j}) \geq (\twts{j} - \twts{i}) \rangle$. This gives us that \\
$((\twts{j} - \twts{i}) \geq 0) \Longrightarrow ((\tits{i} + 2L - \tits{j}) \geq 0)$. 

\noindent From the above implication we get that, 
$(\twts{j} \geq \twts{i}) \Longrightarrow (\tits{i} + 2L \geq \tits{j})$.

\end{proof}

It can be seen that \ksftm algorithm gives preference to transactions with lower \its to commit. To understand this notion of preference, we define a few notions of enablement of a transaction $T_i$ in a history $H$. We start with the definition of \emph{\itsen} as:

\begin{definition}
\label{defn:itsen}
We say $T_i$ is \emph{\itsen} in $H$ if for all transactions $T_j$ with \its lower than \its of $T_i$ in $H$ have \incct to be true. Formally, 
\begin{equation*}
\itsenb{i}{H} = \begin{cases}
True    & (T_i \in \live{H}) \land (\forall T_j \in \txns{H} : (\htits{j}{H} < \htits{i}{H}) \implies (\inct{j}{H})) \\
False	& \text{otherwise}
\end{cases}
\end{equation*}
\end{definition}

\noindent The follow lemma states that once a transaction $T_i$ becomes \itsen it continues to remain so until it terminates. 

\begin{lemma}
\label{lem:itsen-future}
Consider two histories $H1$ and $H2$ with $H2$ being a extension of $H1$. Let a transaction $T_i$ being live in both of them. Suppose $T_i$ is \itsen in $H1$. Then $T_i$ is \itsen in $H2$ as well. Formally, $\langle H1, H2, T_i: (H1 \sqsubseteq H2) \land (T_i \in \live{H1}) \land (T_i \in \live{H2}) \land (\itsenb{i}{H1}) \implies (\itsenb{i}{H2}) \rangle$. 
\end{lemma}

\begin{proof}
When $T_i$ begins in a history $H3$ let the set of transactions with \its less than $\tits{i}$ be $smIts$. Then in any extension of $H3$, $H4$ the set of transactions with \its less than $\tits{i}$ remains as $smIts$. 

Suppose $H1, H2$ are extensions of $H3$. Thus in $H1, H2$ the set of transactions with \its less than $\tits{i}$ will be $smIts$. Hence, if $T_i$ is \itsen in $H1$ then all the transactions $T_j$ in $smIts$ are $\inct{j}{H1}$. It can be seen that this continues to remain true in $H2$. Hence in $H2$, $T_i$ is also \itsen which proves the lemma.
\end{proof}

The following lemma deals with a committed transaction $T_i$ and any transaction $T_j$ that terminates later. In the following lemma, $\incv$ is any constant greater than or equal to 1. 

\begin{lemma}
\label{lem:tryci-j}
Consider a history $H$ with two transactions $T_i, T_j$ in it. Suppose transaction $T_i$ commits before $T_j$ terminates (either by commit or abort) in $H$. Then $\ct_i$ is less than $\ct_j$ by at least $\incv$. Formally, $\langle H, \{T_i, T_j\} \in \txns{H}: (\tryc_i <_H \term_j) \implies (\ct_i + \incv \leq \ct_j)\rangle$. 
\end{lemma}

\begin{proof}
When $T_i$ commits, let the value of the global $\tcntr$ be $\alpha$. It can be seen that in \begt \mth, $\ct_j$ get initialized to $\infty$. The only place where $\ct_j$ gets modified is at \Lineref{tryc-cmt-mod} of \tryc. Thus if $T_j$ gets aborted before executing \tryc \mth or before this line of \tryc we have that $\ct_j$ remains at $\infty$. Hence in this case we have that $\langle \ct_i + \incv < \ct_j \rangle$.

If $T_j$ terminates after executing \Lineref{tryc-cmt-mod} of \tryc \mth then $\ct_j$ is assigned a value, say $\beta$. It can be seen that $\beta$ will be greater than $\alpha$ by at least $\incv$ due to the execution of this line. Thus, we have that $\langle \alpha + \incv \leq \beta \rangle$
\end{proof}

\noindent The following lemma connects the \tltl and \ct of a transaction $T_i$. 

\begin{lemma}
\label{lem:ti|tltl-comt}
Consider a history $H$ with a transaction $T_i$ in it. Then in $H$, $\ttltl{i}$ will be less than or equal to $\ct_i$. Formally, $\langle H, \{T_i\} \in \txns{H}: (\htltl{i}{H} \leq H.\ct_i) \rangle$.
\end{lemma}

\begin{proof}
Consider the transaction $T_i$. In \begt \mth, $\ct_i$ get initialized to $\infty$. The only place where $\ct_i$ gets modified is at \Lineref{tryc-cmt-mod} of \tryc. Thus if $T_i$ gets aborted before this line or if $T_i$ is live we have that $(\ttltl{i} \leq \ct_i)$. On executing \Lineref{tryc-cmt-mod}, $\ct_i$ gets assigned to some finite value and it does not change after that. 

It can be seen that $\ttltl{i}$ gets initialized to $\tcts{i}$ in \Lineref{ti-ts-init} of \begt \mth. In that line, $\tcts{i}$ reads $\tcntr$ and increments it atomically. Then in \Lineref{tryc-cmt-mod}, $\ct_i$ gets assigned the value of $\tcntr$ after incrementing it. Thus, we clearly get that $\tcts{i} (= \ttltl{i}\text{ initially}) < \ct_i$. Then $\ttltl{i}$ gets updated on \Lineref{rd-tltl-inc} of read, \Lineref{tryc-tltl-inc} and \Lineref{ti-updt} of \tryc \mth{s}. Let us analyze them case by case assuming that $\ttltl{i}$ was last updated in each of these \mth{s} before the termination of $T_i$:

\begin{enumerate}
\item \label{case:read} \Lineref{rd-tltl-inc} of read \mth: Suppose this is the last line where $\ttltl{i}$ updated. Here $\ttltl{i}$ gets assigned to 1 + \vt of the previously committed version which say was created by a transaction $T_j$. Thus, we have the following equation, 
\begin{equation}
\label{eq:tltl-vt}
\ttltl{i} = 1 + x[j].\vt
\end{equation}

It can be seen that $x[j].\vt$ is same as $\ttltl{j}$ when $T_j$ executed \Lineref{new-tup} of \tryc. Further, $\ttltl{j}$ in turn is same as $\ttutl{j}$ due to \Lineref{ti-updt} of \tryc. From \Lineref{tryc-ul-cmt}, it can be seen that $\ttutl{j}$ is less than or equal to $\ct_j$ when $T_j$ committed. Thus we have that 
\begin{equation}
\label{eq:tltl-ct}
x[j].\vt = \ttltl{j} = \ttutl{j} \leq \ct_j
\end{equation}

It is clear that from the above discussion that  $T_j$ executed \tryc \mth before $T_i$ terminated (i.e. $\tryc_j <_{H1} \term_i$). From \eqnref{tltl-vt} and \eqnref{tltl-ct}, we get \\
\begin{math}
\ttltl{i} \leq 1 + \ct_j \xrightarrow[]{\incv \geq 1} \ttltl{i} \leq \incv + \ct_j \xrightarrow[]{\lemref{tryci-j}} \ttltl{i} \leq \ct_i
\end{math}
	
\item \label{case:tryc-short} \Lineref{tryc-tltl-inc} of \tryc \mth: The reasoning in this case is very similar to the above case. 
	
\item \label{case:tryc-long} \Lineref{ti-updt} of \tryc \mth: In this line, $\ttltl{i}$ is made equal to $\ttutl{i}$. Further, in \Lineref{tryc-ul-cmt}, $\ttutl{i}$ is made lesser than or equal to $\ct_{i}$. Thus combing these, we get that $\ttltl{i} \leq \ct_{i}$. It can be seen that the reasoning here is similar in part to \csref{read}.  
\end{enumerate}

Hence, in all the three cases we get that $\langle \ttltl{i} \leq \ct_i \rangle$. 
\end{proof}

\noindent The following lemma connects the \tutl,\ct of a transaction $T_i$ with \wts of a transaction $T_j$ that has already committed. 

\begin{lemma}
\label{lem:ti|tutl-comt}
Consider a history $H$ with a transaction $T_i$ in it. Suppose $\ttutl{i}$ is less than $\ct_i$. Then, there is a committed transaction $T_j$ in $H$ such that $\twts{j}$ is greater than $\twts{i}$. Formally, $\langle H \in \gen{\ksftm}, \{T_i\} \in \txns{H}: (\htutl{i}{H} < H.\ct_i) \implies (\exists T_j \in \comm{H}: \htwts{j}{H} > \htwts{i}{H}) \rangle$.
\end{lemma}

\begin{proof}
It can be seen that $\tutl_i$ initialized in \begt \mth to $\infty$. $\ttutl{i}$ is updated in \Lineref{rd-ul-dec} of read \mth, \Lineref{tryc-ul-dec} \& \Lineref{tryc-ul-cmt} of \tryc \mth. If $T_i$ executes \Lineref{rd-ul-dec} of read \mth and/or \Lineref{tryc-ul-dec} of \tryc \mth then $\ttutl{i}$ gets decremented to some value less than $\infty$, say $\alpha$. Further, it can be seen that in both these lines the value of $\ttutl{i}$ is possibly decremented from $\infty$ because of $nextVer$ (or $ver$), a version of $x$ whose \ts is greater than $T_i$'s \wts. This implies that some transaction $T_j$, which is committed in $H$, must have created $nextVer$ (or $ver$) and $\twts{j} > \twts{i}$. 

Next, let us analyze the value of $\alpha$. It can be seen that $\alpha = x[nextVer/ver].vrt - 1$ where $nextVer/ver$ was created by $T_j$. Further, we can see when $T_j$ executed \tryc, we have that $x[nextVer].vrt = \ttltl{j}$ (from \Lineref{new-tup}). From \lemref{ti|tltl-comt}, we get that $\ttltl{j} \leq \ct_j$. This implies that $\alpha < \ct_j$. Now, we have that $T_j$ has already committed before the termination of $T_i$. Thus from \lemref{tryci-j}, we get that $\ct_j < \ct_i$. Hence, we have that, 

\begin{equation}
\label{eq:alph-ct}
\alpha < \ct_i 
\end{equation}

Now let us consider \Lineref{tryc-ul-cmt} executed by $T_i$ which causes $\ttutl{i}$ to change. This line will get executed only after both \Lineref{rd-ul-dec} of read \mth, \Lineref{tryc-ul-dec} of \tryc \mth. This is because every transaction executes \tryc \mth only after read \mth. Further within \tryc \mth, \Lineref{tryc-ul-cmt} follows \Lineref{tryc-ul-dec}. 

There are two sub-cases depending on the value of $\ttutl{i}$ before the execution of \Lineref{tryc-ul-cmt}: (i) If $\ttutl{i}$ was $\infty$  and then get decremented to $\ct_i$ upon executing this line, then we get $\ct_i =  \ttutl{i}$. From \eqnref{alph-ct},  we can ignore this case. (ii) Suppose the value of $\ttutl{i}$ before executing \Lineref{tryc-ul-cmt} was $\alpha$. Then from \eqnref{alph-ct} we get that $\ttutl{i}$ remains at $\alpha$ on execution of \Lineref{tryc-ul-cmt}. This implies that a transaction $T_j$ committed such that $\twts{j} > \twts{i}$.
\end{proof}

\noindent The following lemma connects the \tltl of a committed transaction $T_j$ and \ct of a transaction $T_i$ that commits later. 

\begin{lemma}
\label{lem:tltlj-comti}
Consider a history $H1$ with transactions $T_i, T_j$ in it. Suppose $T_j$ is committed and $T_i$ is live in $H1$. Then in any extension of $H1$, say $H2$, $\ttltl{j}$ is less than or equal to $\ct_i$. Formally, $\langle {H1, H2} \in \gen{\ksftm}, \{T_i, T_j\} \subseteq \txns{H1, H2}: (H1 \sqsubseteq H2) \land (T_j \in \comm{H1}) \land (T_i \in \live{H1}) \implies (\htltl{j}{H2} < H2.\ct_i) \rangle$.
\end{lemma}

\begin{proof}
As observed in the previous proof of \lemref{ti|tltl-comt}, if $T_i$ is live  or aborted in $H2$, then its \ct is $\infty$. In both these cases, the result follows.

If $T_i$ is committed in $H2$ then, one can see that \ct of $T_i$ is not $\infty$. In this case, it can be seen that $T_j$ committed before $T_i$. Hence, we have that $\ct_j < \ct_i$. From \lemref{ti|tltl-comt}, we get that $\ttltl{j} \leq \ct_j$. This implies that $\ttltl{j} < \ct_i$. 
\end{proof}

\noindent In the following sequence of lemmas, we identify the condition by when a transaction will commit. 

\begin{lemma}
\label{lem:its-wts}
Consider two histories $H1, H3$ such that $H3$ is a strict extension of $H1$. Let $T_i$ be a transaction in  $\live{H1}$ such that $T_i$ \itsen in $H1$ and $\gval_i$ flag is true in $H1$. Suppose $T_i$ is aborted in $H3$. Then there is a history $H2$ which is an extension of $H1$ (and could be same as $H1$) such that (1) Transaction $T_i$ is live in $H2$; (2) there is a transaction $T_j$ that is live in ${H2}$; (3) $\htwts{j}{H2}$ is greater than $\htwts{i}{H2}$; (4) $T_j$ is committed in $H3$. Formally, $ \langle H1, H3, T_i: (H1 \sqsubset H3) \land (T_i \in \live{H1}) \land (\htval{i}{H1} = True) \land (\itsenb{i}{H1}) \land (T_i \in \aborted{H3})) \implies (\exists H2, T_j: (H1 \sqsubseteq H2 \sqsubset H3) \land (T_i \in \live{H2}) \land (T_j \in \txns{H2}) \land (\htwts{i}{H2} < \htwts{j}{H2}) \land (T_j \in \comm{H3})) \rangle$. 
\end{lemma}

\begin{proof}
To show this lemma, w.l.o.g we assume that $T_i$ on executing either read or \tryc in $H2$ (which could be same as $H1$) gets aborted resulting in $H3$. Thus, we have that $T_i$ is live in $H2$. Here $T_i$ is \itsen in $H1$. From \lemref{itsen-future}, we get that $T_i$ is \itsen in $H2$ as well.


 Let us sequentially consider all the lines where a $T_i$ could abort. In $H2$, $T_i$ executes one of the following lines and is aborted in $H3$. We start with \tryc method. 

\begin{enumerate}

\item STM \tryc: 

\begin{enumerate}
\item \Lineref{init-tc-chk} \label{case:init-tc-chk}: This line invokes abort() method on $T_i$ which releases all the locks and returns $\mathcal{A}$ to the invoking thread. Here $T_i$ is aborted because its \val flag, is set to false by some other transaction, say $T_j$, in its \tryc algorithm. This can occur in Lines: \ref{lin:addAbl-lar}, \ref{lin:addAbl-sml} where $T_i$ is added to $T_j$'s \abl set. Later in \Lineref{gval-set}, $T_i$'s \val flag is set to false. Note that $T_i$'s \val is true (after the execution of the last event) in $H1$. Thus, $T_i$'s \val flag must have been set to false in an extension of $H1$, which we again denote as $H2$.

This can happen only if in both the above cases, $T_j$ is live in $H2$ and its \its is less than $T_i$'s \its. But we have that $T_i$'s \itsen in $H2$. As a result, it has the smallest among all live and aborted transactions of $H2$. Hence, there cannot exist such a $T_j$ which is live and $\htits{j}{H2} < \htits{i}{H2}$. Thus, this case is not possible. 

\item \Lineref{prev-nil}: This line is executed in $H2$ if there exists no version of $x$ whose \ts is less than $T_i$'s \wts. This implies that all the versions of $x$ have \ts{s} greater than $\twts{i}$. Thus the transactions that created these versions have \wts greater than $\twts{i}$ and have already committed in $H2$. Let $T_j$ create one such version. Hence, we have that $\langle (T_j \in \comm{H2}) \implies (T_j \in \comm{H3}) \rangle$ since $H3$ is an extension of $H2$. 

\item \Lineref{mid-tc-chk} \label{case:mid-tc-chk}: This case is similar to \csref{init-tc-chk}, i.e., \Lineref{init-tc-chk}. 

\item \Lineref{its-chk1} \label{case:its-chk1}: In this line, $T_i$ is aborted as some other transaction $T_j$ in $T_i$'s \lrl has committed. Any transaction in $T_i$'s \lrl has \wts greater than $T_i$'s \wts. This implies that $T_j$ is already committed in $H2$ and hence committed in $H3$ as well. 

\item \Lineref{tc-lts-cross} \label{case:tc-lts-cross}: In this line, $T_i$ is aborted because its lower limit has crossed its upper limit. First, let us consider $\ttutl{i}$. It is initialized in \begt \mth to $\infty$. As long as it is $\infty$, these limits cannot cross each other. Later, $\ttutl{i}$ is updated in \Lineref{rd-ul-dec} of read \mth, \Lineref{tryc-ul-dec} \& \Lineref{tryc-ul-cmt} of \tryc \mth. Suppose $\ttutl{i}$ gets decremented to some value $\alpha$ by one of these lines. 

Now there are two cases here: (1) Suppose $\ttutl{i}$ gets decremented to $\ct_i$ due to \Lineref{tryc-ul-cmt} of \tryc \mth. Then from \lemref{ti|tltl-comt}, we have $\ttltl{i} \leq \ct_i =  \ttutl{i}$. Thus in this case, $T_i$ will not abort. (2) $\ttutl{i}$ gets decremented to $\alpha$ which is less than $\ct_i$. Then from \lemref{ti|tutl-comt}, we get that there is a committed transaction $T_j$ in $\comm{H2}$ such that $\twts{j} > \twts{i}$. This implies that $T_j$ is in $\comm{H3}$.

\item \Lineref{its|lar-sml}: This case is similar to \csref{init-tc-chk}, i.e., \Lineref{init-tc-chk}.

\item \Lineref{its-chk2} \label{case:its-chk2}: In this case, $T_k$ is in $T_i$'s \srl and is committed in $H1$. And, from this case, we have that

\begin{equation}
\label{eq:tltl-k_i}
\htutl{i}{H2} \leq \htltl{k}{H2}
\end{equation}

From the assumption of this case, we have that $T_k$ commits before $T_i$. Thus, from \lemref{tltlj-comti}, we get that $\ct_k < \ct_i$. From \lemref{ti|tltl-comt}, we have that $\ttltl{k} \leq \ct_k$. Thus, we get that $\ttltl{k} < \ct_i$. Combining this with the inequality of this case \eqnref{tltl-k_i}, we get that $\ttutl{i} < \ct_i$.

Combining this inequality with \lemref{ti|tutl-comt}, we get that there is a transaction $T_j$ in $\comm{H2}$ and $\htwts{j}{H2} > \htwts{i}{H2}$. This implies that $T_j$ is in $\comm{H3}$ as well.

\end{enumerate}

\item STM read: 

\begin{enumerate}
\item \Lineref{rd-chk}: This case is similar to \csref{init-tc-chk}, i.e., \Lineref{init-tc-chk}

\item \Lineref{rd-lts-cross}: The reasoning here is similar to \csref{tc-lts-cross}, i.e., \Lineref{tc-lts-cross}.
\end{enumerate}

\end{enumerate}

\end{proof}


The interesting aspect of the above lemma is that it gives us a insight as to when a $T_i$ will get commit. If an \itsen transaction $T_i$ aborts then it is because of another transaction $T_j$ with \wts higher than $T_i$ has committed. To precisely capture this, we define two more notions of a transaction being enabled \emph{\cdsen} and \emph{\finen}. To define these notions of enabled, we in turn define a few other auxiliary notions. We start with \emph{\affset},
\begin{equation*}
\haffset{i}{H} = \{T_j|(T_j \in \txns{H}) \land (\htits{j}{H} < \htits{i}{H} + 2*L)\}
\end{equation*}

From the description of \ksftm algorithm and \lemref{wts-its}, it can be seen that a transaction $T_i$'s commit can depend on committing of transactions (or their \inc{s}) which have their \its less than \its of $T_i$ + $2*L$, which is $T_i$'s \affset. We capture this notion of dependency for a transaction $T_i$ in a history $H$ as \emph{commit dependent set} or \emph{\cdset} as: the set of all transactions $T_j$ in $T_i$'s \affset that do not any \inc that is committed yet, i.e., not yet have their \incct flag set as true. Formally, 

\begin{equation*}
\hcds{i}{H} = \{T_j| (T_j \in \haffset{i}{H}) \land (\neg\inct{j}{H}) \}
\end{equation*}

\noindent Based on this definition of \cdset, we next define the notion of \cdsen. 

\begin{definition}
	\label{defn:cdsen}
	We say that transaction $T_i$ is \emph{\cdsen} if the following conditions hold true (1) $T_i$ is live in $H$; (2) \cts of $T_i$ is greater than or equal to \its of $T_i$ + $2*L$; (3) \cdset of $T_i$ is empty, i.e., for all transactions $T_j$ in $H$ with \its lower than \its of  $T_i$ + $2*L$ in $H$ have their \incct to be true. Formally, 
	
	\begin{equation*}
	\cdsenb{i}{H} = \begin{cases}
	True    & (T_i \in \live{H}) \land (\htcts{i}{H} \geq \htits{i}{H} + 2*L) \land (\hcds{i}{H} = \phi) \\
	False	& \text{otherwise}
	\end{cases}
	\end{equation*}
\end{definition}

\noindent The meaning and usefulness of these definitions will become clear in the course of the proof. In fact, we later show that once the transaction $T_i$ is \cdsen, it will eventually commit. We will start with a few lemmas about these definitions. 

\begin{lemma}
	\label{lem:its-enb}	Consider a transaction $T_i$ in a history $H$. If $T_i$ is \cdsen then $T_i$ is also \itsen. Formally, $\langle H, T_i: (T_i \in \txns{H}) \land (\cdsenb{i}{H}) \implies (\itsenb{i}{H}) \rangle$. 
\end{lemma}

\begin{proof}
	If $T_i$ is \cdsen in $H$ then it implies that $T_i$ is live in $H$. From the definition of \cdsen, we get that $\hcds{i}{H}$ is $\phi$ implying that any transaction $T_j$ with $\tits{k}$ less than $\tits{i} + 2*L$ has its \incct flag as true in $H$. Hence, for any transaction $T_k$ having $\tits{k}$ less than $\tits{i}$, $\inct{k}{H}$ is also true. This shows that $T_i$ is \itsen in $H$.
\end{proof}

\ignore{
\begin{lemma}
\label{lem:cds-h1}
Consider a transaction $T_i$ which is \cdsen in a history $H1$. Let $T_j$ be a transaction in \affset of $T_i$ in $H1$. Consider an extension of $H1$, $H2$ with a transaction $T_k$ in it such that $T_k$ is an \inc of $T_j$. Then $T_k$ is also in the set of transaction of $H1$. Formally, $\langle H1, H2, T_i, T_j, T_k: (H1 \sqsubseteq H2) \land  (\cdsenb{i}{H1}) \land (T_j \in \haffset{i}{H1}) \land (T_k \in \incs{j}{H2}) \implies (T_k \in \txns{H1}) \rangle$
\end{lemma}

\begin{proof}
Once $T_i$ becomes \cdsen, all the transactions in its \affset have an \inc that is committed. Hence, as per our model the corresponding \aptr has committed and the application does not invoke another transaction with the same \its. 

Thus from \obsref{cmt-noinc}, we get that no new \inc of $T_j$ will get invoked by the application in any future extension of $H1$. This implies that the \inc of $T_j$ in $H2$, $T_k$ must have already been invoked before $T_i$ became \enbd. Since $T_i$ is \enbd in $H1$, we get that $T_k$ must also be in the set of transactions of $H1$, i.e., $(T_k \in \txns{H1})$.
\end{proof}
}

\begin{lemma}
\label{lem:cds-tk-h1}
Consider a transaction $T_i$ which is \cdsen in a history $H1$. Consider an extension of $H1$, $H2$ with a transaction $T_j$ in it such that $T_i$ is an \inc of $T_j$. Let $T_k$ be a transaction in the \affset of $T_j$ in $H2$ Then $T_k$ is also in the set of transaction of $H1$. Formally, $\langle H1, H2, T_i, T_j, T_k: (H1 \sqsubseteq H2) \land  (\cdsenb{i}{H1}) \land (T_i \in \incs{j}{H2}) \land (T_k \in \haffset{j}{H2}) \implies (T_k \in \txns{H1}) \rangle$
\end{lemma}

\begin{proof} 
Since $T_i$ is \cdsen in $H1$, we get (from the definition of \cdsen) that 
\begin{equation}
\label{eq:ti-cts-its}
\htcts{i}{H1} \geq \htits{i}{H1} + 2*L
\end{equation}

Here, we have that $T_k$ is in $\haffset{j}{H2}$. Thus from the definition of \affset, we get that 
\begin{equation}
\label{eq:tk-tj-aff}
\htits{k}{H2} < \htits{j}{H2} + 2*L
\end{equation}

Since $T_i$ and $T_j$ are \inc{s} of each other, their \its are the same. Combining this with \eqnref{tk-tj-aff}, we get that 
\begin{equation}
\label{eq:tk-ti-h12}
\htits{k}{H2} < \htits{i}{H1} + 2*L
\end{equation}

We now show this proof through contradiction. Suppose $T_k$ is not in $\txns{H1}$. Then there are two cases:

\begin{itemize}
\item No \inc of $T_k$ is in $H1$: This implies that $T_k$ starts afresh after $H1$. Since $T_k$ is not in $H1$, from \corref{cts-syst} we get that

$\htcts{k}{H2} > \hsyst{H1} \xrightarrow [\htcts{k}{H2} = \htits{k}{H2}] {T_k \text{ starts afresh}}\htits{k}{H2} > \hsyst{H1} \xrightarrow [\hsyst{H1} \geq \htcts{i}{H1}]{(T_i \in H1) \land \lemref{cts-syst}} \htits{k}{H2} > \htcts{i}{H1} \xrightarrow {\eqnref{ti-cts-its}} \htits{k}{H2} > \htits{i}{H1} + 2*L \xrightarrow {\htits{i}{H1} = \htits{j}{H2}} \htits{k}{H2} > \htits{j}{H2} + 2*L$

But this result contradicts with \eqnref{tk-tj-aff}. Hence, this case is not possible. 

\item There is an \inc of $T_k$, $T_l$ in $H1$: In this case, we have that 

\begin{equation}
\label{eq:tl-h1}
\htits{l}{H1} = \htits{k}{H2}
\end{equation}


Now combing this result with \eqnref{tk-ti-h12}, we get that $\htits{l}{H1} < \htits{i}{H1} + 2*L$. This implies that $T_l$ is in \affset of $T_i$ in $H1$. Since $T_i$ is \cdsen, we get that $T_l$'s \incct must be true. 

We also have that $T_k$ is not in $H1$ but in $H2$ where $H2$ is an extension of $H1$. Since $H2$ has some events more than $H1$, we get that $H2$ is a strict extension of $H1$.

Thus, we have that, $(H1 \sqsubset H2) \land (\inct{l}{H1}) \land (T_k \in \txns{H2}) \land (T_k \notin \txns{H1})$. Combining these with \lemref{inct-diff}, we get that $(\htits{l}{H1} \neq \htits{k}{H2})$. But this result contradicts \eqnref{tl-h1}. Hence, this case is also not possible.
\end{itemize}
Thus from both the cases we get that $T_k$ should be in $H1$. Hence proved.
\end{proof}

\begin{lemma}
\label{lem:aff-tkinc-h1}
Consider two histories $H1, H2$ where $H2$ is an extension of $H1$. Let $T_i, T_j, T_k$ be three transactions such that $T_i$ is in $\txns{H1}$ while $T_j, T_k$ are in $\txns{H2}$. Suppose we have that (1) $\tcts{i}$ is greater than $\tits{i} + 2*L$ in $H1$; (2) $T_i$ is an \inc of $T_j$; (3) $T_k$ is in \affset of $T_j$ in $H2$. Then an \inc of $T_k$, say $T_l$ (which could be same as $T_k$) is in $\txns{H1}$. Formally, $\langle H1, H2, T_i, T_j, T_k: (H1 \sqsubseteq H2) \land (T_i \in \txns{H1})  \land (\{T_j, T_k\} \in \txns{H2}) \land (\htcts{i}{H1} > \htits{i}{H1} + 2*L) \land (T_i \in \incs{j}{H2}) \land (T_k \in \haffset{j}{H2}) \implies (\exists T_l: (T_l \in \incs{k}{H2}) \land (T_l \in \txns{H1})) \rangle$
\end{lemma}

\begin{proof} 

\noindent This proof is similar to the proof of \lemref{cds-tk-h1}. We are given that 
\begin{equation}
\label{eq:given-ti-ctsits}
\htcts{i}{H1} \geq \htits{i}{H1} + 2*L
\end{equation}

We now show this proof through contradiction. Suppose no \inc of $T_k$ is in $\txns{H1}$. This implies that $T_k$ must have started afresh in some history $H3$ which is an extension of $H1$. Also note that $H3$ could be same as $H2$ or a prefix of it, i.e., $H3 \sqsubseteq H2$. Thus, we have that 

\noindent
\begin{math}
\htits{k}{H3} > \hsyst{H1} \xrightarrow{\lemref{cts-syst}} \htits{k}{H3} > \htcts{i}{H1} \xrightarrow{\eqnref{given-ti-ctsits}} \htits{k}{H3} > \htits{i}{H1} + 2*L \xrightarrow{\htits{i}{H1} = \htits{j}{H2}} \htits{k}{H3} > \htits{j}{H2} + 2*L \xrightarrow[\obsref{hist-subset}]{H3 \sqsubseteq H2} \htits{k}{H2} > \htits{j}{H2} + 2*L \xrightarrow[definition]{\affset} T_k \notin \haffset{j}{H2}
\end{math}

But we are given that $T_k$ is in \affset of $T_j$ in $H2$. Hence, it is not possible that $T_k$ started afresh after $H1$. Thus, $T_k$ must have a \inc in $H1$.
\end{proof}

\begin{lemma}
\label{lem:aff-same}
Consider a transaction $T_i$ which is \cdsen in a history $H1$. Consider an extension of $H1$, $H2$ with a transaction $T_j$ in it such that $T_j$ is an \inc of $T_i$ in $H2$. Then \affset of $T_i$ in $H1$ is same as the \affset of $T_j$ in $H2$. Formally, $\langle H1, H2, T_i, T_j: (H1 \sqsubseteq H2) \land  (\cdsenb{i}{H1}) \land (T_j \in \txns{H2}) \land (T_i \in \incs{j}{H2}) \implies ((\haffset{i}{H1} = \haffset{j}{H2})) \rangle$
\end{lemma}

\begin{proof}
From the definition of \cdsen, we get that $T_i$ is in $\txns{H1}$. Now to prove that \affset{s} are the same, we have to show that $(\haffset{i}{H1} \subseteq \haffset{j}{H2})$ and $(\haffset{j}{H1} \subseteq \haffset{i}{H2})$. We show them one by one:

\paragraph{$(\haffset{i}{H1} \subseteq \haffset{j}{H2})$:} Consider a transaction $T_k$ in $\haffset{i}{H1}$. We have to show that $T_k$ is also in $\haffset{j}{H2}$. From the definition of \affset, we get that 
\begin{equation}
\label{eq:tk-h1}
T_k \in \txns{H1}
\end{equation}

\noindent Combining \eqnref{tk-h1} with \obsref{hist-subset}, we get that 
\begin{equation}
\label{eq:tk-h2}
T_k \in \txns{H2}
\end{equation}

\noindent From the definition of \its, we get that 
\begin{equation}
\label{eq:its-h1-h2}
\htits{k}{H1} = \htits{k}{H2}
\end{equation}

\noindent Since $T_i, T_j$ are \inc{s} we have that . 
\begin{equation}
\label{eq:its-ij}
\htits{i}{H1} = \htits{j}{H2}
\end{equation}

\noindent From the definition of \affset, we get that, \\ 
$\htits{k}{H1} < \htits{i}{H1} + 2*L \xrightarrow{\eqnref{its-h1-h2}} \htits{k}{H2} < \htits{i}{H1} + 2*L \xrightarrow{\eqnref{its-ij}} \htits{k}{H2} < \htits{j}{H2} + 2*L$

\noindent Combining this result with \eqnref{tk-h2}, we get that $T_k \in \haffset{j}{H2}$. 

\paragraph{$(\haffset{i}{H1} \subseteq \haffset{j}{H2})$:} Consider a transaction $T_k$ in $\haffset{j}{H2}$. We have to show that $T_k$ is also in $\haffset{i}{H1}$. From the definition of \affset, we get that $T_k \in \txns{H2}$.

Here, we have that $(H1 \sqsubseteq H2) \land  (\cdsenb{i}{H1}) \land (T_i \in \incs{j}{H2}) \land (T_k \in \haffset{j}{H2})$. Thus from \lemref{cds-tk-h1}, we get that $T_k \in \txns{H1}$. Now, this case is similar to the above case. It can be seen that Equations \ref{eq:tk-h1}, \ref{eq:tk-h2}, \ref{eq:its-h1-h2}, \ref{eq:its-ij} hold good in this case as well. 

Since $T_k$ is in $\haffset{j}{H2}$, we get that \\
$\htits{k}{H2} < \htits{i}{H2} + 2*L \xrightarrow{\eqnref{its-h1-h2}} \htits{k}{H1} < \htits{j}{H2} + 2*L \xrightarrow{\eqnref{its-ij}} \htits{k}{H1} < \htits{i}{H1} + 2*L $

\noindent Combining this result with \eqnref{tk-h1}, we get that $T_k \in \haffset{i}{H1}$.
\end{proof}

\noindent Next we explore how a \cdsen transaction remains \cdsen in the future histories once it becomes true.

\begin{lemma}
\label{lem:cds-fut}
Consider two histories $H1$ and $H2$ with $H2$ being an extension of $H1$. Let  $T_i$ and $T_j$ be two transactions which are live in $H1$ and $H2$ respectively. Let $T_i$ be an \inc of $T_j$ and $\tcts{i}$ is less than $\tcts{j}$. Suppose $T_i$ is \cdsen in $H1$. Then $T_j$ is \cdsen in $H2$ as well. Formally, $\langle H1, H2, T_i, T_j: (H1 \sqsubseteq H2) \land (T_i \in \live{H1}) \land (T_j \in \live{H2}) \land (T_i \in \incs{j}{H2}) \land (\htcts{i}{H1} < \htcts{j}{H2}) \land (\cdsenb{i}{H1}) \implies (\cdsenb{j}{H2}) \rangle$. 
\end{lemma}

\begin{proof}
We have that $T_i$ is live in $H1$ and $T_j$ is live in $H2$. Since $T_i$ is \cdsen in $H1$, we get (from the definition of \cdsen) that 
\begin{equation}
\label{eq:cts-its}
\htcts{i}{H1} \geq \htits{i}{H2} + 2*L
\end{equation}

We are given that $\tcts{i}$ is less than $\tcts{j}$ and $T_i, T_j$ are incarnations of each other. Hence, we have that

\ignore{
\begin{align*}
\htcts{j}{H2} & > \htcts{i}{H1} \\
\htcts{j}{H2} & > \htits{i}{H1} + 2*L & [\text{From \eqnref{cts-its}}] \\
\htcts{j}{H2} & > \htits{j}{H2} + 2*L & [\tits{i} = \tits{j}] \\
\end{align*}
}

\begin{align*}
\htcts{j}{H2} & > \htcts{i}{H1} \\
& > \htits{i}{H1} + 2*L & [\text{From \eqnref{cts-its}}] \\
& > \htits{j}{H2} + 2*L & [\tits{i} = \tits{j}] \\
\end{align*}

Thus we get that $\tcts{j} > \tits{j} + 2*L$. We have that $T_j$ is live in $H2$. In order to show that $T_j$ is \cdsen in $H2$, it only remains to show that \cdset of $T_j$ in $H2$ is empty, i.e., $\hcds{j}{H2} = \phi$. The \cdset becomes empty when all the transactions of $T_j$'s \affset in $H2$ have their \incct as true in $H2$.

Since $T_j$ is live in $H2$, we get that $T_j$ is in $\txns{H2}$. Here, we have that $(H1 \sqsubseteq H2) \land (T_j \in \txns{H2}) \land (T_i \in \incs{j}{H2}) \land (\cdsenb{i}{H1})$. Combining this with \lemref{aff-same}, we get that $\haffset{i}{H1} = \haffset{j}{H2}$.

Now, consider a transaction $T_k$ in $ \haffset{j}{H2}$. From the above result, we get that $T_k$ is also in $\haffset{i}{H1}$. Since $T_i$ is \cdsen in $H1$, i.e., $\cdsenb{i}{H1}$ is true, we get that $\inct{k}{H1}$ is true. Combining this with \obsref{inct-fut}, we get that $T_k$ must have its \incct as true in $H2$ as well, i.e. $\inct{k}{H2}$. This implies that all the transactions in $T_j$'s \affset have their \incct flags as true in $H2$. Hence the $\hcds{j}{H2}$ is empty. As a result, $T_j$ is \cdsen in $H2$, i.e., $\cdsenb{j}{H2}$. 
\end{proof}

Having defined the properties related to \cdsen, we start defining notions for \finen. Next, we define \emph{\maxwts} for a transaction $T_i$ in $H$ which is the transaction $T_j$ with the largest \wts in $T_i$'s \incset. Formally,
\begin{equation*}
\hmaxwts{i}{H} = max\{\htwts{j}{H}|(T_j \in \incs{i}{H})\}
\end{equation*}

\noindent From this definition of \maxwts, we get the following simple observation. 

\begin{observation}
\label{obs:max-wts} 
For any transaction $T_i$ in $H$, we have that $\twts{i}$ is less than or equal to $\hmaxwts{i}{H}$. Formally, $\htwts{i}{H} \leq \hmaxwts{i}{H}$.
\end{observation}

Next, we combine the notions of \affset and \maxwts to define \emph{\affwts}. It is the maximum of \maxwts of all the transactions in its \affset. Formally, 
\begin{equation*}
\haffwts{i}{H} = max\{\hmaxwts{j}{H}|(T_j \in \haffset{i}{H})\}
\end{equation*}

\noindent Having defined the notion of \affwts, we get the following lemma relating the \affset and \affwts of two transactions. 

\begin{lemma}
\label{lem:affwts-same}
Consider two histories $H1$ and $H2$ with $H2$ being an extension of $H1$. Let  $T_i$ and $T_j$ be two transactions which are live in $H1$ and $H2$ respectively. Suppose the \affset of $T_i$ in $H1$ is same as \affset of $T_j$ in $H2$. Then the \affwts of $T_i$ in $H1$ is same as \affwts of $T_j$ in $H2$. Formally, $\langle H1, H2, T_i, T_j: (H1 \sqsubseteq H2) \land (T_i \in \txns{H1}) \land (T_j \in \txns{H2}) \land (\haffset{i}{H1} = \haffset{j}{H2}) \implies (\haffwts{i}{H1} = \haffwts{j}{H2}) \rangle$. 
\end{lemma}

\begin{proof}

From the definition of \affwts, we get the following equations
\begin{equation}
\label{eq:h1-ti-affwts}
\haffwts{i}{H} = max\{\hmaxwts{k}{H}|(T_k \in \haffset{i}{H1})\}
\end{equation}

\begin{equation}
\label{eq:h2-tj-affwts}
\haffwts{j}{H} = max\{\hmaxwts{l}{H}|(T_l \in \haffset{j}{H2})\}
\end{equation}

From these definitions, let us suppose that $\haffwts{i}{H1}$ is $\hmaxwts{p}{H1}$ for some transaction $T_p$ in $\haffset{i}{H1}$. Similarly, suppose that $\haffwts{j}{H2}$ is $\hmaxwts{q}{H2}$ for some transaction $T_q$ in $\haffset{j}{H2}$.

Here, we are given that $\haffset{i}{H1} = \haffset{j}{H2})$. Hence, we get that $T_p$ is also in $\haffset{i}{H1}$. Similarly, $T_q$ is in $\haffset{j}{H2}$ as well. Thus from Equations \eqref{eq:h1-ti-affwts} \& \eqref{eq:h2-tj-affwts}, we get that 

\begin{equation}
\label{eq:ti-tp-max}
\hmaxwts{p}{H1} \geq  \hmaxwts{q}{H2}
\end{equation}

\begin{equation}
\label{eq:tj-tq-max}
\hmaxwts{q}{H2} \geq \hmaxwts{p}{H1}
\end{equation}

Combining these both equations, we get that $\hmaxwts{p}{H1} = \hmaxwts{q}{H2}$ which in turn implies that $\haffwts{i}{H1} = \haffwts{j}{H2}$.

\end{proof}

\noindent Finally, using the notion of \affwts and \cdsen, we define the notion of \emph{\finen}

\begin{definition}
\label{defn:finen}
We say that transaction $T_i$ is \emph{\finen} if the following conditions hold true (1) $T_i$ is live in $H$; (2) $T_i$ is \cdsen is $H$; (3) $\htwts{j}{H}$ is greater than $\haffwts{i}{H}$. Formally, 

\begin{equation*}
\finenb{i}{H} = \begin{cases}
True    & (T_i \in \live{H}) \land (\cdsenb{i}{H}) \land (\htwts{j}{H} > \haffwts{i}{H}) \\
False	& \text{otherwise}
\end{cases}
\end{equation*}
\end{definition}


It can be seen from this definition, a transaction that is \finen is also \cdsen. We now show that just like \itsen and \cdsen, once a transaction is \finen, it remains \finen until it terminates. The following lemma captures it. 

\ignore{
\begin{lemma}
\label{lem:fin-sam-fut}
Consider two histories $H1$ and $H2$ with $H2$ being an extension of $H1$. Let  $T_i$ be a transaction live in $H1$ and $H2$. Suppose $T_i$ is \finen in $H1$. Then $T_i$ is \finen in $H2$ as well. Formally, $\langle H1, H2, T_i, T_j: (H1 \sqsubseteq H2) \land (T_i \in \{\live{H1} \cup \live{H2}\}) \land (\finenb{i}{H1}) \implies (\finenb{i}{H2}) \rangle$. 
\end{lemma} \todo{Proof to be added here}
}

\begin{lemma}
\label{lem:fin-fut}
Consider two histories $H1$ and $H2$ with $H2$ being an extension of $H1$. Let  $T_i$ and $T_j$ be two transactions which are live in $H1$ and $H2$ respectively. Suppose $T_i$ is \finen in $H1$. Let $T_i$ be an \inc of $T_j$ and $\tcts{i}$ is less than $\tcts{j}$. Then $T_j$ is \finen in $H2$ as well. Formally, $\langle H1, H2, T_i, T_j: (H1 \sqsubseteq H2) \land (T_i \in \live{H1}) \land (T_j \in \live{H2}) \land (T_i \in \incs{j}{H2}) \land (\htcts{i}{H1} < \htcts{j}{H2}) \land (\finenb{i}{H1}) \implies (\finenb{j}{H2}) \rangle$. 
\end{lemma}

\begin{proof}
Here we are given that $T_j$ is live in $H2$. Since $T_i$ is \finen in $H1$, we get that it is \cdsen in $H1$ as well. Combining this with the conditions given in the lemma statement, we have that, \\ 

\begin{equation}
\label{eq:fin-given}
\begin{split}
\langle (H1 \sqsubseteq H2) \land (T_i \in \live{H1}) \land (T_j \in \live{H2}) \land (T_i \in \incs{j}{H2}) \land (\htcts{i}{H1} < \htcts{j}{H2}) \\
\land (\cdsenb{i}{H1}) \rangle
\end{split}
\end{equation}

Combining \eqnref{fin-given} with \lemref{cds-fut}, we get that $T_j$ is \cdsen in $H2$, i.e., $\cdsenb{j}{H2}$. Now, in order to show that $T_j$ is \finen in $H2$ it remains for us to show that $\htwts{j}{H2} > \haffwts{j}{H2}$.

We are given that $T_j$ is live in $H2$ which in turn implies that $T_j$ is in $\txns{H2}$. Thus changing this in \eqnref{fin-given}, we get the following 
\begin{equation}
\label{eq:mod-given}
\begin{split}
\langle (H1 \sqsubseteq H2) \land (T_j \in \txns{H2}) \land (T_i \in \incs{j}{H2}) \land (\htcts{i}{H1} < \htcts{j}{H2}) \\
\land (\cdsenb{i}{H1}) \rangle
\end{split}
\end{equation}

\noindent Combining \eqnref{mod-given} with \lemref{aff-same} we get that 
\begin{equation}
\label{eq:affs-eq}
\haffwts{i}{H1} = \haffwts{j}{H2}
\end{equation}

\noindent We are given that $\htcts{i}{H1} < \htcts{j}{H2}$. Combining this with the definition of \wts, we get 
\begin{equation}
\label{eq:titj-wts}
\htwts{i}{H1} < \htwts{j}{H2}
\end{equation}

\noindent Since $T_i$ is \finen in $H1$, we have that \\
$\htwts{i}{H1} > \haffwts{i}{H1} \xrightarrow{\eqnref{titj-wts}} \htwts{j}{H2} > \haffwts{i}{H1} \xrightarrow{\eqnref{affs-eq}} \htwts{j}{H2} > \\
\haffwts{j}{H2}$

\end{proof}

\noindent Now, we show that a transaction that is \finen will eventually commit. 


\begin{lemma}
\label{lem:enbd-ct}
Consider a live transaction $T_i$ in a history $H1$. Suppose $T_i$ is \finen in $H1$ and $\tval{i}$ is true in $H1$. Then there exists an extension of $H1$, $H3$ in which $T_i$ is committed. Formally, $\langle H1, T_i: (T_i \in \live{H1}) \land (\htval{i}{H1}) \land (\finenb{i}{H1}) \implies (\exists H3: (H1 \sqsubset H3) \land (T_i \in \comm{H3})) \rangle$. 
\end{lemma}

\begin{proof}
Consider a history $H3$ such that its \syst being greater than $\tcts{i} + L$. We will prove this lemma using contradiction. Suppose $T_i$ is aborted in $H3$. 

Now consider $T_i$ in $H1$: $T_i$ is live; its \val flag is true; and is \enbd. From the definition of \finen, we get that it is also \cdsen. From \lemref{its-enb}, we get that $T_i$ is \itsen in $H1$. Thus from \lemref{its-wts}, we get that there exists an extension of $H1$, $H2$ such that (1) Transaction $T_i$ is live in $H2$; (2) there is a transaction $T_j$ in ${H2}$; (3) $\htwts{j}{H2}$ is greater than $\htwts{i}{H2}$; (4) $T_j$ is committed in $H3$. Formally, 

\begin{equation}
\label{eq:its-wts-ant}
\begin{split}
\langle (\exists H2, T_j: (H1 \sqsubseteq H2 \sqsubset H3) \land (T_i \in \live{H2}) \land (T_j \in \txns{H2}) \land (\htwts{i}{H2} < \htwts{j}{H2}) \\
\land (T_j \in \comm{H3})) \rangle
\end{split}
\end{equation}

Here, we have that $H2$ is an extension of $H1$ with $T_i$ being live in both of them and $T_i$ is \finen in $H1$. Thus from \lemref{fin-fut}, we get that $T_i$ is \finen in $H2$ as well. Now, let us consider $T_j$ in $H2$. From \eqnref{its-wts-ant}, we get that $(\htwts{i}{H2} < \htwts{j}{H2})$. Combining this with the observation that $T_i$ being live in $H2$, \lemref{wts-its} we get that $(\htits{j}{H2} \leq \htits{i}{H2} + 2*L)$.


This implies that $T_j$ is in \affset of $T_i$ in $H2$, i.e., $(T_j \in \haffset{i}{H2})$. From the definition of \affwts, we get that 

\begin{equation}
\label{eq:max-affwts}
(\haffwts{i}{H2} \geq \hmaxwts{j}{H2}) 
\end{equation}

Since $T_i$ is \finen in $H2$, we get that $\twts{i}$ is greater than \affwts of $T_i$ in $H2$. 
\begin{equation}
\label{eq:wts-affwts}
(\htwts{i}{H2} > \haffwts{i}{H2}) 
\end{equation}

Now combining Equations \ref{eq:max-affwts}, \ref{eq:wts-affwts} we get,
\begin{align*}
\htwts{i}{H2} & > \haffwts{i}{H2} \geq \hmaxwts{j}{H2} & \\
& > \haffwts{i}{H2} \geq \hmaxwts{j}{H2} \geq \htwts{j}{H2}  & [\text{From \obsref{max-wts}}] \\
& > \htwts{j}{H2}
\end{align*}

But this equation contradicts with \eqnref{its-wts-ant}. Hence our assumption that $T_i$ will get aborted in $H3$ after getting \finen is not possible. Thus $T_i$ has to commit in $H3$.
\end{proof}

\noindent Next we show that once a transaction $T_i$ becomes \itsen, it will eventually become \finen as well and then committed. We show this change happens in a sequence of steps. We first show that Transaction $T_i$ which is \itsen first becomes \cdsen (or gets committed). We next show that $T_i$ which is \cdsen becomes \finen or get committed. On becoming \finen, we have already shown that $T_i$ will eventually commit. 

Now, we show that a transaction that is \itsen will become \cdsen or committed. To show this, we introduce a few more notations and definitions. We start with the notion of \emph{\depits (dependent-its)} which is the set of \its{s} that a transaction $T_i$ depends on to commit. It is the set of \its of all the transactions in $T_i$'s \cdset in a history $H$. Formally, 

\begin{equation*}
\hdep{i}{H} = \{\htits{j}{H}|T_j \in \hcds{i}{H}\}
\end{equation*}

\noindent We have the following lemma on the \depits of a transaction $T_i$ and its future \inc $T_j$ which states that \depits of a $T_i$ either reduces or remains the same. 

\begin{lemma}
\label{lem:depits-fut}
Consider two histories $H1$ and $H2$ with $H2$ being an extension of $H1$. Let  $T_i$ and $T_j$ be two transactions which are live in $H1$ and $H2$ respectively and $T_i$ is an \inc of $T_j$. In addition, we also have that $\tcts{i}$ is greater than $\tits{i} + 2*L$ in $H1$. Then, we get that $\hdep{j}{H2}$ is a subset of $\hdep{i}{H1}$. Formally, $\langle H1, H2, T_i, T_j: (H1 \sqsubseteq H2) \land (T_i \in \live{H1}) \land (T_j \in \live{H2}) \land (T_i \in \incs{j}{H2}) \land (\htcts{i}{H1} \geq \htits{i}{H1} + 2*L) \implies (\hdep{j}{H2} \subseteq \hdep{i}{H1}) \rangle$. 
\end{lemma}

\begin{proof}
Suppose $\hdep{j}{H2}$ is not a subset of $\hdep{i}{H1}$. This implies that there is a transaction $T_k$ such that $\htits{k}{H2} \in \hdep{j}{H2}$ but $\htits{k}{H1} \notin \hdep{j}{H1}$. This implies that $T_k$ starts afresh after $H1$ in some history say $H3$ such that $H1 \sqsubset H3 \sqsubseteq H2$. Hence, from \corref{cts-syst} we get the following

\noindent
\begin{math}
\htits{k}{H3} > \hsyst{H1} \xrightarrow{\lemref{cts-syst}} \htits{k}{H3} > \htcts{i}{H1} \implies \htits{k}{H3} > \htits{i}{H1} + 2*L \xrightarrow{\htits{i}{H1} = \htits{j}{H2}} \htits{k}{H3} > \htits{j}{H2} + 2*L \xrightarrow[definitions]{\affset, \depits} \htits{k}{H2} \notin \hdep{j}{H2}
\end{math}

We started with $\tits{k}$ in $\hdep{j}{H2}$ and ended with $\tits{k}$ not in $\hdep{j}{H2}$. Thus, we have a contradiction. Hence, the lemma follows.

\end{proof}

\noindent Next we denote the set of committed transactions in $T_i$'s \affset in $H$ as \emph{\cis (commit independent set)}. Formally, 

\begin{equation*}
\hcis{i}{H} = \{T_j| (T_j \in \haffset{i}{H}) \land (\inct{j}{H}) \}
\end{equation*}

\noindent In other words, we have that $\hcis{i}{H} = \haffset{i}{H} - \hcds{i}{H}$. Finally, using the notion of \cis we denote the maximum of \maxwts of all the transactions in $T_i$'s \cis as \emph{\pawts} (partly affecting \wts). It turns out that the value of \pawts affects the commit of $T_i$ which we show in the course of the proof. Formally, \pawts is defined as 

\begin{equation*}
\hpawts{i}{H} = max\{\hmaxwts{j}{H}|(T_j \in \hcis{i}{H})\}
\end{equation*}

\noindent Having defined the required notations, we are now ready to show that a \itsen transaction will eventually become \cdsen. 

\begin{lemma}
\label{lem:its-cds}
Consider a transaction $T_i$ which is live in a history $H1$ and $\tcts{i}$ is greater than or equal to $\tits{i} + 2*L$. If $T_i$ is \itsen in $H1$ then there is an extension of $H1$, $H2$ in which an \inc $T_i$, $T_j$ (which could be same as $T_i$), is either committed or \cdsen. Formally, $\langle H1, T_i: (T_i \in \live{H1}) \land (\htcts{i}{H1} \geq \htits{i}{H1} + 2*L) \land (\itsenb{i}{H1}) \implies (\exists H2, T_j: (H1 \sqsubset H2) \land (T_j \in \incs{i}{H2}) \land ((T_j \in \comm{H2}) \lor (\cdsenb{j}{H2}))) \rangle$. 
\end{lemma}

\begin{proof}
We prove this by inducting on the size of $\hdep{i}{H1}$, $n$. For showing this, we define a boolean function $P(k)$ as follows:

\begin{math}
P(k) = \begin{cases}
True & \langle H1, T_i: (T_i \in \live{H1}) \land (\htcts{i}{H1} \geq \htits{i}{H1} + 2*L) \land (\itsenb{i}{H1}) \land \\
& (k \geq |\hdep{i}{H1}|) \implies (\exists H2, T_j: (H1 \sqsubset H2) \land (T_j \in \incs{i}{H2}) \land \\
& ((T_j \in \comm{H2}) \lor (\cdsenb{j}{H2}))) \rangle \\
False & \text{otherwise}
\end{cases}
\end{math}

As can be seen, here $P(k)$ means that if (1) $T_i$ is live in $H1$; (2) $\tcts{i}$ is greater than or equal to $\tits{i} + 2*L$; (3) $T_i$ is \itsen in $H1$ (4) the size of $\hdep{i}{H1}$ is less than or equal to $k$;  then there exists a history $H2$ with a transaction $T_j$ in it which is an \inc of $T_i$ such that $T_j$ is either committed or \cdsen in $H2$. We show $P(k)$ is true for all (integer) values of $k$ using induction. 

\vspace{1mm}
\noindent
\textbf{Base Case - $P(0)$:} Here, from the definition of $P(0)$, we get that $|\hdep{i}{H1}| = 0$. This in turn implies that $\hcds{i}{H1}$ is null. Further, we are already given that $T_i$ is live in $H1$ and $\htcts{i}{H1} \geq \htits{i}{H1} + 2*L$. Hence, all these imply that $T_i$ is \cdsen in $H1$. 

\vspace{1mm}
\noindent
\textbf{Induction case - To prove $P(k+1)$ given that $P(k)$ is true:} If $|\hdep{i}{H1}| \leq k$, from the induction hypothesis $P(k)$, we get that $T_j$ is either committed or \cdsen in $H2$. Hence, we consider the case when 

\begin{equation}
\label{eq:hdep-assn}
|\hdep{i}{H1}| = k + 1
\end{equation}

Let $\alpha$ be $\hpawts{i}{H1}$. Suppose $\htwts{i}{H1} < \alpha$. Then from \lemref{wts-great}, we get that there is an extension of $H1$, say $H3$ in which an \inc of $T_i$, $T_l$ (which could be same as $T_i$) is committed or is live in $H3$ and has \wts greater than $\alpha$. If $T_l$ is committed then $P(k+1)$ is trivially true. So we consider the latter case in which $T_l$ is live in $H3$. In case $\htwts{i}{H1} \geq \alpha$, then in the analysis below follow where we can replace $T_l$ with $T_i$.

Next, suppose $T_l$ is aborted in an extension of $H3$, $H5$. Then from \lemref{its-wts}, we get that there exists an extension of $H3$, $H4$ in which (1) $T_l$ is live; (2) there is a transaction $T_m$ in $\txns{H4}$; (3) $\htwts{m}{H4} > \htwts{l}{H4}$ (4) $T_m$ is committed in $H5$. 

Combining the above derived conditions (1), (2), (3) with \lemref{ti|tltl-comt} we get that in $H4$, 

\begin{equation}
\label{eq:ml-tits}
\htits{m}{H4} \leq \htits{l}{H4} + 2*L
\end{equation}

\eqnref{ml-tits} implies that $T_m$ is in $T_l$'s \affset. Here, we have that $T_l$ is an \inc of $T_i$ and we are given that $\htcts{i}{H1} \geq \htits{i}{H1} + 2*L$. Thus from \lemref{aff-tkinc-h1}, we get that there exists an \inc of $T_m$, $T_n$ in $H1$.

Combining \eqnref{ml-tits} with the observations (a) $T_n, T_m$ are \inc{s}; (b) $T_l, T_i$ are \inc{s}; (c) $T_i, T_n$ are in $\txns{H1}$, we get that $\htits{n}{H1} \leq \htits{i}{H1} + 2*L$. This implies that $T_n$ is in $\haffset{i}{H1}$. Since $T_n$ is not committed in $H1$ (otherwise, it is not possible for $T_m$ to be an \inc of $T_n$), we get that $T_n$ is in $\hcds{i}{H1}$. Hence, we get that $\htits{m}{H4} = \htits{n}{H1}$ is in $\hdep{i}{H1}$.  

From \eqnref{hdep-assn}, we have that $\hdep{i}{H1}$ is $k+1$. From \lemref{depits-fut}, we get that $\hdep{i}{H4}$ is a subset of $\hdep{i}{H1}$. Further, we have that transaction $T_m$ has committed. Thus $\htits{m}{H4}$ which was in $\hdep{i}{H1}$ is no longer in $\hdep{i}{H4}$. This implies that $\hdep{i}{H4}$ is a strict subset of $\hdep{i}{H1}$ and hence $|\hdep{i}{H4}| \leq k$. 

\noindent Since $T_i$ and $T_l$ are \inc{s}, we get that $\hdep{i}{H4} = \hdep{l}{H1}$. Thus, we get that 

\begin{equation}
\label{eqn:hdep-ilh4}
|\hdep{i}{H4}| \leq k \implies |\hdep{l}{H4}| \leq k
\end{equation}

\noindent Further, we have that $T_l$ is a later \inc of $T_i$. So, we get that

\begin{equation}
\label{eqn:cts-its}
\htcts{l}{H4} > \htcts{i}{H4} \xrightarrow{given} \htcts{l}{H4} > \htits{i}{H4} + 2*L \xrightarrow{\htits{i}{H4} = \htits{l}{H4}} \htcts{l}{H4} > \htits{l}{H4} + 2*L
\end{equation}

We also have that $T_l$ is live in $H4$. Combining this with Equations \ref{eqn:hdep-ilh4}, \ref{eqn:cts-its} and given the induction hypothesis that $P(k)$ is true, we get that there exists a history extension of $H4$, $H6$ in which an \inc of $T_l$ (also $T_i$), $T_p$ is either committed or \cdsen. This proves the lemma.
\end{proof}

\begin{lemma}
\label{lem:cds-fin}
Consider a transaction $T_i$ in a history $H1$. If $T_i$ is \cdsen in $H1$ then there is an extension of $H1$, $H2$ in which an \inc $T_i$, $T_j$ (which could be same as $T_i$), is either committed or \finen. Formally, $\langle H1, T_i: (T_i \in \live{H}) \land (\cdsenb{i}{H1}) \implies (\exists H2, T_j: (H1 \sqsubset H2) \land (T_j \in \incs{i}{H2}) \land ((T_j \in \comm{H2}) \lor (\finenb{j}{H2})) \rangle$. 
\end{lemma}

\begin{proof}
In $H1$, suppose $\haffwts{i}{H1}$ is $\alpha$. From \lemref{wts-great}, we get that there is a extension of $H1$, $H2$ with a transaction $T_j$ which is an \inc of $T_i$. Here there are two cases: (1) Either $T_j$ is committed in $H2$. This trivially proves the lemma; (2) Otherwise, $\twts{j}$ is greater than $\alpha$. 

\noindent In the second case, we get that 

\begin{equation}
\label{eq:ext}
\begin{split}
(T_i \in \live{H1}) \land (T_j \in \live{H2}) \land (\cdsenb{i}{H}) \land (T_j \in \incs{i}{H2}) \land \\
(\htwts{i}{H1} < \htwts{j}{H2})
\end{split}
\end{equation}

\noindent Combining the above result with \lemref{cts-wts}, we get that $\htcts{i}{H1} < \htcts{j}{H2}$. Thus the modified equation is 

\begin{equation}
\label{eq:new-ext}
\begin{split}
(T_i \in \live{H1}) \land (T_j \in \live{H2}) \land (\cdsenb{i}{H1}) \land (T_j \in \incs{i}{H2}) \land \\
(\htcts{i}{H1} < \htcts{j}{H2})
\end{split}
\end{equation}

\noindent Next combining \eqnref{new-ext} with \lemref{aff-same}, we get that 
\begin{equation}
\label{eq:affs-ij}
\haffset{i}{H1} = \haffset{j}{H2}
\end{equation}

\noindent Similarly, combining \eqnref{new-ext} with \lemref{cds-fut} we get that $T_j$ is \cdsen in $H2$ as well. Formally, 
\begin{equation}
\label{eq:th-cdsen}
\cdsenb{j}{H2}
\end{equation}

Now combining \eqnref{affs-ij} with \lemref{affwts-same}, we get that 
\begin{equation}
\label{eq:affwts-same}
\haffwts{i}{H1} = \haffwts{j}{H2}
\end{equation}

From our initial assumption we have that $\haffwts{i}{H1}$ is $\alpha$. From \eqnref{affwts-same}, we get that $\haffwts{j}{H2} = \alpha$. Further, we had earlier also seen that $\htwts{j}{H2}$ is greater than $\alpha$. Hence, we have that $\htwts{j}{H2} > \haffwts{j}{H2}$. 

\noindent Combining the above result with \eqnref{th-cdsen}, $\cdsenb{j}{H2}$, we get that $T_j$ is \finen, i.e., $\finenb{j}{H2}$. 
\end{proof}

\noindent Next, we show that every live transaction eventually become \itsen. 

\begin{lemma}
\label{lem:live-its}
Consider a history $H1$ with $T_i$ be a transaction in $\live{H1}$. Then there is an extension of $H1$, $H2$ in which an \inc of $T_i$, $T_j$ (which could be same as $T_i$) is either committed or is \itsen. Formally, $\langle H1, T_i: (T_i\in \live{H}) \implies (\exists T_j, H2: (H1 \sqsubset H2) \land (T_j \in \incs{i}{H2}) \land (T_j \in \comm{H2}) \lor (\itsenb{i}{H})) \rangle$. 
\end{lemma}

\begin{proof}
\noindent We prove this lemma by inducting on \its. 

\vspace{1mm}
\noindent
\textbf{Base Case - $\tits{i} = 1$:} In this case, $T_i$ is the first transaction to be created. There are no transactions with smaller \its. Thus $T_i$ is trivially \itsen. 

\vspace{1mm}
\noindent
\textbf{Induction Case:} Here we assume that for any transaction $\tits{i} \leq k$ the lemma is true. 
\end{proof}

Combining these lemmas gives us the result that for every live transaction $T_i$ there is an incarnation $T_j$ (which could be the same as $T_i$) that will commit. This implies that every \aptr eventually commits. The follow lemma captures this notion.

\begin{theorem}
\label{thm:hwtm-com}
Consider a history $H1$ with $T_i$ be a transaction in $\live{H1}$. Then there is an extension of $H1$, $H2$ in which an \inc of $T_i$, $T_j$ is committed. Formally, $\langle H1, T_i: (T_i\in \live{H}) \implies (\exists T_j, H2: (H1 \sqsubset H2) \land (T_j \in \incs{i}{H2}) \land (T_j \in \comm{H2})) \rangle$. 
\end{theorem}

\begin{proof}
Here we show the states that a transaction $T_i$ (or one of it its \inc{s}) undergoes before it commits. In all these transitions, it is possible that an \inc of $T_i$ can commit. But to show the worst case, we assume that no \inc of $T_i$ commits. Continuing with this argument, we show that finally an \inc of $T_i$ commits. 

Consider a live transaction $T_i$ in $H1$. Then from \lemref{live-its}, we get that there is a history $H2$, which is an extension of $H1$, in which $T_j$ an \inc of $T_i$ is either committed or \itsen. If $T_j$ is \itsen in $H2$, then from \lemref{its-cds}, we get that $T_k$, an \inc of $T_j$, will be \cdsen in a extension of $H2$, $H3$ (assuming that $T_k$ is not committed in $H3$). 

From \lemref{cds-fin}, we get that there is an extension of $H3$, $H4$ in which an \inc of $T_k$, $T_l$ will be \finen assuming that it is not committed in $H4$. Finally, from \lemref{enbd-ct}, we get that there is an extension of $H4$ in which $T_m$, an \inc of $T_l$, will be committed. This proves our theorem.
\end{proof}

\section{Discussion and Conclusion}
\label{sec:conc}


In this paper, we propose a $K$ version \emph{\stf} STM system, \emph{\ksftm}. The algorithm ensures that if an \emph{aborted} transaction is retried successively, then it will eventually commit. The algorithm maintains $K$ versions where $K$ can range from between one to infinity. For correctness, we show \ksftm{} satisfies strict-serializability \cite{Papad:1979:JACM} and local opacity \cite{KuzSat:NI:ICDCN:2014, KuzSat:NI:TCS:2016}. To the best of our knowledge, this is the first work to explore \emph{\stfdm}  with \mvstm{s}. 

Our experiments show that \ksftm performs better than single-version STMs (ESTM, Norec STM) under high contention and also single-version \emph{\stf}  STM \svsftm developed based on the principle of priority. On the other hand, its performance is comparable or slightly worse than multi-version STM, \pkto (around 2\%). This is the cost of the overhead required to achieve \emph{\stfdm}  which we believe is a marginal price. 

In this document, we have not considered a transactional solution based on two-phase locking (2PL) and its multi-version variants \cite{WeiVoss:2002:Morg}. With the carefully designed 2PL solution, one can ensure that none of the transactions abort \cite{WeiVoss:2002:Morg}. But this will require advance knowledge of the code of the transactions which may not always be available with the STM library. Without such knowledge, it is possible that a 2PL solution can deadlock and cause further aborts which will, raise the issue of \emph{\stfdm}  again. 

\noindent
Since we have considered \stsble as one of the \emph{correctness-criteria}, this algorithm can be extended to databases as well. In fact, to the best of our knowledge, there has been no prior work on \emph{\stfdm} in the context of database concurrency control.



\bibliography{citations}

\end{document}